\documentclass[a4paper,10pt]{amsart}

\usepackage{graphicx}
\usepackage{tikz}
\usepackage{amsmath,esint,ifthen,commath}
\usepackage{amssymb,amsthm,enumitem,color,stmaryrd}
\usepackage{hyperref}
\usepackage{comment}

\newtheorem{thm}{Theorem}
\newtheorem{prop}[thm]{Proposition}
\newtheorem{lem}[thm]{Lemma}

\theoremstyle{definition}
\newtheorem{Def}[thm]{Definition}

\theoremstyle{remark}
\newtheorem{rem}[thm]{Remark}

\DeclareMathOperator{\Ad}{Ad}

\DeclareMathOperator{\Pf}{Pf}
\DeclareMathOperator{\RYG}{RYG}

\title{Dimers on Rail Yard Graphs}
\author[C. Boutillier]{C\'edric Boutillier}
\address{Laboratoire de Probabilit\'es et Mod\`eles Al\'eatoires, UMR 7599, Universit\'e Pierre et Marie Curie, 4 place Jussieu, F-75005 Paris}
\email{cedric.boutillier@upmc.fr}
\author[J. Bouttier]{J\'er\'emie Bouttier}
\address{Institut de Physique Th\'eorique, Universit\'e Paris-Saclay, CEA, CNRS, F-91191 Gif-sur-Yvette and D\'epartement de Math\'ematiques et Applications, \'Ecole normale sup\'erieure, 45 rue d'Ulm, F-75231 Paris Cedex 05}
\email{jeremie.bouttier@cea.fr}
\author[G. Chapuy]{Guillaume Chapuy}
\address{LIAFA, CNRS et Universit\'e Paris Diderot, Case 7014, F-75205 Paris Cedex 13}
\email{guillaume.chapuy@liafa.univ-paris-diderot.fr}
\author[S. Corteel]{Sylvie Corteel}
\address{LIAFA, CNRS et Universit\'e Paris Diderot, Case 7014, F-75205 Paris Cedex 13}
\email{corteel@liafa.univ-paris-diderot.fr}
\author[S. Ramassamy]{Sanjay Ramassamy}
\address{Mathematics Departement, Brown University, Box 1917, 151 Thayer Street, Providence, RI 02912}
\email{sanjay\_ramassamy@brown.edu}

\thanks{We acknowledge financial support from the Ville de Paris via
  Projet \'Emergences ``Combinatoire \`a Paris'' (JB, GC, SC), from
  the Monahan Foundation via a Monahan Science Fellowship (SR) and
  from the Agence Nationale de la Recherche via the grants
  ANR-08-JCJC-0011 ``IComb'' (SC), ANR-10-BLAN-0123 ``MAC2'' (CB),
  ANR-12-JS02-0001 ``Cartaplus'' (JB, GC), ANR-14-CE25-0014 ``GRAAL''
  (JB). JB acknowledges the hospitality of LIAFA. SR acknowledges the
  hospitality and financial support of the Erwin Sch\"odinger
  Institute in Vienna, where part of this work was done.}

\date{\today}

\begin{document}

\begin{abstract}
  We introduce a general model of dimer coverings of certain plane
  bipartite graphs, which we call rail yard graphs (RYG). The transfer
  matrices used to compute the partition function are shown to be
  isomorphic to certain operators arising in the so-called
  boson-fermion correspondence. This allows to reformulate the RYG
  dimer model as a Schur process, i.e.\ as a random sequence of
  integer partitions subject to some interlacing conditions.

  Beyond the computation of the partition function, we provide an
  explicit expression for all correlation functions or, equivalently,
  for the inverse Kasteleyn matrix of the RYG dimer model. This
  expression, which is amenable to asymptotic analysis, follows from
  an exact combinatorial description of the operators localizing
  dimers in the transfer-matrix formalism, and then a suitable
  application of Wick's theorem.

  Plane partitions, domino tilings of the Aztec diamond, pyramid
  partitions, and steep tilings arise as particular cases of the RYG
  dimer model. For the Aztec diamond, we provide new derivations of
  the edge-probability generating function, of the biased creation
  rate, of the inverse Kasteleyn matrix and of the arctic circle
  theorem.
\end{abstract}

\maketitle

\section{Introduction}

The two-dimensional dimer model is arguably the most studied exactly
solvable model in statistical mechanics (note that it encompasses, in
a sense, the equally well-known two-dimensional Ising model), see for
instance \cite[Chapter 5]{MR0253689} for a review of the seminal works
of Kasteleyn, Temperley and Fisher, and the introduction of
\cite{Allegra14} for a nice survey of the more recent
literature. Dimer configurations are also known as perfect matchings
in combinatorics and theoretical computer science. Actually, perhaps
the oldest exact solution of a 2D dimer model is MacMahon's
enumeration of plane partitions \cite{MR2417935}, as these were later
identified with lozenge tilings, or alternatively dimer configurations
on the hexagonal lattice.

Kasteleyn's method allows to reduce the problem of computing the
partition function (and the correlation functions) of the dimer model
on any finite weighted planar graph (assuming that the dimers interact
only through their hard-core repulsion) to the evaluation of a
determinant (or Pfaffian) whose size is linear in the number of
vertices of the graph. Under the usual assumption that the graph is
periodic in two directions, one can then evaluate this determinant and
take the thermodynamic limit to obtain the free energy, study phase
transitions, etc. In this paper, we consider a dimer model on a new
family of graphs, called rail yard graphs, which are periodic in one
direction but not in the other.

One of our motivations is that the rail yard graph dimer model
encompasses both the plane partitions mentioned above and another
celebrated model, namely domino tilings of the Aztec diamond
\cite{EKLP1992,EKLP1992b} (corresponding to, roughly speaking, dimer
configurations on the portion of the square lattice fitting into a
large square tilted by $45^\circ$). What relates these two models is
that they can be seen as Schur processes \cite{OR1}, that is to say
random sequences of integer partitions whose transition probabilities
are given by Schur functions. If the relation between plane partitions
and Schur processes was explicited by Okounkov and Reshetikhin, the
case of the Aztec diamond appears implicitly in \cite{MR1900323} and
has, to the best of our knowledge, remained in such implicit form
until \cite{BCC}, of which this paper is a continuation (see
below). The interest of making the connection between dimer models and
Schur processes explicit is that it allows to use an operator
formalism coming from the boson-fermion correspondence (see the
references given at the beginning of Section~\ref{sec:bos}) which is
both powerful and intuitive, as the operators are nothing but transfer
matrices or observables satisfying some particularly simple
commutation relations. Furthermore it allows us to say that the RYG
dimer model forms another situation, besides the 2D Ising model
\cite{Dubedat11}, where ``bosonization'' works at an exact discrete
level.  The rail yard graph dimer model corresponds essentially to the
most general Schur process with nonnegative transition probabilities.

Before describing our work in more detail, let us further discuss some
history and background behind it. Bender and Knuth \cite{MR0299574}
made the link between plane partitions and the
Robinson-Schensted-Knuth correspondence, see also
\cite[Chapter~7]{Stanley}.
Okounkov~\cite{Okounkov:wedge, MR2059364} used the boson-fermion
correspondence to define the so-called Schur measure over integer
partitions, and study its correlation functions. The Schur process
\cite{OR1,OR2} is a time-dependent version of this measure, that can
also be viewed as a system of particles with certain dynamics. It
contains as a special case a generalization of plane partitions,
namely plane partitions with an evolving ``back wall''. Numerous
papers followed on this subject~\cite{B2007,
  MR3148098,Borodin:dynamics,MR2889660,BBBCCV14,BorodinFerrari15} and on its extension to the
Hall-Littlewood and Macdonald
cases~\cite{MR2349310,MR2520768,MR2471939,MR2721059,CSV,MR3152785}.

In \cite{BCC}, three authors of the present paper introduced a general
class of domino tilings called steep tilings, encompassing both
tilings of the Aztec diamond and the so-called pyramid partitions
\cite{Kenyon2005, Young:pyramid}. It was shown in~\cite{BCC} that
steep tilings also correspond to Schur processes and, using the vertex
operator formalism, their partition functions (of the ``hook formula''
type) were computed for a variety of boundary conditions.
Since both (generalized) plane partitions and steep tilings are
special instances of the Schur process, it is then natural to ask if
there is a more general model of tilings or dimer coverings that would
reformulate the Schur process in full generality, at least when the
number of underlying parameters is finite. Such a model was sketched
in~\cite[Section~7]{BCC}, that can be viewed as a preliminary attempt
at what we reach in the present paper. 

The rail yard graphs (RYG) that we introduce here are infinite
bipartite plane graphs, obtained by the ``concatenation'' of
column-shaped elementary graphs, and come with a family of admissible
dimer coverings. The RYG dimer model is then a probability measure
over such coverings. The elementary graphs can be of four types that
correspond to the four possible types of ``atomic'' transitions in the
Schur process. For the special families of RYG that correspond to the
special families of Schur processes considered in~\cite{OR1, BCC}, we
recover generalized plane partitions and steep tilings,
respectively. As we hope will be apparent in this paper, RYG provide a
nice and natural formulation of the Schur process in terms of dimers,
in a well-adapted system of coordinates, much simpler than the one
from~\cite[Section 7]{BCC}.  Having shown the correspondence between
rail yard graphs and Schur process, we can then apply the same
classical tools as in~\cite{OR1} to get explicitly the partition
functions in a nice (hook-type) product form. Even more, we can
interpret these partition functions in terms of a combinatorial
parameter related to the flip operation on coverings.

Beyond the partition function, we compute all the dimer correlation
functions, which requires the introduction of suitably defined
observables (or constrained transfer matrices) that enable us to
localize, in the algebraic setting, a given set of dimers. To prevent
any confusion, let us note that the particle correlations computed
in~\cite{OR1} for the Schur process, when translated in terms of RYG,
give only a special case of this result. Indeed, as we will see, there
are three kinds of dimers in a RYG, and particles correspond to one of
the three kinds (so in our setting the correlations results
of~\cite{OR1} only describe correlations between dimers of the first
kind). Once the observables are constructed, we use classical
fermionic tools such as Wick's formula to evaluate the correlation
functions in an explicit determinantal form. We also make the
connection with the general Kasteleyn theory: it is a general fact
that correlations between dimers on plane bipartite graphs have a
determinantal form, underlaid by an inverse of the so-called Kasteleyn
matrix of the model. For RYG, we show that the determinantal form we
obtain by our approach indeed gives an inverse Kasteleyn matrix, as
was remarked in \cite{OR2} for the case of skew plane partitions, see
also \cite[Section~5]{MR3148098}. Our approach generalizes both this
case and that of the Aztec diamond (for the so-called $q^{\text{vol}}$
weighting), treated previously in~\cite{MR3197656} by a very tricky
and somehow mysterious calculation. As further applications concerning
the Aztec diamond, we rederive the so-called edge-probability
generating function and biased creation rate, and the arctic circle
theorem using the general saddle-point techniques of \cite{OR1}.

We now present the structure of the paper. Section~\ref{sec:main} is
devoted to the basic definitions (rail yard graphs in
Subsection~\ref{sec:ryg}, their dimer coverings in
Subsection~\ref{sec:admiss}, flips in Subsection~\ref{sec:flipdef})
and to the statement of our main results, namely the expression for
the partition function (Subsection~\ref{sec:enum}) and for the dimer
correlation functions (Subsection~\ref{sec:corr}).
Section~\ref{sec:bos} introduces bosonic operators
(Subsection~\ref{sec:bosrem}) that act as transfer matrices in the RYG
dimer model (Subsection~\ref{sec:bosmat}), allowing to compute
efficiently the partition function
(Subsection~\ref{sec:enumproof}). Section~\ref{sec:ferm} considers
fermionic operators (Subsection~\ref{sec:fermrem}) that play the role
of observables in the RYG dimer model
(Subsection~\ref{sec:dimerop}). Rewriting the correlation functions in
the ``Heisenberg picture'' (Subsection~\ref{sec:schrheis}), we derive
their expression in the form of a determinant
(Subsection~\ref{sec:det_corr_proof}), before making the connection
with Kasteleyn's theory
(Subsection~\ref{sec:fermkast}). Section~\ref{sec:particular}
discusses the previously known cases: plane partitions and lozenge
tilings (Subsection~\ref{sec:lozenge}) and steep domino tilings
(Subsection~\ref{sec:steep}). In Section~\ref{sec:aztecCorr} we address the specific case of the Aztec diamond, for which 
we provide new derivations of the edge-probability generating function and biased creation rate (Subsection~\ref{subsec:edgep}), of the inverse Kasteleyn matrix (Subsection~\ref{subsec:aztecKasteleyn}) and of the arctic
circle theorem (Subsection~\ref{subsec:aztecArctic}). Concluding remarks are gathered in
Section~\ref{sec:conc}. Some auxiliary material is given in the
appendix: a combinatorial proof of the bosonic-fermionic commutation
relations (Appendix~\ref{sec:fermcom}) and a rederivation of Wick's
formula (Appendix~\ref{sec:wick}).

After the completion of this work, Alexei Borodin and Senya Shlosman
informed us (by private communication) that graphs similar to those
described in Section~\ref{sec:ryg} were suggested to them by Richard
Kenyon, providing a dimer interpretation of the $\alpha\beta$-paths
considered in \cite{BorodinShlosman}.

\newpage
\section{Basic definitions and main results}
\label{sec:main}

\subsection{Rail yard graphs}
\label{sec:ryg}

We start by defining the underlying graph of our dimer model. We fix
two integers $\ell,r$ such that $\ell \leq r$, and denote by
$[\ell..r]$ the set of integers between $\ell$ and $r$. We then
consider two binary sequences indexed by the elements of $[\ell..r]$:
\begin{itemize}
\item the \emph{LR sequence} $\underline{a}=(a_\ell,a_{\ell+1},\ldots,a_r) \in
  \{L,R\}^{[\ell..r]}$,
\item the \emph{sign sequence} $\underline{b}=(b_\ell,b_{\ell+1},\ldots,b_r) \in
  \{+,-\}^{[\ell..r]}$.
\end{itemize}
The \emph{rail yard graph} associated with the integers $\ell$ and
$r$, the LR sequence $\underline{a}$ and the sign sequence
$\underline{b}$, and denoted by
$\RYG(\ell,r,\underline{a},\underline{b})$, is the bipartite
plane graph defined as follows. Its vertex set is $[2\ell-1..2r+1]
\times (\mathbb{Z}+1/2)$, and we say that a vertex is \emph{even}
(resp.\ \emph{odd}) if its abscissa is an even (resp.\ odd)
integer. Each even vertex $(2m,y)$, $m \in [\ell..r]$, is then
incident to three edges: two \emph{horizontal} edges connecting it to
the odd vertices $(2m-1,y)$ and $(2m+1,y)$, and one \emph{diagonal}
edge connecting it to
\begin{itemize}
\item the odd vertex $(2m-1,y+1)$ if $a_m=L$ and $b_m=+$,
\item the odd vertex $(2m-1,y-1)$ if $a_m=L$ and $b_m=-$,
\item the odd vertex $(2m+1,y+1)$ if $a_m=R$ and $b_m=+$,
\item the odd vertex $(2m+1,y-1)$ if $a_m=R$ and $b_m=-$.
\end{itemize}
Hopefully, this explains our motivations for using the symbols $L,R$
and $+,-$. Drawing the edges straight, the graph is indeed bipartite
and plane by construction. For $e$ an edge, we write
$e=(\alpha,\beta)$ to mean that $\alpha$ is the even endpoint of $e$,
and $\beta$ its odd endpoint. For $v$ a vertex, we will denote by
$v^{\mathrm{x}} \in \mathbb{Z}$ its abscissa, and by
$v^{\mathrm{y}} \in \mathbb{Z}+1/2$ its ordinate.

\begin{figure}[htpb]
  \centering
  \includegraphics[width=\textwidth]{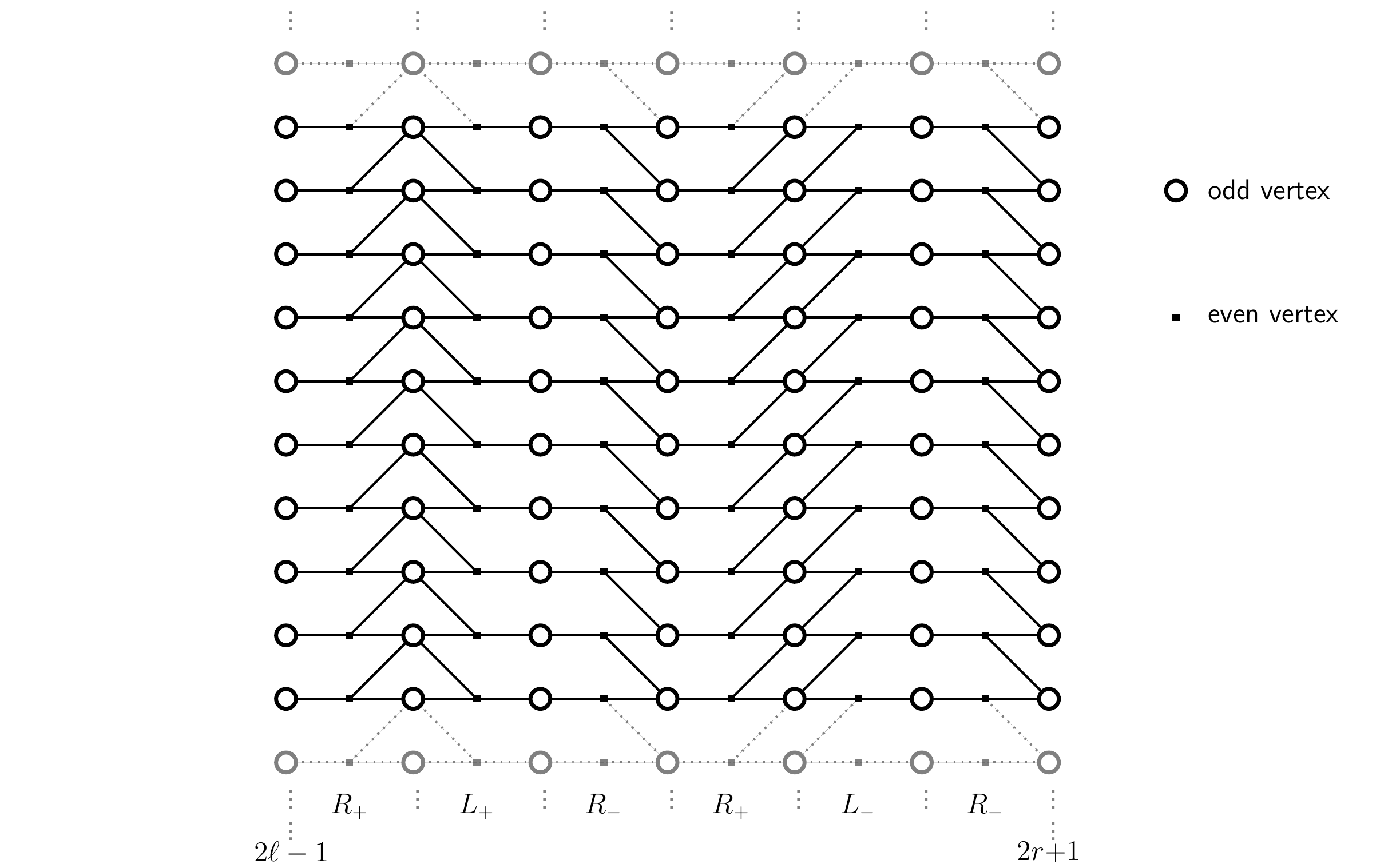}
  \caption{The rail yard graph associated with
the LR sequence $RLRRLR$ and the sign sequence $++-+--$ (with
$\ell=1$, $r=6$). It is infinite and periodic in the
    vertical direction, but finite in the horizontal
    direction.}
  \label{fig:RYGgraph}
\end{figure}

Figure~\ref{fig:RYGgraph} displays the rail yard graph associated with
the LR sequence $RLRRLL$ and the sign sequence $++-+--$ (with
$\ell=1$, $r=6$). Observe that a rail yard graph is infinite and
$1$-periodic in the vertical direction. When $\ell=r$, the LR and sign
sequences both consist of a single element, and the corresponding rail
yard graph, which is said \emph{elementary}, is of one of four
possible types, see Figure~\ref{fig:elemRYG}. Given two rail yard
graphs $\RYG(\ell,r,\underline{a},\underline{b})$ and
$\RYG(\ell',r',\underline{a}',\underline{b}')$ such that
$\ell'=r+1$, we define their \emph{concatenation} by taking the union
of their vertex and edge sets. It is nothing but the rail yard graph
$\RYG(\ell,r',\underline{a}\underline{a}',\underline{b}\underline{b}')$
where $\underline{a}\underline{a}'$ and $\underline{b}\underline{b}'$
denote the concatenations of the LR and sign sequences. Clearly, a
general rail yard graph is obtained by concatenating elementary ones.

\begin{figure}[htpb]
  \centering
  \includegraphics[width=0.67\textwidth]{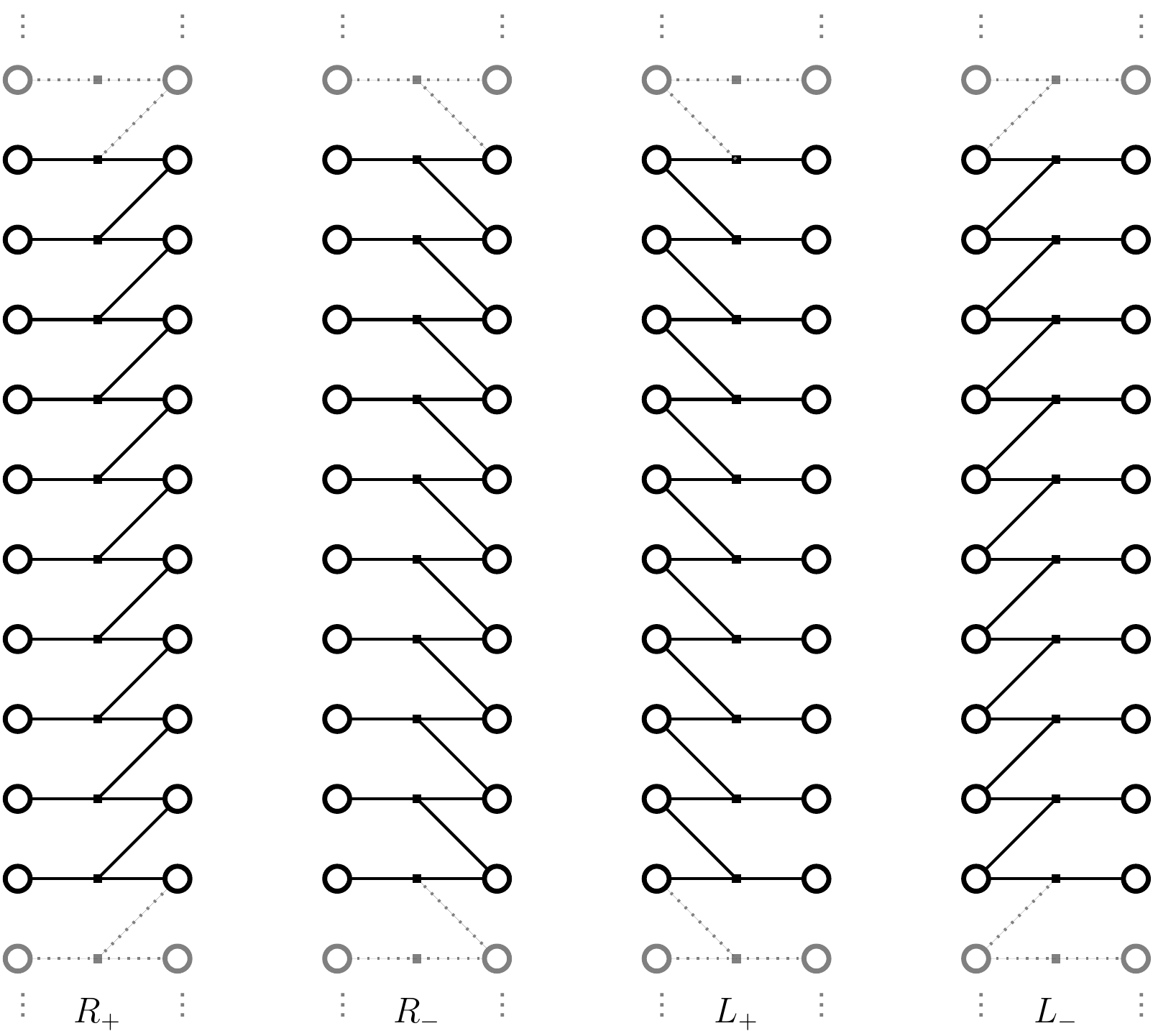}
  \caption{The four elementary rail yard graphs $L+$, $L-$, $R+$, and $R-$ (up to horizontal translation) }
  \label{fig:elemRYG}
\end{figure}

The \emph{left boundary} (resp.\ \emph{right boundary}) of a rail
yard graph consists of all odd vertices with abscissa $2\ell-1$
(resp.\ $2r+1$). Vertices which do not belong to the boundaries are
said \emph{inner}. When drawn in the plane, the graph delimits some
faces, and the bounded ones are called \emph{inner faces}. Note that
inner faces may be incident to $4$, $6$ or $8$ edges. Finally, observe
that our definition works equally well if we take $\ell=-\infty$
and/or $r=+\infty$, thus considering infinite LR and sign sequences. In
that case, the rail yard graph ``fills'' either the whole plane
or a half-plane, boundaries being sent to infinity.

\subsection{Admissible and pure dimer coverings}
\label{sec:admiss}

We now turn to the characterization of the configurations of our dimer
model. Given a rail yard graph
$\RYG(\ell,r,\underline{a},\underline{b})$ with $\ell,r$
finite, an \emph{admissible dimer covering} is a partial matching of
this graph such that:
\begin{itemize}
\item each inner vertex is covered (i.e.\ matched),
\item there exists an integer $N \geq 0$ such that: any left boundary
  vertex $(2\ell-1,y)$ is covered for $y>N$ and uncovered for $y<-N$,
  any right boundary vertex $(2r+1,y)$ is covered for $y<-N$ and
  uncovered for $y>N$,
\item only a finite number of diagonal edges are covered.
\end{itemize}
A \emph{pure dimer covering} is an admissible dimer covering for which
the second property above holds for $N=0$: in other words the
uncovered vertices are precisely the left boundary vertices with
negative ordinate and the right boundary vertices with positive
ordinate (see Figure~\ref{fig:covering}). The \emph{fundamental dimer covering} is the pure dimer
covering where no diagonal edge is covered (it is not difficult to
check its existence and uniqueness e.g.\ by induction on $r-\ell$).
Observe that any admissible dimer covering coincides with the
fundamental dimer covering outside a finite region. An
\emph{elementary dimer covering} is an admissible dimer covering of an
elementary rail yard graph (see Figure~\ref{fig:elemCovering}).

\begin{figure}[htpb]
  \centering
  \includegraphics[width=0.8\textwidth]{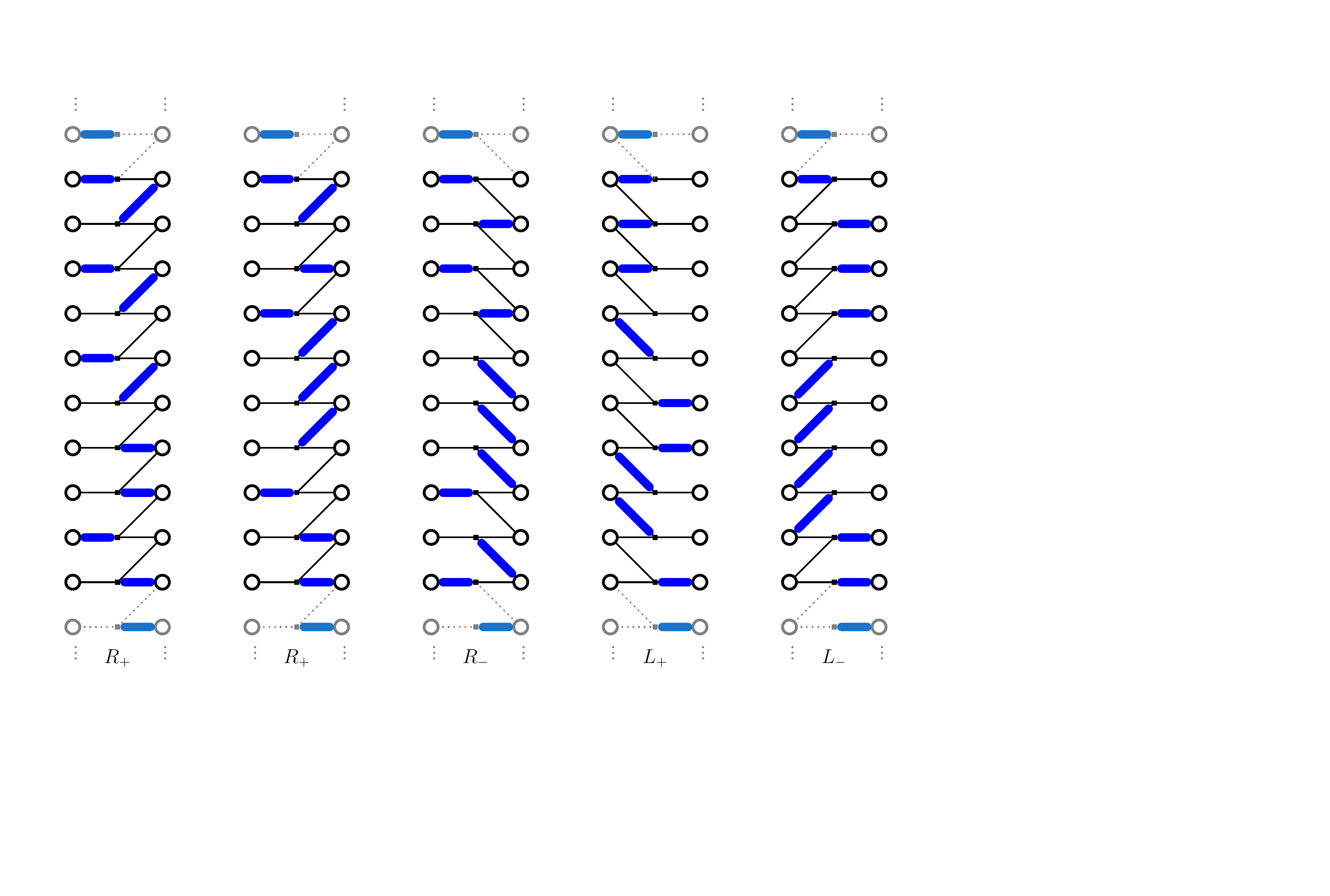}
  \caption{Some elementary dimer coverings. The underlying elementary rail yard graph has type $R+$ in the first two cases, and $R-$, $L+$, and $L-$, in the three others, from left to right. Outside of the displayed region, the configuration is identical towards the top (resp. bottom) on each horizontal level to what it is on the topmost (resp. bottommost) displayed level.}
  \label{fig:elemCovering}
\end{figure}

\begin{figure}[htpb]
  \centering
  \includegraphics[width=0.9\textwidth]{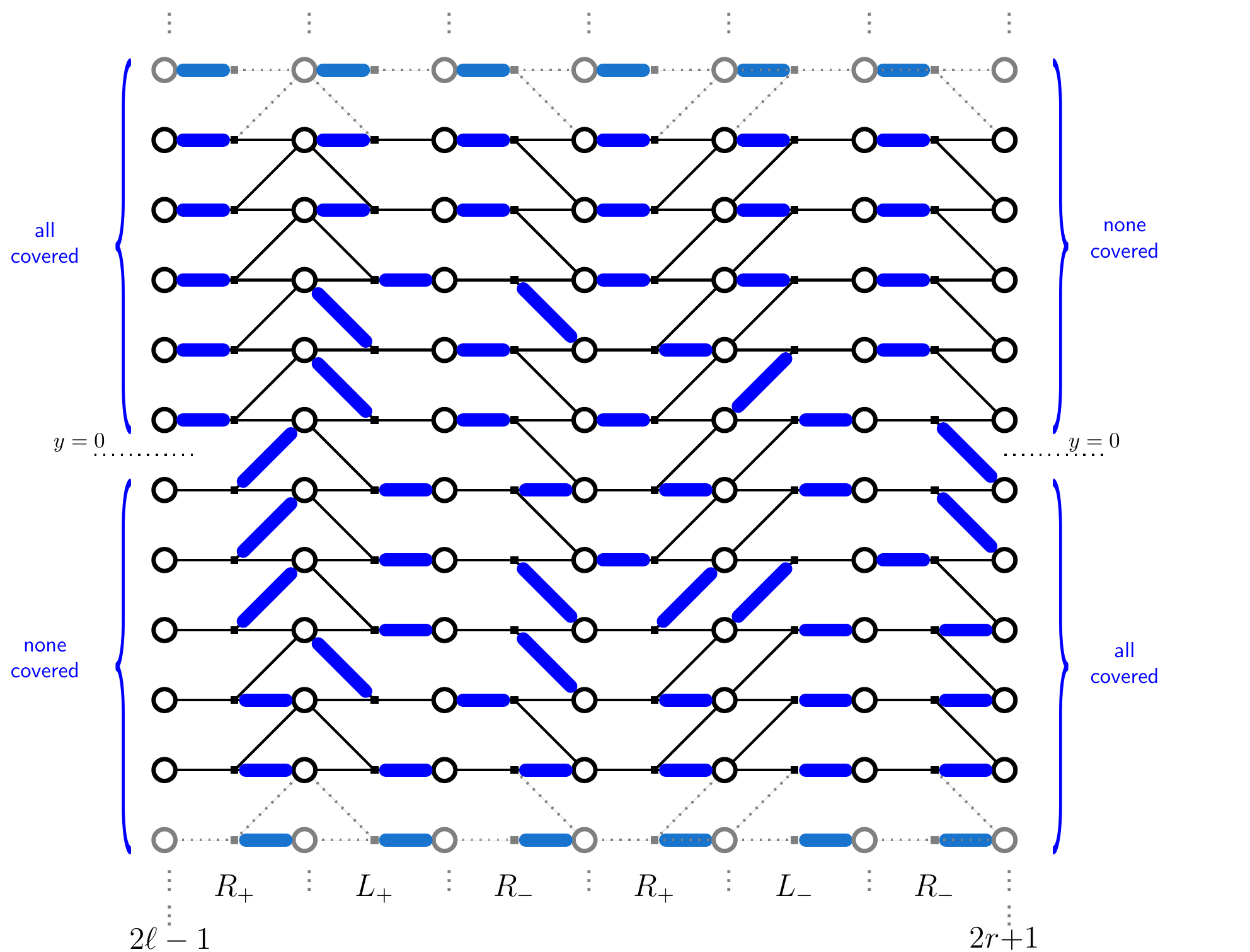}
  \caption{A pure dimer covering of the rail yard graph of Figure~\ref{fig:RYGgraph}.}
  \label{fig:covering}
\end{figure}

Similarly to rail yard graphs, admissible dimer coverings behave
nicely with respect to concatenation. More precisely, consider two
rail yard graphs $G=\RYG(\ell,r,\underline{a},\underline{b})$
and $G'=\RYG(\ell',r',\underline{a}',\underline{b}')$ which
are concatenable (i.e.\ $\ell'=r+1$) and let $GG'$ be their
concatenation. Let $C$ and $C'$ be admissible dimer coverings of
respectively $G$ and $G'$: we say that $C$ and $C'$ are
\emph{compatible} if, for each $y \in \mathbb{Z}+1/2$, the vertex
$(2r+1,y)=(2\ell'-1,y)$ is covered in $C$ if and only if it is not
covered in $C'$. In that case, by taking the union of $C$ and $C'$, we
obtain an admissible dimer covering of $GG'$, which we denote by
$CC'$. Conversely, any admissible dimer covering can be decomposed as
the concatenation of elementary dimer coverings which are sequentially
compatible.

It is also interesting to consider the limiting cases $\ell=-\infty$
and/or $r=+\infty$, which requires a slight adaptation of our
definitions. An admissible (resp.\ a pure, resp.\ the fundamental)
dimer covering is then a matching such that each inner vertex is
covered, and such that there exists finite integers $\ell',r'$ such
that :
\begin{itemize}
\item inside the strip $[2\ell'-1,2r'+1] \times \mathbb{R}$, we see an
  admissible (resp.\ a pure, resp.\ the fundamental) dimer covering
  in the previous sense,
\item outside this strip, all covered edges are horizontal.
\end{itemize}
(Note that this definition works in all situations: it coincides with
the previous one when $\ell,r$ are both finite.)

Our motivation for considering pure dimer coverings of rail yard
graphs is that we recover several well-known dimer models as
specializations. For instance, taking $r-\ell=2n$, a LR sequence of
the form $LRLRLR\cdots$ and a sign sequence of the form
$+-+-+-\cdots$, the corresponding pure dimer configurations are in
bijection with domino tilings of the Aztec diamond of size $n$. We
also recover plane partitions and so-called pyramid partitions, which
requires taking $\ell=-\infty$ and $r=+\infty$: plane partitions are
obtained by taking a constant LR sequence and a sign sequence of the
form $\cdots++++----\cdots$, while pyramid partitions are obtained by
taking an alternating LR sequence ($\cdots LRLRLRLR \cdots$) and the
same sign sequence. We will discuss these specializations in greater
detail in Section~\ref{sec:particular}.

\subsection{Flips}
\label{sec:flipdef}

We now define a local transformation on admissible coverings called the \emph{flip}. Let $G$ be a rail yard graph, $C$ be an admissible covering of $G$, and let $f$ be an inner face
of $G$. If exactly half of the edges bordering $f$ belong to $C$, then removing
these edges from $C$ and replacing them by the other edges bordering $f$ gives
another admissible covering $C'$ of $G$. The operation that replaces $C$ by $C'$ is called the \emph{flip} of the face $f$, see Figure~\ref{fig:flips}.
\begin{figure}[htpb]
  \centering
\begin{center}
  \includegraphics[width=0.9\textwidth]{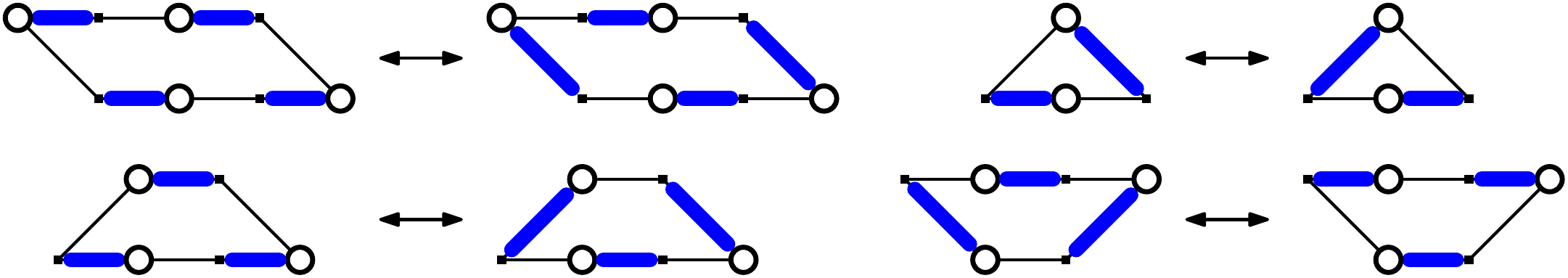}
\end{center}
  \caption{Examples of flips. All these flips are positive when the
transformation is made from left to right.}
  \label{fig:flips}
\end{figure}

We say that the flip of an inner face $f$ is \emph{positive} if after performing the
flip, the edges of $f$
that belong to the covering are oriented from odd to even
vertices in counterclockwise direction around $f$. The flip is \emph{negative}
otherwise.
 For example, each flip displayed on
Figure~\ref{fig:flips} is positive when performed from left to right.       
It follows from \cite[Theorem~2]{Propp} that the positive flip relation
endows the set of all pure coverings of a given rail yard graph with a distributive lattice structure. In
particular, each rail yard graph has a unique \emph{minimal} pure covering
from which all other ones can be reached using positive flips only. The minimal
covering is the only pure covering on which no negative flip is possible. Using
this criterion one easily checks that the minimal covering coincides with the
fundamental covering defined above.
The \emph{flip distance} between two coverings is the minimal number of flips
needed to go from one to the other. When one of the two coverings is the
fundamental one, the flip distance is realized by a sequence that uses positive
flips only.

\subsection{Enumeration}
\label{sec:enum}
Our main enumerative result is an expression for the partition
function of the RYG dimer model, which we now define. Consider
a rail yard graph
$G=\RYG(\ell,r,\underline{a},\underline{b})$, and a sequence
of formal variables $\underline{x}=(x_\ell,x_{\ell+1},\ldots,x_r)$,
where possibly $\ell=-\infty$ or $r=+\infty$. The \emph{weight} of an
admissible dimer covering $C$ of $G$ is then defined as
\begin{equation}
  \label{eq:xweights}
  w(C) = \prod_{i=\ell}^r x_i^{d_i(C)}
\end{equation}
where $d_i(C)$ is the number of diagonal dimers in column $i$ (i.e.\
the number of covered diagonal edges incident to an even vertex with
abscissa $2i$). This weight is well-defined since $\sum d_i(C)$ is
finite by the definition of an admissible dimer covering. The
\emph{partition function of the multivariate RYG dimer model}, denoted
$Z(G;\underline{x})$,
is then the sum of the weights of all \emph{pure} dimer coverings of
$G$.

\begin{thm}
  \label{thm:enum1}
  The partition function of the multivariate RYG dimer model reads
  \begin{equation} \label{eq:hooklength} Z(G;\underline{x}) =
    \prod_{\substack{\ell \leq i < j \leq r\\ b_i=+,b_j=-}} z_{ij}
  \end{equation}
  where
  \begin{equation}
    z_{ij} =
    \begin{cases}
      1 + x_i x_j & \text{if $a_i \neq a_j$,} \\
      (1 - x_i x_j)^{-1} & \text{if $a_i = a_j$.} \\
    \end{cases}
  \end{equation}
\end{thm}

\begin{rem}
  The partition function is always a well-defined power series in the
  $x_i$'s: indeed, all but finitely many factors contribute a factor
  $1$ to the coefficient of a given monomial in
  \eqref{eq:hooklength}.
\end{rem}

An interesting specialization is the \emph{$q$-RYG dimer model}: given
a formal variable $q$, we attach to each configuration a weight $q^d$
with $d$ its flip distance to the fundamental one. As explained in
Section~\ref{sec:enumproof} below, this can be achieved by taking, for
all $i\in [\ell..r]$, $x_i=q^i$ if $b_i=-$, and $x_i=1/q^i$ if
$b_i=+$, with $q$ an indeterminate. A caveat is that, when $\ell$ or
$r$ is infinite, this specialization may be ill-defined since an
infinite number of monomials in the $x_i$'s might specialize to the
same monomial $q^d$. A sufficient condition for the specialization to
be well-defined is the following \emph{finiteness condition} on the
sign sequence:
\begin{itemize}
\item if $\ell=-\infty$, then there exists $\ell'$ finite such that $b_i=+$
  for all $i<\ell'$,
\item if $r=+\infty$, then there exists $r'$ finite such that $b_i=-$
  for all $i>r'$.
\end{itemize}
(This condition is essentially necessary, because any initial run of
$-$ or final run of $+$ in the sign sequence does not contribute to
the partition function, and can be removed without loss of generality:
any pure dimer covering coincides with the fundamental dimer covering
in the corresponding regions.)

\begin{thm}
  \label{thm:enum2}
  Assuming that the finiteness condition holds, the partition function
  of the $q$-RYG dimer model is
  \begin{equation} \label{eq:hooklength2}
    Z(G;q) = \prod_{\substack{\ell
        \leq i < j \leq r\\ b_i=+,b_j=-}} z_{ij}
  \end{equation}
  \begin{equation}
    z_{ij} =
    \begin{cases}
      1 + q^{j-i} & \text{if $a_i \neq a_j$,} \\
      (1 - q^{j-i})^{-1} & \text{if $a_i = a_j$.} \\
    \end{cases}
  \end{equation}
\end{thm}

The proofs of Theorems~\ref{thm:enum1} and \ref{thm:enum2} are given
in Section~\ref{sec:enumproof}.

\begin{rem}
  The product form \eqref{eq:hooklength2} is strongly reminiscent of a
  hook-length formula, upon interpreting the sign sequence as
  describing the shape of a (possibly infinite) Young diagram, see
Figure~\ref{fig:hooklength}. The
  finiteness condition ensures that there are finitely many ``hooks''
  of a given length, and hence that~\eqref{eq:hooklength2} is a well-defined
formal power series in $q$.
\end{rem}
\begin{figure}[htpb]
  \centering
\begin{center}
  \includegraphics[width=0.7\textwidth]{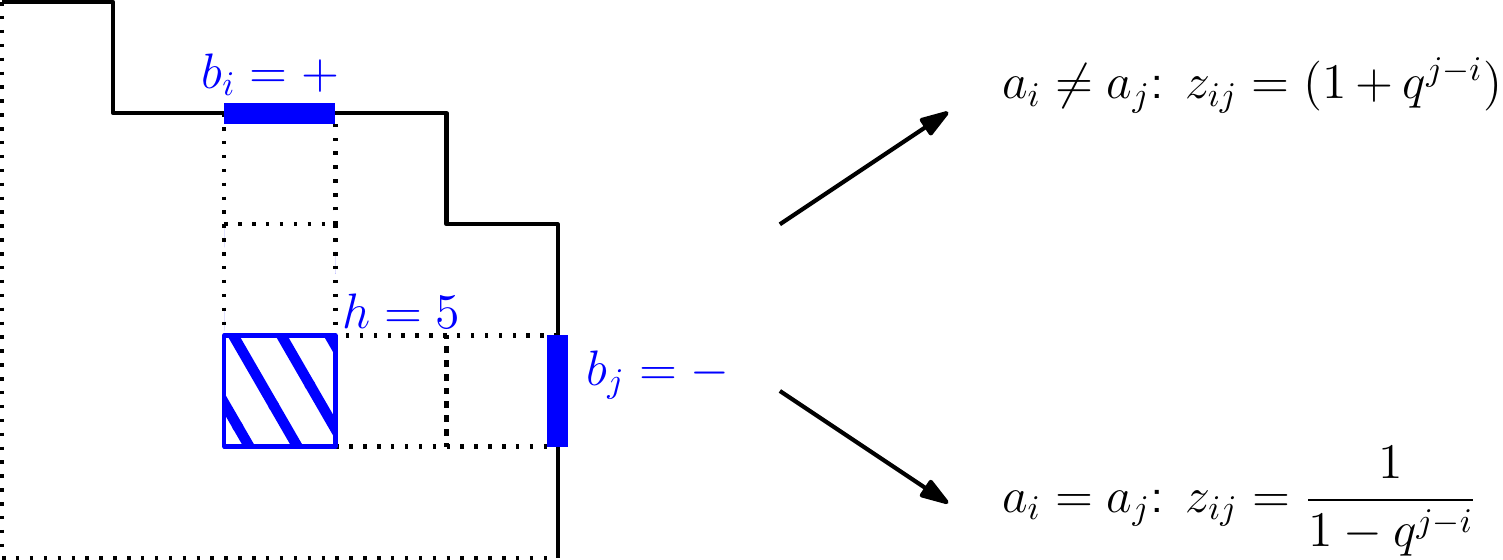}
\end{center}
  \caption{The interpretation of Theorem~\ref{thm:enum2} as a
hook-length formula. Displaying the sign sequence $\underline{b}$ as a lattice path whose horizontal
(resp. vertical) steps correspond to $+$ (resp. to $-$), one obtains a
Young diagram whose boxes are indexed by pairs $(i,j)$ such that 
$i<j$, $b_i=+$,
$b_j=-$. The quantity $h=j-i$ is the ``hook-length'' of the box. In
\eqref{eq:hooklength2}, each box gives rise to a multiplicative
factor $z_{i,j}$ whose value $1+q^h$ or $(1-q^h)^{-1}$ is determined by
comparing the two terms $a_i$ and $a_j$ of the LR sequence $\underline{a}$.
}
  \label{fig:hooklength}
\end{figure}

\subsection{Correlations}
\label{sec:corr}

So far we have introduced the partition function of the RYG dimer
model, which depends on a sequence of formal variables $\underline{x}$
in general and on a single variable $q$ in the flip
specialization. For a probabilistic or statistical physics
interpretation, one shall rather consider the $x_i$'s or $q$ as
nonnegative real numbers such that the sum $Z(G;\underline{x})$ of the
weights $w(C)$ over all pure dimer coverings $C$ is convergent. As
apparent from Theorem~\ref{thm:enum1}, this is the case if and only if
\begin{equation}
  \label{eq:cc1}
  x_i x_j < 1 \text{ for all $i<j$ such that $a_i=a_j$, $b_i=+$ and
    $b_j=-$}
\end{equation}
and, when $\ell$ or $r$ is infinite,
\begin{equation}
  \label{eq:cc2}
  \sum_{\substack{\ell \leq i < j \leq r\\ b_i=+,b_j=-}} x_i x_j < \infty.
\end{equation}
In the $q$-RYG dimer model, \eqref{eq:cc1} is satisfied whenever
$q<1$, and \eqref{eq:cc2} amounts to the finiteness condition defined
before.

Assuming that the RYG dimer model is \emph{well-defined}, that is to
say \eqref{eq:cc1} and \eqref{eq:cc2} are satisfied, we may interpret
$w(C)/Z(G;\underline{x})$ as the probability of the pure dimer
covering $C$. For a finite set $E$ of edges, we denote by
$P_{G;\underline{x}}(E)$ the probability that all the edges of $E$ are
covered by a dimer. Our main probabilistic result is an explicit
determinantal expression for $P_{G;\underline{x}}(E)$, which requires
to introduce some notations.
For $k,k'$ two integers, we set
\begin{equation}
F_k(z)=\frac
{\displaystyle
 \prod\limits_{\substack{i:(a_i,b_i)=(R,+) \\ 2i<k}}{\left(1+x_iz\right)}
 \prod\limits_{\substack{j:(a_j,b_j)=(L,-) \\ 2j>k}}{\left(1-\frac{x_j}{z}\right)}}
{\displaystyle
 \prod\limits_{\substack{i:(a_i,b_i)=(L,+) \\ 2i\leq k}}{\left(1-x_iz\right)}
 \prod\limits_{\substack{j:(a_j,b_j)=(R,-) \\ 2j\geq k}}{\left(1+\frac{x_j}{z}\right)}}.
 \label{eq:def_Fk}
\end{equation}
and
\begin{equation}
  \label{eq:def_Gkk}
  G_{k,k'}(z,w)=  \frac{F_{k}(z)}{F_{k'}(w)}  \frac{\sqrt{zw}}{z-w}.
\end{equation}
Note that all the products in \eqref{eq:def_Fk} are convergent by
\eqref{eq:cc2}, hence $F_k(z)$ is a meromorphic function on the whole
complex plane, whose all zeros and poles are on the real axis.

For $\alpha,\beta$ two vertices of $G$ such that $\alpha^{\mathrm{x}}$
is even and $\beta^{\mathrm{x}}$ is odd, we set
\begin{equation}
  \label{eq:Cdef}
  \mathcal{C}_{\alpha,\beta}= \frac{1}{\left(2i\pi\right)^2}
  \oint_{C_{z}}\oint_{C_{w}}
  G_{\alpha^{\mathrm{x}},\beta^{\mathrm{x}}}(z,w)
  \frac{w^{\beta^{\mathrm{y}}}}{z^{\alpha^{\mathrm{y}}}}
  \frac{\mathrm{d}z}{z}
  \frac{\mathrm{d}w}{w}
\end{equation}
where the contours must satisfy the following conditions: (i) $C_z$
should encircle $0$ and all the negative poles of
$F_{\alpha^{\mathrm{x}}}(z)$, but not the positive ones; (ii) $C_w$
should encircle $0$ and all the positive zeros of
$F_{\beta^{\mathrm{x}}}(w)$, but not the negative ones; (iii) $C_z$
and $C_w$ should not intersect, and $C_z$ should surround $C_w$ if and
only if $\alpha^{\mathrm{x}} < \beta^{\mathrm{x}}$. We shall check in
Section~\ref{sec:det_corr_proof} that the assumptions \eqref{eq:cc1}
and \eqref{eq:cc2} imply that such contours always exist but, at this
stage, let us mention their intuitive interpretation:
$\mathcal{C}_{\alpha,\beta}$ is obtained by extracting the coefficient
of $z^{\alpha^{\mathrm{y}}} w^{-\beta^{\mathrm{y}}}$ in
$G_{\alpha^{\mathrm{x}},\beta^{\mathrm{x}}}(z,w)$, when we treat each
factor $(1 - x_i z)^{-1}$ as a power series in $z$, each factor
$(1 + x_j / z)^{-1}$ as a power series in $z^{-1}$, each factor
$(1 + x_i w)^{-1}$ as a power series in $w$, each factor
$(1 - x_j / w)^{-1}$ as a power series in $w^{-1}$, and finally we
expand $\sqrt{zw}/(z-w)= \sum_{k \in \mathbb{N}+1/2} (w/z)^k$ if
$\alpha^{\mathrm{x}} < \beta^{\mathrm{x}}$, or
$\sqrt{zw}/(z-w)= \sum_{k \in \mathbb{N}+1/2} (w/z)^{-k}$ if
$\alpha^{\mathrm{x}} > \beta^{\mathrm{x}}$.

We are now ready to state our theorem. Recall that, for an edge
$e=(\alpha,\beta)$, $\alpha$ and $\beta$ are assumed to be
respectively the even and the odd endpoint of $e$.

\begin{thm}[Dimer correlations]
  \label{thm:det_corr}
  Let $E=\{e_1,\ldots,e_s\}$ be a finite set of edges of
  $\RYG(\ell,r, \underline{a}, \underline{b})$, with
  $e_i=(\alpha_i,\beta_i)$. Then, we have
  \begin{equation}
    \label{eq:det_corr}
    P_{G;\underline{x}}(E)=(-1)^{H(E)}\underline{x}^{\underline{n}}
    \det_{1\leq i,j\leq s} (\mathcal{C}_{\alpha_i,\beta_j}),
  \end{equation}
  with $H(E)$ the number of horizontal edges in $E$ whose right
  endpoint is at an even abscissa,
  $\underline{x}^{\underline{n}}=x_\ell^{n_\ell}\cdots x_r^{n_r}$ with
  $n_k$ the number of diagonal edges in $E$ in column $k$, and
  $\mathcal{C}$ defined as in \eqref{eq:Cdef}.
\end{thm}

The proof of Theorem~\ref{thm:det_corr} is given in
Section~\ref{sec:ferm}, where we also prove that the infinite matrix
$\mathcal{C}$, with rows indexed by even vertices $\alpha$ and columns
by odd vertices $\beta$, is an inverse of the Kasteleyn matrix of the
rail yard graph for a suitable Kasteleyn orientation (see
Theorem~\ref{thm:inverse_kast}). Applications will be discussed in
Section~\ref{sec:particular}.

\section{Bosonic operators}
\label{sec:bos}

The purpose of this section is to establish Theorems~\ref{thm:enum1}
and \ref{thm:enum2}. This is done naturally by the transfer-matrix
method, which here consists in decomposing the pure dimer coverings we
want to enumerate into a sequence of compatible elementary dimer
coverings. It turns out that the transfer matrices are isomorphic to
certain operators arising in the so-called boson-fermion
correspondence. For more details on this latter subject, we refer the
reader to one of the many references available in the mathematical
physics literature, for instance \cite[Chapter 14]{Kac}, \cite{MJD},
\cite[Appendix A]{Okounkov:wedge},
\cite{OR2}, \cite{Tin11} and
\cite{AlZa13}.

We start by giving the necessary reminders in
Section~\ref{sec:bosrem}, then make the connection with rail yard
graphs in Section~\ref{sec:bosmat}, and finally complete the proofs of
Theorems~\ref{thm:enum1} and \ref{thm:enum2} in
Section~\ref{sec:enumproof}.

\subsection{Reminders}
\label{sec:bosrem}

An \emph{integer partition}, or \emph{partition} for short, is a
nonincreasing sequence $\lambda=(\lambda_i)_{i \geq 1}$ of integers
which vanishes eventually.
The \emph{size} of a partition $\lambda$ is
$|\lambda|=\sum_{i \geq 1} \lambda_i$. We say that two partitions
$\lambda$ and $\mu$ are \emph{interlaced}, and we write
$\lambda \succ \mu$ or $\mu \prec \lambda$, if we have
\begin{equation}
  \lambda_1 \geq \mu_1 \geq \lambda_2 \geq \mu_2 \geq \lambda_3 \geq \cdots
\end{equation}
In the well-known pictorial representation in terms of Young diagrams,
this means that the skew shape $\lambda/\mu$ is a horizontal strip, see e.g. \cite[Chap. 7]{Stanley} for more precise definitions. To
a partition $\lambda$ we may associate its \emph{conjugate}
$\lambda'$, whose Young diagram is the image of that of $\lambda$ by a
reflection along the main diagonal. In more explicit terms, we have
$\lambda'_i=\#\{j \geq 0,\ \lambda_j\geq i\}$. Note that the relation
$\lambda \succ \mu$ amounts to
\begin{equation}
  \label{eq:conjinter}
  \lambda'_i - \mu'_i \in \{0,1\} \qquad \text{for all $i \geq 1$.}
\end{equation}

The \emph{bosonic Fock space}, denoted $\mathcal{B}$, is the infinite
dimensional Hilbert space spanned by orthonormal basis vectors
$|\lambda\rangle$ where $\lambda$ runs over the set of integer
partitions. Here we will use the bra-ket notation so that
$\langle \lambda|$ denotes the dual basis vector. For $x$ a formal or
complex variable, we introduce the operators
$\Gamma_{L+}(x),\Gamma_{L-}(x),\Gamma_{R+}(x),\Gamma_{R-}(x)$ whose
action on basis vectors reads
\begin{equation}
  \label{eq:gdef}
  \begin{split}
    \Gamma_{L+}(x) |\lambda\rangle = \sum_{\mu: \mu \prec \lambda}
    x^{|\lambda|-|\mu|} |\mu\rangle, &\qquad \Gamma_{R+}(x)
    |\lambda\rangle = \sum_{\mu: \mu' \prec \lambda'}
    x^{|\lambda|-|\mu|} |\mu\rangle, \\
    \Gamma_{L-}(x) |\lambda\rangle = \sum_{\mu: \mu \succ \lambda}
    x^{|\mu|-|\lambda|} |\mu\rangle, &\qquad \Gamma_{R-}(x)
    |\lambda\rangle = \sum_{\mu: \mu' \succ \lambda'}
    x^{|\mu|-|\lambda|} |\mu\rangle.
  \end{split}
\end{equation}
These operators are sometimes called (half-)vertex operators.
Let us mention that, in the literature, $\Gamma_{L\pm}(x)$ is
often denoted $\Gamma_\pm(x)$, see e.g.\ \cite{Okounkov:wedge}, while $\Gamma_{R\pm}(x)$ is
sometimes denoted $\Gamma'_\pm(x)$ \cite{Young:orbifolds,BCC}. Observe that we have
\begin{equation}\label{eq:emptyKilled}
  \Gamma_{L+}(x) |\emptyset \rangle = \Gamma_{R+}(x) |\emptyset
  \rangle = |\emptyset \rangle, \qquad
  \langle \emptyset | \Gamma_{L-}(x) = \langle \emptyset | \Gamma_{R-}(x)
  = \langle \emptyset |
\end{equation}
where $\emptyset$ denotes the empty partition. Note also that
$\Gamma_{L-}$ (resp.\ $\Gamma_{R-}$) is the dual of $\Gamma_{L+}$
(resp.\ $\Gamma_{R+}$), and that $\Gamma_{R+}$ (resp.\ $\Gamma_{R-}$)
is conjugated to $\Gamma_{L+}$ (resp.\ $\Gamma_{L-}$) via the
involution $\omega$ of $\mathcal{B}$ sending $|\lambda\rangle$ to
$|\lambda'\rangle$.

\begin{rem}
  \label{rem:ccgam}
  For $a_1,a_2 \in \{L,R\}$, the product
  $\Gamma_{a_1-}(x_1) \Gamma_{a_2+}(x_2)$ is clearly well-defined,
  because its coefficient between two states $\langle \lambda |$ and
  $| \mu \rangle$ involves only a finite sum. The same is true for
  $\Gamma_{a_1+}(x_1) \Gamma_{a_2-}(x_2)$ when $a_1 \neq a_2$ (observe
  that the ``intermediate'' partitions cannot get too large). Infinite
  sums arise when considering $\Gamma_{a_1+}(x_1) \Gamma_{a_2-}(x_2)$
  with $a_1=a_2$, but its coefficients are power series in $x_1$ and
  $x_2$, which are convergent for $|x_1 x_2|<1$ as apparent from the
  following proposition.
\end{rem}

\begin{prop}[Commutation relations]
  \label{prop:commut}
  For $a_1,a_2 \in \{L,R\}$, we have 
  \begin{equation}
    \label{eq:commutgam}
    \Gamma_{a_1+}(x_1) \Gamma_{a_2-}(x_2) = 
    \begin{cases}
      (1-x_1x_2)^{-1} \, \Gamma_{a_2-}(x_2) \Gamma_{a_1+}(x_1) &
      \text{if $a_1=a_2$,} \\
      (1+x_1x_2)^{\phantom{-1}}\, \Gamma_{a_2-}(x_2) \Gamma_{a_1+}(x_1) &
      \text{if $a_1 \neq a_2$,} \\
    \end{cases}
  \end{equation}
  while $\Gamma_{a_1+}(x_1)$ commutes with $\Gamma_{a_2+}(x_2)$, and
  $\Gamma_{a_1-}(x_1)$ commutes with $\Gamma_{a_2-}(x_2)$.
\end{prop}

\begin{proof}
  See for instance \cite[Lemma~3.3]{Young:orbifolds} for an algebraic
  proof, and \cite[Section 3]{BBBCCV14} for a bijective proof of
  \eqref{eq:commutgam}.  The celebrated Bender-Knuth involution
  \cite[pp.\ 46-47]{MR0299574} implies that $\Gamma_{a\pm}(x_1)$
  commutes with $\Gamma_{a\pm}(x_2)$ for $a=L$ or $R$. That
  $\Gamma_{L\pm}(x_1)$ commutes with $\Gamma_{R\pm}(x_2)$ is also
  well-known, but for completeness let us here sketch a short proof:
  for two partitions $\lambda,\mu$, one sees easily that the two sets
  $\{\nu: \mu \prec \nu,\ \nu' \prec \lambda'\}$ and
  $\{\nu: \mu' \prec \nu',\ \nu \prec \lambda\}$ are nonempty if and
  only if $\lambda/\mu$ is a skew shape containing no $2\times 2$
  square. In that case, both sets have the same cardinality $2^C$,
  where $C$ is the number of connected components of $\lambda/\mu$,
  and one easily constructs a bijection between them proving the
  wanted commutation relation. Another byproduct of this bijection is
  that
  \begin{equation}
    \label{eq:gaminv}
    \Gamma_{L\pm}(x) \Gamma_{R\pm}(-x) = \Gamma_{R\pm}(-x) \Gamma_{L\pm}(x) =
    1. \qedhere
  \end{equation}
\end{proof}

Given two symbols $\bullet$ and $\circ$ (called respectively
\emph{black} and \emph{white marbles}), a \emph{Maya diagram}
\cite{MJD} is an element $\mathbf{m}$ of
$\{\bullet,\circ\}^{\mathbb{Z}+1/2}$ such that $\mathbf{m}_k$ is
eventually equal to $\bullet$ for $k \to -\infty$ and to $\circ$ for
$k \to +\infty$. It then not difficult to check that the quantity
\begin{equation}
  c = \# \{ k>0, \mathbf{m}_k = \bullet \} -
  \# \{ k<0, \mathbf{m}_k = \circ\}
\end{equation}
is a finite integer, and we call it the \emph{charge} of
$\mathbf{m}$. Let $k_1 > k_2 > \cdots$ be the positions of $\bullet$
in $\mathbf{m}$ enumerated in decreasing order, and let $\lambda_i =
k_i-c+i-1/2$: it is easily seen that $\lambda$ is a partition and that
the correspondence $\mathbf{m} \mapsto (\lambda,c)$ is one-to-one, the
pair $(\lambda,c)$ being called a \emph{charged partition}. 
Observe that we may extend the involution $\omega$ to charged
partitions by setting $\omega(\lambda,c) = (\lambda',-c)$, and this
corresponds on Maya diagrams to performing a reflection across $0$ and
exchanging $\bullet$ and $\circ$: in other words, $\mathbf{m}$ is sent
to $\mathbf{m}'$ such that
$\{\mathbf{m}_k,\mathbf{m}'_{-k}\}=\{\bullet,\circ\}$ for all
$k \in \mathbb{Z}+1/2$. By a slight abuse, we still denote by
$\emptyset$ the Maya diagram of charge $0$ corresponding to the empty
partition.

The \emph{fermionic Fock space}, denoted $\mathcal{F}$, is the
infinite dimensional Hilbert space spanned by orthonormal basis vectors
$|\mathbf{m}\rangle$ where $\mathbf{m}$ runs over the set of all Maya
diagrams. For $c \in \mathbb{Z}$, let
$\mathcal{F}_c \subset \mathcal{F}$ denote the subspace spanned by
Maya diagrams of charge $c$, so that
$\mathcal{F}=\oplus_{c\in\mathbb{Z}} \mathcal{F}_c$.  By the bijection
between Maya diagrams and charged partitions, each $\mathcal{F}_c$ may
be canonically identified with $\mathcal{B}$. This defines the action
of the bosonic operators $\Gamma_{L\pm}$ and $\Gamma_{R\pm}$ on
$\mathcal{F}$, leaving each subspace $\mathcal{F}_c$ invariant (by a
slight abuse we keep the same notations for the operators acting on
this larger space, and note that the commutations relations of
Proposition~\ref{prop:commut} remain valid). We now end this section
devoted to reminders, leaving the discussion of fermionic operators to
Section~\ref{sec:fermrem}.

\subsection{Interpretation as transfer matrices for RYGs}
\label{sec:bosmat}

The purpose of this section is to explain how the operators
$\Gamma_{L\pm}/\Gamma_{R\pm}$ may be identified with dimer transfer
matrices. The key observation is that Maya diagrams describe the
\emph{boundary states} in our model. More precisely, let us consider
an admissible dimer covering $C$ of
$G=\RYG(\ell,r,\underline{a},\underline{b})$. We define the
\emph{left boundary state} $\mathbf{l}(C)$ of $C$ by setting, for all
$k \in \mathbb{Z}+1/2$,
\begin{equation}
  \label{eq:lbsdef}
  \mathbf{l}(C)_k =
  \begin{cases}
    \circ & \text{if $(2\ell-1,k)$ is covered by a dimer,} \\
    \bullet & \text{otherwise.}
  \end{cases}
\end{equation}
It is a Maya diagram by the definition of an admissible dimer
covering. Similarly, the \emph{right boundary state} $\mathbf{r}(C)$
of $C$ is the Maya diagram defined by
\begin{equation}
  \label{eq:rbsdef}
  \mathbf{r}(C)_k =
  \begin{cases}
    \bullet & \text{if $(2\ell+1,k)$ is covered by a dimer,} \\
    \circ & \text{otherwise.}
  \end{cases}
\end{equation}
for all $k \in \mathbb{Z}+1/2$. See Figure~\ref{fig:whiteParticleJumping}(a). A pure dimer covering has both
boundary states equal to $\emptyset$. Note that, if $G'$ is a
rail yard graph which is concatenable after $G$, and if $C'$ is an
admissible dimer covering of $G'$, then $C$ and $C'$ are compatible if
and only if $\mathbf{r}(C)=\mathbf{l}(C')$. We may now state the main
result of this section.

\begin{prop}[Transfer matrix decomposition]\label{prop:RYGtransfer}
  Given a rail yard graph
  $G=\RYG(\ell,r,\underline{a},\underline{b})$ with $\ell,r$
  finite, and two Maya diagrams $\mathbf{l}$ and $\mathbf{r}$, the sum
  of the weights \eqref{eq:xweights} of all admissible dimer coverings
  of $G$ with left boundary state $\mathbf{l}$ and right boundary
  state $\mathbf{r}$ is given by
  \begin{equation}
    \label{eq:Zvert}
    Z(G,\mathbf{l},\mathbf{r};\underline{x}) = \langle \mathbf{l} |
    \Gamma_{a_\ell b_\ell}(x_\ell) \Gamma_{a_{\ell+1} b_{\ell+1}}(x_{\ell+1}) \cdots
    \Gamma_{a_r b_r}(x_r)
    | \mathbf{r} \rangle.
  \end{equation}
  In particular, the partition function reads
  \begin{equation}
    \label{eq:Zpurvert}
        Z(G;\underline{x}) = \langle \emptyset |
    \Gamma_{a_\ell b_\ell}(x_\ell) \Gamma_{a_{\ell+1} b_{\ell+1}}(x_{\ell+1}) \cdots
    \Gamma_{a_r b_r}(x_r)
    | \emptyset \rangle.
  \end{equation}
\end{prop}

\begin{proof}
  Note that \eqref{eq:Zpurvert} follows from \eqref{eq:Zvert} by
  taking $\mathbf{l}=\mathbf{r}=\emptyset$, which amounts to
  considering pure dimer coverings.

  We first verify \eqref{eq:Zvert} for $\ell=r$, i.e.\ when $G$ is an
  elementary rail yard graph. Let us here treat the case $a_\ell=L$,
  $b_\ell=+$ (displayed third on Figure~\ref{fig:elemRYG}(c)) and
  leave the other cases to the reader.  Let $s_1<s_2<\cdots$ (resp.\
  $t_1<t_2<\cdots$) be the positions of $\circ$ in $\mathbf{l}$
  (resp. $\mathbf{r}$) enumerated in increasing order. Then, we claim
  that both sides of \eqref{eq:Zvert} are equal to $x_\ell^{\sum
    (s_i-t_i)}$ if the two conditions
  \begin{gather}
    s_i-t_i \in \{0,1\} \qquad \text{for all $i \geq 1$}
    \label{eq:intercond1}\\
    \sum_{i \geq 1} (s_i-t_i) < \infty \label{eq:intercond2}
  \end{gather}
  hold, and that both sides vanish otherwise. Indeed, on the one hand,
  it is not difficult to check (see
  Figure~\ref{fig:whiteParticleJumping}) that there is at most one
  elementary dimer configuration with prescribed boundary states
  $\mathbf{l}$ and $\mathbf{r}$, and that there is exactly one such
  configuration (containing $\sum (s_i-t_i)$ diagonal dimers) if and
  only if \eqref{eq:intercond1} and \eqref{eq:intercond2} hold.  On
  the other hand, let $\lambda$ and $\mu$ be the integer partitions
  associated with the Maya diagrams $\mathbf{l}$ and $\mathbf{r}$: the
  quantity
  $\langle \mathbf{l}| \Gamma_{L+}(x_\ell) |\mathbf{r} \rangle$ is
  equal to $x_\ell^{|\mu|-|\lambda|}$ if the two conditions
  \begin{gather}
    \lambda \prec \mu \label{eq:intercondbis1}\\
    \text{$\mathbf{l}$ and $\mathbf{r}$ have the same charge $c$}
    \label{eq:intercondbis2}  
  \end{gather}
  hold, and it vanishes otherwise. But we have
  $\lambda'_i=c+i-1/2-s_i$ and $\mu'_i=c+i-1/2-t_i$ for all $i \geq 1$
  hence, in view of \eqref{eq:conjinter}, we find that the conditions
  \eqref{eq:intercond1}-\eqref{eq:intercond2} amount to
  \eqref{eq:intercondbis1}-\eqref{eq:intercondbis2}, and then that
  $|\mu|-|\lambda|=\sum (s_i-t_i)$ as wanted.

\begin{figure}[htpb]
  \centering
  \includegraphics[width=0.8\textwidth]{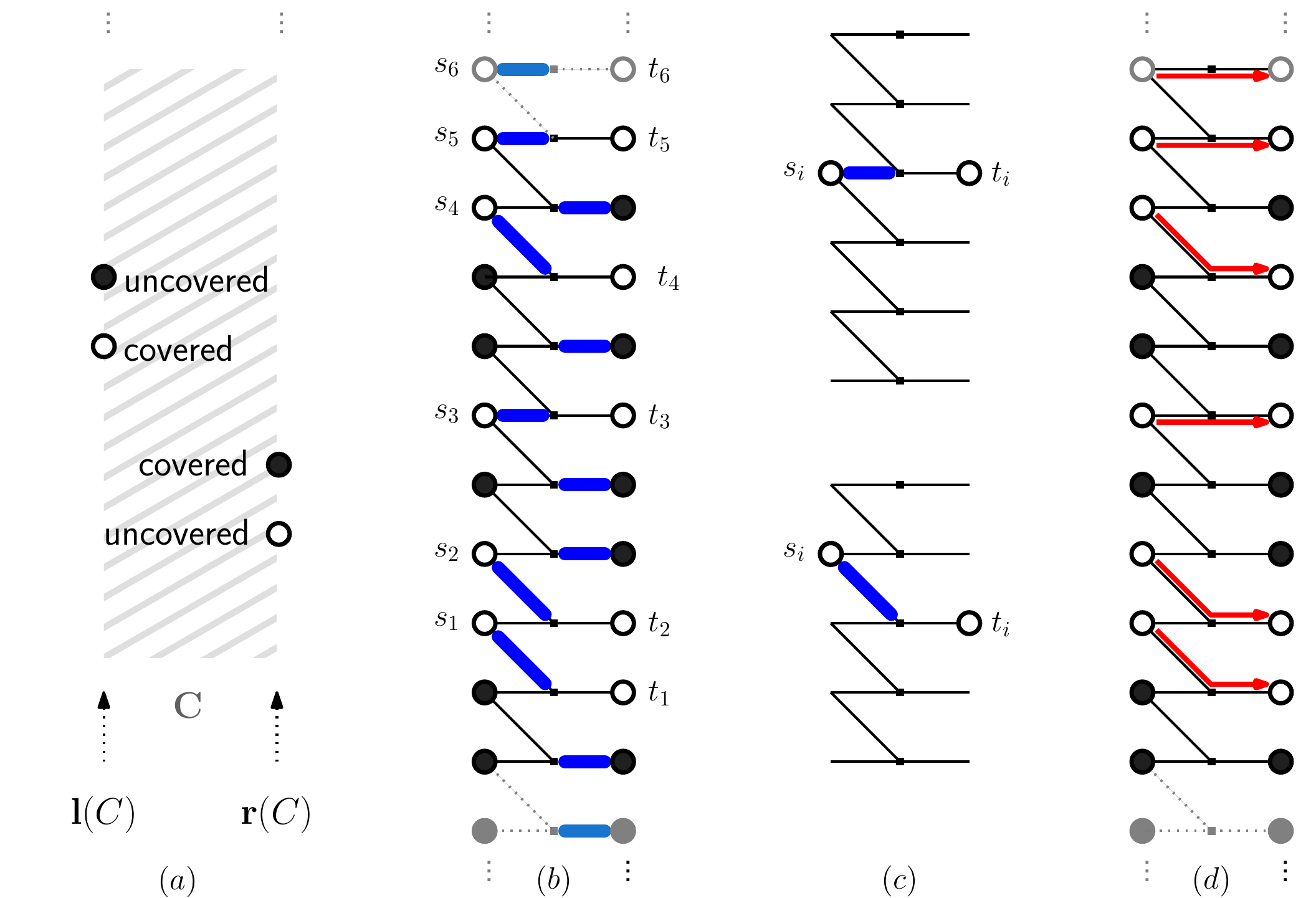}
  \caption{(a) The rules~\eqref{eq:lbsdef}-\eqref{eq:rbsdef} that define the
two Maya diagrams $\mathbf{l}(C)$ and $\mathbf{r}(C)$. (b) An admissible dimer
covering of the elementary rail yard graph $L+$. White marbles are numbered
as in the proof of Proposition~\ref{prop:RYGtransfer}. (c) 
If the left boundary vertex at ordinate $s_i$ is covered by a horizontal (resp.
diagonal) dimer, then the right boundary vertex at ordinate $s_i$ (resp.
$s_i-1$) is necessarily uncovered. Therefore there is a white
marble at this ordinate on the right boundary, and induction implies that it is the $i$-th white marble on the right boundary. This proves that $s_i-t_i \in \{0,1\}$ and that $\sum_i
(s_i-t_i)$ is equal to the number of diagonal dimers. Moreover, since dimers incident to the left boundary can be recovered from the knowledge of the sequences $(s_i)$ and $(t_i)$, it is clear that there is at most one dimer configuration with given boundary states $\mathbf{l}$ and $\mathbf{r}$, and that there is one if and only if \eqref{eq:intercond1} and \eqref{eq:intercond2} hold.
 (d) Condition~\eqref{eq:intercond1} can be
interpreted by saying that white marbles are ``jumping downwards'' by~$0$
or~$1$. 
 The cases of the elementary graphs $L-$, $R+$, and $R-$ have similar
interpretation, respectively with white marbles jumping upwards, black
marbles jumping upwards, and black marbles jumping downwards, in all cases
by $0$ or $1$.}
          \label{fig:whiteParticleJumping}
\end{figure}

  It remains to verify \eqref{eq:Zvert} for $\ell<r$, which may be
  easily done by induction: it suffices to observe that, if $G'$ is
  the rail yard graph obtained by removing the last ``strip'' of $G$,
  then any admissible dimer covering $C$ of $G$ with boundary states
  $\mathbf{l},\mathbf{r}$ is uniquely decomposed into a pair formed by
  an admissible dimer covering $C'$ of $G'$ with boundary states
  $\mathbf{l},\mathbf{m}$, for some Maya diagram $\mathbf{m}$, and an
  elementary dimer covering $E$ with boundary states
  $\mathbf{m},\mathbf{r}$, such that $w(C)=w(C')w(E)$.
\end{proof}

\subsection{Proof of enumeration results and computation of the partition function}
\label{sec:enumproof}

We are now ready to prove Theorems~\ref{thm:enum1}
and \ref{thm:enum2}. The first one is a direct consequence of the formalism developed above, whereas the second deserves an inspection of the different types of flips in rail yard graphs.

\begin{proof}[Proof of Theorem~\ref{thm:enum1}] We simply have to
  evaluate the right-hand side of \eqref{eq:Zpurvert}, which can be
  done as in in \cite[Section 4.1]{OR2}, \cite[Section
  4]{Young:orbifolds} or~\cite[Section 5.1]{BCC}: first observe that,
  by~\eqref{eq:emptyKilled}, for any $k, m$, any
  $(c_1,c_2,\dots c_{k+m}) \in \{L,R\}^{k+m}$, one has:
 \begin{equation}
        \langle \emptyset |
    \prod_{i=1}^k\Gamma_{c_i -}(z_i)
    \prod_{i=k+1}^{k+m}\Gamma_{c_i +}(z_i)
    | \emptyset \rangle = 1,
 \end{equation}
the $z_i$ being formal variables.

Now, by applying successively the commutation relations of Proposition~\ref{prop:commut}, one can transform~\eqref{eq:Zpurvert} into a scalar product of this form, up to a multiplicative prefactor, by moving to the left all the operators $\Gamma_{a_j b_j}(x_j)$ such that $b_j=-$. In order to do that, we have, for each $\ell\leq i<j\leq r$ such that $b_i=+$ and $b_j=-$, to transform the product of operators $\Gamma_{a_i+}(x_i)\Gamma_{b_j-}(x_j)$ into the product 
 $\Gamma_{b_j-}(x_j)\Gamma_{a_i+}(x_i)$. For each such transformation, we obtain a multiplicative contribution given by \eqref{eq:commutgam}, and the result follows.
\end{proof} 

\begin{proof}[Proof of Theorem~\ref{thm:enum2}]
As explained before the statement of Theorem~\ref{thm:enum2}, we have to prove that, specializing for $i\in [l..r]$ the weights \eqref{eq:xweights} to $x_i=q^i$ if $b_i=-$ and $x_i=1/q^i$ if $b_i=+$ amounts to attaching to each configuration a weight $q^d$, where $d$ is the flip distance to the fundamental configuration.

 First, this is true for the fundamental configuration that receives a weight $1$ in both cases. Second, since by Propp's theory (recalled in Section~\ref{sec:flipdef}) each shortest path from the fundamental state to any configuration is realized using positive flips only, it is enough to check that, in this specialization, each positive flip increases the weight of a configuration by a factor of $q$. 

\begin{figure}[htpb]
  \centering
  \includegraphics[width=1\linewidth]{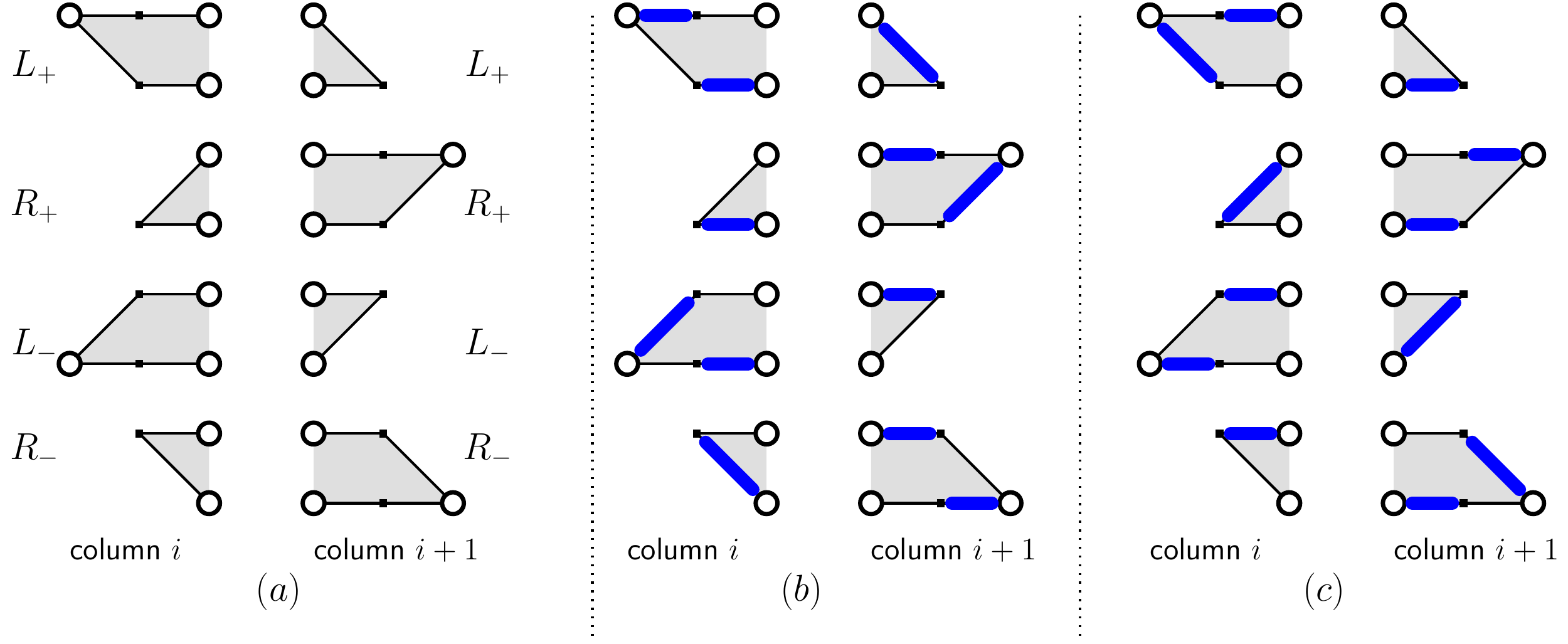}
  \caption{(a) The 16 possible face types of rail yard graphs are obtained by
    matching one of the half-face types on the left with one on the right. (b-c)
    The dimer configuration around half-faces before (b) and after (c) a
  positive flip.}
  \label{fig:flipWeights}
\end{figure}

 Consider an inner face $f$ in a rail yard graph. Then $f$ is made by the union of two half-faces as shown on Figure~\ref{fig:flipWeights}(a).
Each of these two half-faces is incident to a diagonal edge, one in column $i$, and one in column $i+1$, in the sense of Section~\ref{sec:enum}, for some $i\in[l..r-1]$.
Then, a case inspection (see Figure~\ref{fig:flipWeights}(b-c)) shows that the following is true: when performing a positive flip on $f$, the number of diagonal dimers on column $i$ increases (resp. decreases) by $1$ if $b_i=+$ (resp. $b_i=-$), and the number of diagonal dimers on column $i+1$ decreases (resp. increases) by $1$ if $b_{i+1}=+$ (resp. $b_{i+1}=-$).

Therefore this flip multiplies the weight~\eqref{eq:xweights} of the configuration by a factor of:
\begin{eqnarray}\label{eq:flipFactor}
x_i^{{b_i}}/x_{i+1}^{b_{i+1}},
\end{eqnarray}
where we identified $+,-$ with $+1,-1$ respectively.
But, in the specialization we are considering, we have $x_j^{b_j}=q^{-j}$ for all $j \in [l..r]$, so~\eqref{eq:flipFactor} is equal to $q$ and the proof is complete.
\end{proof}

\section{Fermionic operators}
\label{sec:ferm}

The purpose of this section is to establish
Theorem~\ref{thm:det_corr}. We start in Section~\ref{sec:fermrem} by
recalling the definitions and basic properties of fermionic
operators. In Section~\ref{sec:dimerop}, we show that these
operators can be used to construct \emph{constrained transfer
  matrices}, that is transfer matrices enumerating dimer
configurations containing a given subset of edges. We rewrite the
product of constrained transfer matrices in another convenient form in
Section~\ref{sec:schrheis}, and complete the proof of
Theorem~\ref{thm:det_corr} in
Section~\ref{sec:det_corr_proof}. Finally, we elucidate the connection
with Kasteleyn's theory in Section~\ref{sec:fermkast}.

\subsection{Reminders}
\label{sec:fermrem}

Recall that the fermionic Fock space $\mathcal{F}$, introduced at the
end of Section~\ref{sec:bosrem}, is the infinite dimensional Hilbert
space spanned by orthonormal basis vectors $|\mathbf{m}\rangle$ where
$\mathbf{m}$ runs over the set of all Maya diagrams. For
$k \in \mathbb{Z}+1/2$, we define the \emph{fermionic operators}
$\psi_k$ and $\psi^*_k$ (also called creation/annihilation
operators) through their action on a basis vector $|\mathbf{m}\rangle$
by
\begin{equation}
  \label{eq:fermopdef}
  \begin{split}
    \psi_k |\mathbf{m}\rangle &=
    \begin{cases}
      (-1)^{\# \{j>k, \mathbf{m}_j=\bullet\}}
      |\mathbf{m}^{(k)}\rangle & \text{if $\mathbf{m}_k=\circ$,}\\
      0 & \text{otherwise},
    \end{cases} \\
    \psi^*_k |\mathbf{m}\rangle &=
    \begin{cases}
      (-1)^{\# \{j>k, \mathbf{m}_j=\bullet\}}
      |\mathbf{m}^{(k)}\rangle & \text{if $\mathbf{m}_k=\bullet$,}\\
      0 & \text{otherwise},
    \end{cases}
  \end{split}
\end{equation}
where $\mathbf{m}^{(k)}$ is the Maya diagram obtained from
$\mathbf{m}$ by inverting the color of the marble on site $k$. Observe
that the operators $\psi_k$ and $\psi^*_k$ are adjoint to one another.
In particular, $\psi_k\psi^*_k$ (resp. $\psi^*_k\psi_k$) is the
orthogonal projector on the space spanned by Maya diagrams
$\mathbf{m}$ with $\mathbf{m}_k=\bullet$
(resp. $\mathbf{m}_k=\circ$). Fermionic operators obey the following
well-known canonical anticommutation relations:

\begin{prop}
\label{prop:anticom}
For any $k$ and $k'$ in $\mathbb{Z}+1/2$, we have
 \begin{equation}
  \{\psi_k,\psi_{k'}\}=0, \qquad
  \{\psi^*_k,\psi^*_{k'}\}=0, \qquad
  \{\psi_k,\psi^*_{k'}\}=\delta_{k,k'}.
 \end{equation}
 Here $\{a,b\}$ denotes the anticommutator of $a$ and $b$ :
 $\{a,b\}=ab+ba$.
\end{prop}

\begin{proof}
 Easy.
\end{proof}

Define the fermionic generating functions
\begin{equation}
  \label{eq:psigen}
 \psi(z)=\sum\limits_{k\in\mathbb{Z}+\frac{1}{2}}{z^k\psi_k},\qquad
 \psi^*(z)=\sum\limits_{k\in\mathbb{Z}+\frac{1}{2}}{z^{-k}\psi^*_k}.
\end{equation}
Proposition~\ref{prop:anticom} translates into
\begin{equation}
  \label{eq:anticomgen}
  \{\psi(z),\psi(w)\}=0, \qquad
  \{\psi^*(z),\psi^*(w)\}=0, \qquad
  \{\psi(z),\psi^*(w)\}=\delta(z,w)
\end{equation}
where $\delta(z,w)=\sum_{k\in\mathbb{Z}+\frac{1}{2}} (z/w)^k$ is the
formal Dirac delta function. It is straightforward to check that
\begin{equation}
  \label{eq:psipsi}
  \begin{split}
    \langle\emptyset|\psi(z)\psi^*(w)|\emptyset\rangle
    &=\sum_{\substack{k \in \mathbb{Z}+1/2 \\ k<0}}{
      \left(\frac{z}{w}\right)^k} =\frac{\sqrt{zw}}{z-w}
    \quad \text{for $|z|>|w|$,}\\
    \langle\emptyset|\psi^*(w)\psi(z)|\emptyset\rangle
    &=\sum_{\substack{k \in \mathbb{Z}+1/2 \\ k>0}}{
      \left(\frac{z}{w}\right)^k} =-\frac{\sqrt{zw}}{z-w}
    \quad \text{for $|w|>|z|$.}
  \end{split}
\end{equation}
Here the leftmost equal signs correspond to formal identities, but the
rightmost equal signs require to treat $z$ and $w$ as complex
variables. Let us also mention a lesser-known fact about the action
of the involution $\omega$ on the fermionic operators (recall that
$\omega$ is the involution that maps a charged partition $(\lambda,c)$
to $(\lambda',-c)$, hence can be seen as acting on $\mathcal{F}$).

\begin{prop}
  \label{prop:fermomega}
  For $k \in \mathbb{Z}+1/2$, we have
  \begin{equation}
    \label{eq:psiconj}
    \omega \psi^*_k \omega = (-1)^{C+k+1/2} \psi_{-k}, \qquad \text{i.e.}
    \qquad \omega \psi^*(z) \omega = (-1)^{C+1/2} \psi(-z)
  \end{equation}
  where $C$ is the charge operator (acting on $\mathcal{F}_c$ as the
  multiplication by $c$).
\end{prop}

\begin{proof}
  Follows from the fact that, for any \emph{integer} $k'$ and any Maya
  diagram $\mathbf{m}$ of charge $c$, we have
  \begin{equation}
    \label{eq:def_sl}
    \#\{j>k', \mathbf{m}_j=\bullet\} -
    \#\{j<k', \mathbf{m}_j=\circ\}
    = c-k'.
    \qedhere
  \end{equation}
\end{proof}

Last but not least, we have the following commutation relations
between bosonic and fermionic operators.

\begin{prop}
\label{prop:com_gamma_psi}
Given two formal variables $x,z$ we have
\begin{align}
 \Gamma_{L+}(x)\psi(z) &= \frac{1}{1-xz}\psi(z)\Gamma_{L+}(x) 
\label{eq:cgp1}\\
 \Gamma_{R+}(x)\psi(z) &= (1+xz)\psi(z)\Gamma_{R+}(x) 
\label{eq:cgp2}\\
 \Gamma_{L+}(x)\psi^*(z) &= (1-xz)\psi^*(z)\Gamma_{L+}(x) 
\label{eq:cgp5}\\
 \Gamma_{R+}(x)\psi^*(z) &= \frac{1}{1+xz}\psi^*(z)\Gamma_{R+}(x) 
\label{eq:cgp6}\\
 \psi(z)\Gamma_{L-}(x) &= (1-x/z)\Gamma_{L-}(x)\psi(z)
\label{eq:cgp3}\\
 \psi(z)\Gamma_{R-}(x) &= \frac{1}{1+x/z}\Gamma_{R-}(x)\psi(z)
\label{eq:cgp4}\\
 \psi^*(z)\Gamma_{L-}(x) &= \frac{1}{1-x/z}\Gamma_{L-}(x)\psi^*(z)
\label{eq:cgp7}\\
 \psi^*(z)\Gamma_{R-}(x) &= (1+x/z)\Gamma_{R-}(x)\psi^*(z)
\label{eq:cgp8}
\end{align}
The first four (resp. last four) formal identities correspond to
actual converging series when $|z|<x^{-1}$ (resp. $|z| > x$).
\end{prop}

We give in Appendix~\ref{sec:fermcom} a combinatorial proof of these
identities (note that it is sufficient to establish only one of them,
the others follow by taking duals, inverses and conjugates by
$\omega$). For algebraic derivations, see for instance the references
given at the beginning of Section~\ref{sec:bos}.

\subsection{Constrained transfer matrices}
\label{sec:dimerop}

The fermionic operators can be used to enumerate constrained dimer
configurations. A first natural idea, already used in \cite{OR1},
consists in inserting some orthogonal projectors $\psi_k \psi_k^*$ or
$\psi^*_k \psi_k$ (with various $k$'s) within the product of bosonic
operators \eqref{eq:Zpurvert} forming the partition function, which
has the effect of forcing black or white marbles to be present at
given positions. However, this does not fully determine the positions
of the dimers (there are ambiguities for the columns containing both
horizontal and diagonal edges). Remarkably, for an arbitrary rail yard
graph and an arbitrary finite set $E$ of edges, there is a suitable
way of inserting fermionic operators which precisely forces each edge
of $E$ to be covered by a dimer.

We first introduce convenient notations. Recall that writing
$(\alpha,\beta)$ for an edge implies that its endpoints $\alpha$ and
$\beta$ are such that $\alpha^{\mathrm{x}}$ is even and
$\beta^{\mathrm{x}}$ is odd. Any finite set $E$ of edges of a rail
yard graph can be decomposed ``column by column'', hence written in
the form
\begin{equation}
  \label{eq:Edec}
  E = \bigcup_{i\in\mathbb{Z}} \{(\alpha_{i,1},\beta_{i,1}),\ldots,
  (\alpha_{i,m_i},\beta_{i,m_i}),(\gamma_{i,1},\delta_{i,1}),\ldots,
  (\gamma_{i,m'_i},\delta_{i,m'_i})\}
\end{equation}
where $\alpha_{i,j}^{\mathrm{x}}=\gamma_{i,j}^{\mathrm{x}}=2i$,
$\beta_{i,j}^{\mathrm{x}}=2i-1$ and
$\delta_{i,j}^{\mathrm{x}}=2i+1$. Here $m_i$ (resp.\ $m'_i$) is the
number of edges of $E$ connecting vertices with abscissas $2i-1$ and $2i$ (resp.\ $2i$ and
$2i+1$), and is zero except for finitely many $i$.

\begin{thm}[Constrained transfer matrix decomposition]
  \label{thm:fermdim}
  Let $E$ be an arbitrary finite subset of edges of the graph
  $G=\RYG(\ell,r,\underline{a},\underline{b})$, which we decompose as
  in \eqref{eq:Edec}, and let $n_i$ ($i=\ell..r$) be the number
  of diagonal edges of $E$ with an endpoint of abscissa $2i$. The sum
  of the weights \eqref{eq:xweights} of all pure dimer configurations
  containing $E$ reads
  \begin{equation}
    \label{eq:fermdim}
    W(G;\underline{x},E) = \langle \emptyset | T_\ell T_{\ell+1} \cdots T_r
    | \emptyset \rangle
  \end{equation}
  where, for all $i\in[\ell..r]$, the \emph{constrained transfer
    matrix} $T_i$ is given by
  \begin{equation}
    \label{eq:fermdimel}
    T_i = (-1)^{k_i} x_i^{n_i}
      \left(\prod_{j=1}^{m_i} \psi_{\beta_{i,j}^{\mathrm{y}}}^* \right)
      \left(\prod_{j=1}^{m_i} \psi_{\alpha_{i,j}^{\mathrm{y}}} \right)
      \left(\prod_{j=1}^{m'_i} \psi_{\gamma_{i,j}^{\mathrm{y}}} \right)
      \Gamma_{R,b_i}(x_i)
      \left(\prod_{j=1}^{m'_i} \psi_{\delta_{i,j}^{\mathrm{y}}}^* \right)
  \end{equation}
  with $k_i=m_i(m_i-1)/2  + m'_i(m'_i-1)/2$ if $a_i=R$, and
  \begin{equation}
    \label{eq:fermdimelbis}
    T_i = (-1)^{k_i} x_i^{n_i}
      \left(\prod_{j=1}^{m_i} \psi_{\beta_{i,j}^{\mathrm{y}}}^* \right)
      \Gamma_{L,b_i}(x_i)
      \left(\prod_{j=1}^{m_i} \psi_{\alpha_{i,j}^{\mathrm{y}}} \right)
      \left(\prod_{j=1}^{m'_i} \psi_{\gamma_{i,j}^{\mathrm{y}}} \right)
      \left(\prod_{j=1}^{m'_i} \psi_{\delta_{i,j}^{\mathrm{y}}}^* \right)
  \end{equation}
  with $k_i=m_i(m_i-1)/2 + m'_i(m'_i-1)/2 + n_i$ if $a_i=L$.  More
  generally, the sum of the weights of all admissible dimer coverings
  with left boundary state $\mathbf{l}$ and right boundary state
  $\mathbf{r}$ containing $E$ reads
  \begin{equation}
    \label{eq:fermdimgen}
    W(G,\mathbf{l},\mathbf{r};\underline{x},E) = \langle \mathbf{l} |
    T_\ell T_{\ell+1} \cdots T_r
    | \mathbf{r} \rangle.
  \end{equation}
\end{thm}

Observe that we recover \eqref{eq:Zpurvert} and
\eqref{eq:Zvert} when $E$ is empty. The order in which we take the
products of fermionic operators in \eqref{eq:fermdimel} and
\eqref{eq:fermdimelbis} is irrelevant, as long as we take the same
order for both products from $1$ to $m_i$, and for both products from
$1$ to $m'_i$ (otherwise, we might get a wrong sign).

\begin{proof}[Proof of Theorem~\ref{thm:fermdim}]
  It is sufficient to prove \eqref{eq:fermdimgen} in the case of an
  elementary RYG, i.e.\ to prove that $T_i$ is indeed the wanted
  constrained transfer matrix. The general case immediately follows by
  concatenation (i.e.\ by the transfer matrix method), as done in the
  proof of Proposition~\ref{prop:RYGtransfer}.

  Let $G=\RYG(i,i,a_i,b_i)$ be an elementary RYG,
  $\mathbf{l},\mathbf{r}$ two boundary states, and $E$ a finite subset
  of edges of $G$. The general decomposition \eqref{eq:Edec} reads
  here simply
  \begin{equation}
    \label{eq:Edecel}
    E = \{(\alpha_1,\beta_1),\ldots,
    (\alpha_m,\beta_m),(\gamma_1,\delta_1),\ldots,
    (\gamma_{m'},\delta_{m'})\}.
  \end{equation}

  \begin{figure}[t]
    \centering
    \begin{center}
      \includegraphics[width=0.8\textwidth]{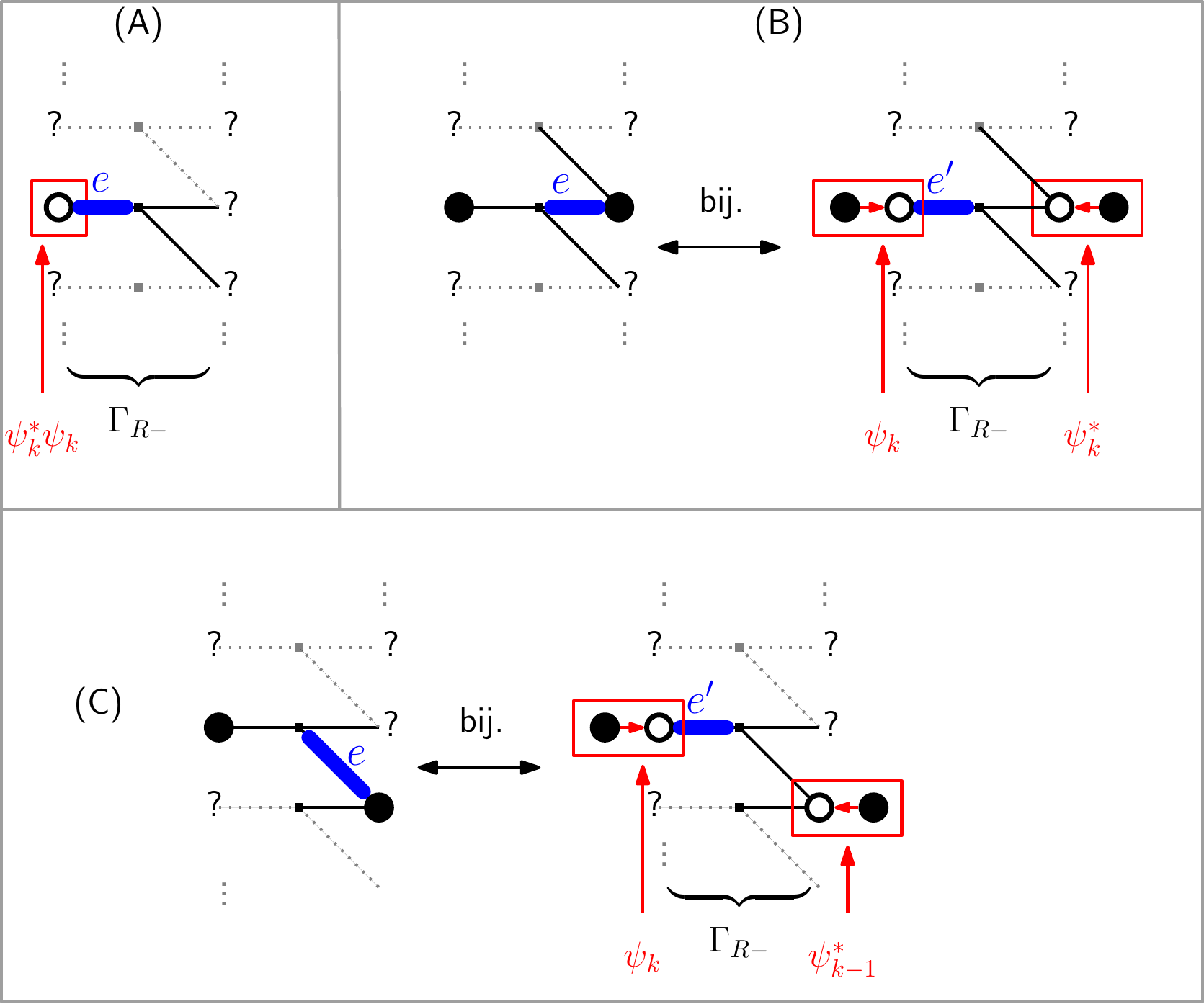}
    \end{center}
    \caption{Localization of dimers in the case of the elementary
      graph or type $R-$ (the case of $R+$ is similar). Marbles
      (resp. edges) whose status is not fixed by the discussion are
      represented with question marks (resp. dotted lines).  (A) There
      is a horizontal dimer in the left column if and only if the
      left odd vertex is occupied by a white marble. This marble is
      localized by applying the operator $\psi_k^*\psi_k$.  (B)
      Configurations with a horizontal dimer in the right column are
      such that both marbles on this level are black (although this
      condition is not sufficient). They are in bijection with
      configurations with a dimer in the left column such that both
      marbles on this level are white. The operator $\psi^*_k$
      (resp. $\psi_k$) inserted on the right (resp. left) has the
      double effect of switching the color of marbles on the $k$-th
      level and of forcing the colors of these marbles. (C) The case
      of a diagonal dimer.  }
    \label{fig:dimerOperators}
  \end{figure}

  Let us first assume that $a_i=R$, we then need to check that
  \begin{equation}
    \label{eq:fermdimgenel}
    W(G,\mathbf{l},\mathbf{r};x_i,E) = \langle \mathbf{l} |
    T_i | \mathbf{r} \rangle.
  \end{equation}
  with $T_i$ given by \eqref{eq:fermdimel}. This is immediate in the
  case $m'=0$: indeed the presence of a dimer on the edge
  $(\alpha_j,\beta_j)$ (which is necessarily horizontal) is tantamount
  to having a white marble $\circ$ at position
  $\alpha_j^{\mathrm{y}}=\beta_j^{\mathrm{y}}$ in $\mathbf{l}$, see
  Figure~\ref{fig:dimerOperators}(A). This can be achieved at the
  level of transfer matrices by multiplying the unconstrained transfer
  matrix $\Gamma_{a_i,b_i}(x_i)$ on the left by the projectors
  $\psi^*_k \psi_k$ with $k=\beta_j^{\mathrm{y}}$, $j=1..m$, and we
  get \eqref{eq:fermdimel} upon anticommuting all $\psi^*$'s to the
  left.

  The case $m'\neq 0$ is slightly more involved and requires the
  introduction of suitable ``particle hopping'' operators. Recall that the notation
  $\mathbf{m}^{(k)}$ denotes the Maya diagram obtained from
  $\mathbf{m}$ by inverting the color of the marble on site $k$. Let us
  consider an edge $e_j=(\gamma_j,\delta_j)$: having this edge covered
  by a dimer implies that
  $\mathbf{l}_{\gamma_j^{\mathrm{y}}}=\mathbf{r}_{\delta_j^{\mathrm{y}}}=\bullet$,
  but the converse is not true. However there is a bijection between,
  on the one hand, admissible dimer configurations with boundary
  states $\mathbf{l},\mathbf{r}$ containing $e_j$ and, on the other
  hand, admissible dimer configurations with boundary states
  $\mathbf{l}^{(\gamma_j^{\mathrm{y}})},\mathbf{r}^{(\delta_j^{\mathrm{y}})}$,
  that necessarily contain the edge $e'_j$ on the left of $\gamma_j$,
  see Figure~\ref{fig:dimerOperators}(B-C). We
  deduce that, at the transfer matrix level, we can force the presence
  of $e_j$ by multiplying the unconstrained transfer matrix
  $\Gamma_{a_i,b_i}(x_i)$ by $\psi_{\gamma_j^{\mathrm{y}}}$ on the
  left and by $\psi^*_{\delta_j^{\mathrm{y}}}$ on the right (if
  $\mathbf{l}$ or $\mathbf{r}$ do not have black marbles at the
  required positions then $\langle \mathbf{l}|$ or
  $|\mathbf{r}\rangle$ will be killed by this $\psi$ or $\psi^*$, as
  it should be). More generally, the product
  $\langle \mathbf{l}| \psi_{\gamma_1^{\mathrm{y}}} \cdots
  \psi_{\gamma_{m'}^{\mathrm{y}}} \Gamma_{a_i,b_i}(x_i)
  \psi^*_{\delta_{m'}^{\mathrm{y}}} \cdots
  \psi^*_{\delta_1^{\mathrm{y}}} |\mathbf{r}\rangle$
  is nonzero if and only if there is a dimer configuration with
  boundary states $\mathbf{l},\mathbf{r}$ containing all the $e_j$,
  $j=1..m'$. In that case this configuration is unique and the product
  is equal to $x_i^n$, with $n$ the number of dimers on diagonal edges
  other than the $e_j$ (it is easily seen that the signs induced by
  the $\psi/\psi^*$ all cancel out). Upon reordering the $\psi^*$'s in
  reverse order, multiplying by $x_i^{n_i}$ (to account for the weight of
	the dimers on diagonal edges in $e_j$, $j=1..m'$) 
	and multiplying by projectors on the left, we
  conclude that \eqref{eq:fermdimgenel} is true with $T_i$ given by
  \eqref{eq:fermdimel}.

  The case $a_i=L$ can be deduced from the discussion of the case
  $a_i=R$, by performing a central symmetry and exchanging the colors
  of the marbles. In other words, we simply need to take the dual of
  \eqref{eq:fermdimel} (vertical symmetry) and conjugate with the
  involution $\omega$ acting on charged partitions/Maya diagrams
  (horizontal and color symmetry). There is a slight subtlety
  regarding the sign though, which can be treated using
  Proposition~\ref{prop:fermomega}.  When taking the dual of
  $\eqref{eq:fermdimel}$, the order of the operators is reversed,
  $\Gamma_{R,\pm}(x_i)$ is changed into $\Gamma_{R,\mp}(x_i)$ and each
  $\psi$ is changed into a $\psi^*$ (and vice versa). When conjugating
  by $\omega$, $\Gamma_{R,\mp}(x_i)$ is changed into
  $\Gamma_{L,\mp}(x_i)$ and each $\psi^*$ is changed back into a
  $\psi$ (and vice versa), up to a sign. By \eqref{eq:psiconj}, the
  signs cancel out for horizontal edges in $E$, but combine into a
  sign $-1$ for each diagonal edge in $E$, which explains why the sign
  $(-1)^{k_i}$ is different in \eqref{eq:fermdimelbis}.
\end{proof}

\begin{rem}
  Similar bijective arguments are used in Appendix~\ref{sec:fermcom}
  to prove the commutation relations between bosonic and fermionic
  operators stated in Proposition~\ref{prop:com_gamma_psi}.
\end{rem}

\subsection{From the Schr\"odinger to the Heisenberg picture}
\label{sec:schrheis}

Theorem~\ref{thm:fermdim} expresses $W(G;\underline{x},E)$ as a
product of bosonic and fermionic operators taken between two vacuum
states, which we may rewrite using a strategy coming from \cite{OR1},
similar to that used in Section~\ref{sec:enumproof} for the proof of
Theorem~\ref{thm:enum1}: move all $\Gamma_+$'s to the right, and all
$\Gamma_-$'s to the left, so that they are absorbed by the vacuum
states at the end. In this process, we first pick multiplicative
factors due to the commutations between $\Gamma_+$'s and $\Gamma_-$'s:
those are precisely the same as for the partition function
$Z(G;\underline{x})$. Second, the fermionic operators get
``conjugated'' by the $\Gamma$'s crossing them. All this allows to
rewrite
\begin{equation}
  \label{eq:fermdimrew}
  W(G;\underline{x},E) = Z(G;\underline{x}) \langle \emptyset | \tilde{T}_\ell \tilde{T}_{\ell+1} \cdots \tilde{T}_r
    | \emptyset \rangle
\end{equation}
where, for all $i\in[\ell..r]$, we set
\begin{equation}
  \label{eq:fermdimelrew}
  \tilde{T}_i = (-1)^{k_i} x_i^{n_i}
  \left(\prod_{j=1}^{m_i} \Psi^*(\beta_{i,j}) \right)
  \left(\prod_{j=1}^{m_i} \Psi(\alpha_{i,j}) \right)
  \left(\prod_{j=1}^{m'_i} \Psi(\gamma_{i,j}) \right)
  \left(\prod_{j=1}^{m'_i} \Psi^*(\delta_{i,j}) \right)
\end{equation}
with, for $\beta,\alpha$ respectively odd and even vertices of $G$,
\begin{equation}
  \label{eq:Psidef}
  \begin{split}
    \Psi^*(\beta) &= \Ad\left(
      \prod_{\substack{i \leq \lfloor \beta^{\mathrm{x}}/2 \rfloor \\ b_i=+}}
      \Gamma_{a_i,+}(x_i)
      \prod_{\substack{i \geq \lceil \beta^{\mathrm{x}}/2 \rceil \\ b_i=-}}
      \Gamma^{-1}_{a_i,-}(x_i)
    \right) \psi^*_{\beta^{\mathrm{y}}},\\
    \Psi(\alpha) &= \Ad\left(
      \prod_{\substack{i \leq \alpha^{\mathrm{x}}/2 \text{ if } a_i=L\\
          i < \alpha^{\mathrm{x}}/2 \text{ if } a_i=R\\ b_i=+}}
      \Gamma_{a_i,+}(x_i)
      \prod_{\substack{i \geq \alpha^{\mathrm{x}}/2 \text{ if } a_i=R\\
          i > \alpha^{\mathrm{x}}/2 \text{ if } a_i=L\\ b_i=-}}
      \Gamma^{-1}_{a_i,-}(x_i)
    \right) \psi_{\alpha^{\mathrm{y}}}.
  \end{split}
\end{equation}
Here $\Ad$ denotes the adjoint action:
\begin{equation}
  \Ad(A)B=A B A^{-1}
\end{equation}
with $A,B$ operators acting on $\mathcal{F}$, with $A$ invertible
(recall that, by \eqref{eq:gaminv}, this is the case for the
$\Gamma$'s). Two remarks are in order. First, we may put the
$\Gamma_+$'s and the $\Gamma_-$'s in any order we want in
\eqref{eq:Psidef}, as this does not change their adjoint action.
Second, by Proposition~\ref{prop:com_gamma_psi}, $\Psi^*(\beta)$ and
$\Psi(\alpha)$ are (formal) linear combinations of $\psi^*$'s and
$\psi$'s, respectively.

\begin{rem} In physical terms, passing from the $\psi^*$/$\psi$'s to
  the $\Psi^*$/$\Psi$'s can indeed be interpreted as going from the
  Schr\"odinger to the Heisenberg picture of quantum mechanics (the
  abscissa playing the role of time). It appears from
  \eqref{eq:Psidef} that creation and annihilation operators are
  naturally attached to respectively even and odd sites. An analogous
  situation appears in the path integral formalism: if $K$ denotes the
  Kasteleyn matrix of a finite planar bipartite graph, then the
  determinant of $K$, yielding the dimer partition function, can be
  written as a Grassmann-Berezin integral \cite{MR1175176}
  \begin{equation}
    \det K = \int e^{\sum_{i,j} \xi_i K_{i,j} \xi^*_j}
    \prod_i d\xi_i \prod_j d\xi^*_j
  \end{equation}
  where the $\xi_i$'s and $\xi^*_j$ are Grassmann variables attached
  to the even and odd vertices of the graph,
  respectively. Furthermore, the contribution of dimer configurations
  containing a given collection of edges $(i_1,j_1),\ldots,(i_s,j_s)$
  is proportional to
  \begin{equation}
    \int \xi_{i_1} \xi^*_{j_1} \cdots \xi_{i_s} \xi^*_{j_s}
    e^{\sum_{i,j} \xi_i K_{i,j} \xi^*_j} \prod_i d\xi_i \prod_j d\xi^*_j.
  \end{equation}
  In other words, dimer correlations are given by the expectation
  value of a product of fermionic operators, whose form is reminiscent
  of \eqref{eq:fermdimelrew}. Note however that the approach followed
  in this paper is more akin to canonical quantization.
\end{rem}

\subsection{Proof of Theorem~\ref{thm:det_corr}}
\label{sec:det_corr_proof}

By \eqref{eq:fermdimrew} and \eqref{eq:fermdimelrew}, the ratio
\begin{equation}
  \label{eq:PEratio}
  P_{G;\underline{x}}(E)=\frac{W(G;\underline{x},E)}{Z(G;\underline{x})}
\end{equation}
is given by a product of fermionic operators taken between two vacuum
states. This can be rewritten as a determinant using Wick's formula,
as follows.

For an even vertex $\alpha$ and an odd vertex $\beta$, define the
\emph{naturally ordered product} (or ``time-ordered product'') of
$\Psi(\alpha)$ and $\Psi^*(\beta)$ by
\begin{equation}
  \label{eq:tdef}
  \mathcal{T}\left( \Psi(\alpha), \Psi^*(\beta) \right) =
  \begin{cases}
    \Psi(\alpha) \Psi^*(\beta) &
    \text{if $\alpha^{\mathrm{x}} < \beta^{\mathrm{x}}$}, \\
    - \Psi^*(\beta) \Psi(\alpha)  &
    \text{if $\alpha^{\mathrm{x}} > \beta^{\mathrm{x}}$}.
  \end{cases}
\end{equation}
We may more generally consider the naturally ordered product of more
than two $\Psi$ or $\Psi^*$, by ordering them according to the
abscissa of their argument and multiplying by the sign of the
corresponding permutation (note that operators with the same abscissa
anticommute, hence their order is irrelevant). Observe that
$\tilde{T}_\ell \tilde{T}_{\ell+1} \cdots \tilde{T}_r$ is, up to a
factor, the naturally ordered product of the $\Psi$ and $\Psi^*$
associated with the endpoints of the edges of $E$. Denoting now by
$\{e_1,\ldots,e_s\}$ the edges of $E$, with $e_i=(\alpha_i,\beta_i)$,
and by $H(E)$ the number of horizontal edges in $E$ whose
right endpoint is at an even abscissa,
we have
\begin{equation}
  \label{eq:wickapply}
  \begin{split}
    P_{G;\underline{x}}(E) &= (-1)^{H(E)} \left( \prod_{i=\ell}^r x_i^{n_i}\right)
    \langle \emptyset |
    \mathcal{T}\left(\Psi(\alpha_1),\Psi^*(\beta_1),\ldots,
      \Psi(\alpha_s),\Psi^*(\beta_s)\right) | \emptyset \rangle \\
    &= (-1)^{H(E)} \left( \prod_{i=\ell}^r x_i^{n_i}\right) \det_{1 \leq i,j \leq s}
    \langle \emptyset | \mathcal{T}\left(\Psi(\alpha_i),\Psi^*(\beta_j)\right) | \emptyset \rangle
  \end{split}
\end{equation}
where we apply Wick's formula to pass from the first to the second
line, and where the sign $(-1)^{H(E)}$ arises from the reordering of
the fermionic operators (in particular, dimers having their right
endpoint at an even abscissa appear in the ``wrong order'' in the
naturally ordered product, but the resulting sign is cancelled in the
case of diagonal dimers by that present in
Theorem~\ref{thm:fermdim}). For completeness, we provide a detailed
derivation of Wick's formula in Appendix~\ref{sec:wick}.  To complete
the proof of Theorem~\ref{thm:det_corr}, it remains to check that
\begin{equation}
  \label{eq:CTprod}
  \mathcal{C}_{\alpha,\beta} = \langle \emptyset |
  \mathcal{T}\left(\Psi(\alpha),\Psi^*(\beta)\right) | \emptyset \rangle
\end{equation}
has the announced expression \eqref{eq:Cdef}. At this stage we need to
discuss a bit analyticity conditions (so far all our
computations were done by treating the weights $x_i$'s as formal
variables). Recall that we aim at proving Theorem~\ref{thm:det_corr}
under the mere assumption that the partition function
$Z(G;\underline{x})$ is a convergent sum, which boils down to the
conditions \eqref{eq:cc1} and \eqref{eq:cc2} (see again
Remark~\ref{rem:ccgam}).

Let us temporarily strengthen \eqref{eq:cc1} into the condition
\begin{equation}
  \label{eq:cc1bis}
  x_i x_j < 1 \text{ for all $i<j$ such that $b_i=+$ and
    $b_j=-$}
\end{equation}
(that is, we also impose $x_i x_j < 1$ when $a_i \neq a_j$). We
introduce the quantities
\begin{equation}
  \label{eq:rhodef}
    \rho_+(k) = \inf_{\substack{i:b_i=+ \\ 2i \leq k}}
    \left(\frac{1}{x_i}\right), \qquad
    \rho_{-}(k) = \sup_{\substack{j:b_j=- \\ 2j \geq k}} x_j,
\end{equation}
which are nonnegative nonincreasing functions of $k$, such that
$\rho_-(k) < \rho_+(k)$ by \eqref{eq:cc1bis} and
\eqref{eq:cc2}. Recalling the definition \eqref{eq:psigen} of the
fermionic generating functions $\psi(z)$ and $\psi^*(z)$, the
definition \eqref{eq:Psidef} of $\Psi(\alpha)$ and $\Psi^*(\beta)$,
and the bosonic-fermionic commutation relations of
Proposition~\ref{prop:com_gamma_psi}, we may write
\begin{equation}
  \Psi(\alpha) = \frac{1}{2i\pi} \oint_{C_{z}}
  \frac{F_{\alpha^{\mathrm{x}}}(z) \psi(z)}{z^{\alpha^{\mathrm{y}}}} \frac{dz}{z}, 
  \qquad
  \Psi^*(\beta) = \frac{1}{2i\pi} \oint_{C_{w}}
  \frac{\psi^*(w) w^{\beta^{\mathrm{y}}}}{F_{\beta^{\mathrm{x}}}(w)} \frac{dw}{w}
\end{equation}
with $C_z$ (resp.\ $C_w$) a circle centered at $0$ with radius
comprised between $\rho_-(\alpha^{\mathrm{x}})$ and
$\rho_+(\alpha^{\mathrm{x}})$ (resp.\ between
$\rho_-(\beta^{\mathrm{x}})$ and $\rho_+(\beta^{\mathrm{x}})$). We
deduce that
$\langle \emptyset |
\mathcal{T}\left(\Psi(\alpha),\Psi^*(\beta)\right) | \emptyset
\rangle$
is equal to the wanted expression \eqref{eq:Cdef} provided that
\begin{equation}
  \label{eq:tpsipsi}
  \langle \emptyset | \mathcal{T}\left(\psi(z),\psi^*(w)\right)
  | \emptyset \rangle = \frac{\sqrt{zw}}{z-w}.
\end{equation}
But this readily follows from the definition \eqref{eq:tdef} of the
naturally ordered product and from the relations \eqref{eq:psipsi},
provided that we take the radius of $C_z$ to be strictly larger than
that of $C_w$ if $\alpha^{\mathrm{x}}<\beta^{\mathrm{x}}$, and vice
versa otherwise. We may now freely deform the contours $C_z$ and
$C_w$, as long as we hit no pole of the integrand: this establishes
\eqref{eq:Cdef} under the conditions (i)-(iii) for the contours, hence
Theorem~\ref{thm:det_corr} under the assumption \eqref{eq:cc1bis}.

We now explain how to relax this assumption into \eqref{eq:cc1}.
We proceed by multiplying each $x_i$ by a factor $t \in [0,1]$, and
noting that the assumption \eqref{eq:cc1bis} hence the identity
\eqref{eq:det_corr} are satisfied for $t$ small enough. We will show
that both sides of \eqref{eq:det_corr} have an analytic
continuation in $t$ to a domain containing the closed unit disk. As
apparent from Theorem~\ref{thm:enum1}, this is the case for the
partition function $Z(G;t \underline{x})$ (provided that
\eqref{eq:cc1} is satisfied, of course). The quantity
$W(G;t\underline{x},E)$, being a sum over a restricted subset of
configurations, is analytic too, hence so does
$P_{G;t\underline{x}}(E)$ which is the left-hand side of
\eqref{eq:det_corr}. As for the right-hand side, let us show that we
may find two contours $C_z$ and $C_w$ such that
$\mathcal{C}_{\alpha,\beta}^{(t)}$, as defined by \eqref{eq:Cdef} with
the $x_i$'s multiplied by $t$, is manifestly analytic as $t$ varies
over the unit disk. For $t=1$, let us introduce the quantities
\begin{equation}
  \label{eq:rhodefref}
  \begin{split}
    \rho_{R+}(k) = \inf_{\substack{i:(a_i,b_i)=(R,+) \\ 2i<k}}
    \left(\frac{1}{x_i}\right), \qquad
    \rho_{L-}(k) = \sup_{\substack{j:(a_j,b_j)=(L,-) \\ 2j>k}} x_j, \\
    \rho_{L+}(k) = \inf_{\substack{i:(a_i,b_i)=(L,+) \\ 2i\leq k}}
    \left(\frac{1}{x_i}\right), \qquad
    \rho_{R-}(k) = \sup_{\substack{j:(a_j,b_j)=(R,-) \\ 2j\geq k}} x_j, \\
  \end{split}
\end{equation}
which are refinements of $\rho_+(k)$ and $\rho_-(k)$. We take $C_z$
to be a circle centered on the real axis, surrounding the real
interval $[-\rho_{R-}(\alpha^{\mathrm{x}}),0]$ but not intersecting
$[\rho_{L+}(\alpha^{\mathrm{x}}),+\infty[$. Similarly, we take $C_w$
to be a circle centered on the real axis, surrounding the real
interval $[0,\rho_{L-}(\beta^{\mathrm{x}})]$ but not intersecting
$]-\infty,-\rho_{R+}(\beta^{\mathrm{x}})]$. We furthermore want $C_z$
and $C_w$ not to intersect, with $C_z$ surrounding $C_w$ iff
$\alpha^x < \beta^x$: this is possible because the functions
$\rho_{L\pm}$ are decreasing with $k$ and such that
$\rho_{L-}(k) < \rho_{L+}(k)$ for all $k$ hence, for
$\alpha^x < \beta^x$, we have
$\rho_{L-}(\beta^x) \leq \rho_{L-}(\alpha^x) < \rho_{L+}(\alpha^x)$ so
that the ``window'' $]\rho_{L-}(\beta^x),\rho_{L+}(\alpha^x)[$ in
which both $C_z$ and $C_w$ must pass is nonempty, see Figure~\ref{fig:contours} (and the case
$\alpha^x > \beta^x$ is treated by a similar reasoning on
$\rho_{R\pm}$).  Observe that these contours satisfy precisely the
conditions (i)-(iii) stated below equation \eqref{eq:Cdef}. Keeping these
contours fixed, we let $t$ vary in the interval $[0,1]$: no pole of
$G_{\alpha^{\mathrm{x}},\beta^{\mathrm{x}}}^{(t)}(z,w)$ ever hits the
contours, the conditions (i)-(iii) remain satisfied, and we conclude
that \eqref{eq:Cdef} defines the wanted analytic continuation of
$\mathcal{C}_{\alpha,\beta}^{(t)}$. \hfill $\square$
\begin{figure}
  \centering
  \includegraphics[width=0.75\linewidth]{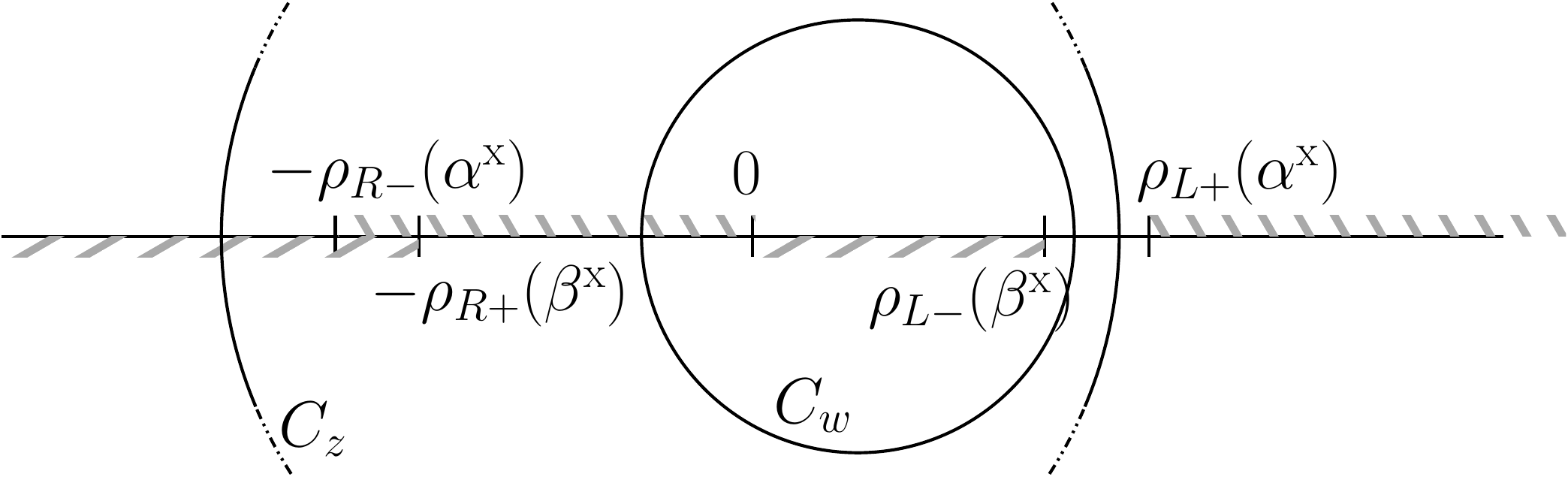}
  \caption{The contours $C_z$ and $C_w$ in the case $\beta^\mathrm{x}>\alpha^\mathrm{x}$. The upper (resp. lower) dashed regions are the intervals that the contour $C_z$ (resp. $C_w$) must avoid. On this figure it is assumed that $\rho_{R-}(\alpha^\mathrm{x})>\rho_{R+}(\beta^\mathrm{x})$ but the converse can also hold.}
  \label{fig:contours}
\end{figure}

\subsection{Inverse Kasteleyn matrix}
\label{sec:fermkast}

Let us now connect Theorem~\ref{thm:det_corr} to the general Kasteleyn
theory for the dimer model on planar bipartite graphs. We mention that
a similar connection was already observed in \cite{OR2} for the case
of lozenge tilings corresponding to skew plane partitions, see also
\cite[Section~5]{MR3148098}.

\begin{Def}
Let $G=(V,E)$ be a planar bipartite graph with no multiple edges. Consider a collection of weights on the edges $\omega:E\to\mathbb{R}_+$.

A \emph{Kasteleyn orientation} of $G$ is a map $\eta:E\to\{-1,1\}$ such that for any face $F$ the product over the edges surrounding $F$ gives :

\begin{equation}
\prod\limits_{e\in F}{\eta(e)}=
\begin{cases}
1& \mbox{if } F \mbox{ is of degree } 2 \pmod{4} \\
-1& \mbox{if } F \mbox{ is of degree } 0 \pmod{4}
\end{cases}
\end{equation}

The \emph{Kasteleyn matrix} of $G$ associated to $\eta$ is the matrix $\mathcal{K}$ whose rows (resp. columns) are indexed by white (resp. black) vertices of $G$, such that for any couple $(w,b)$ of a white vertex and a black vertex,

\begin{equation}
\mathcal{K}(w,b)=
\begin{cases}
0& \mbox{if } w \mbox{ not adjacent to } b \\
\eta(w,b)\omega(w,b)& \mbox{if there is an edge } (w,b)
\end{cases}.
\label{eq:def_kasteleyn}
\end{equation}
\end{Def}

Local statistics for dimers are known to be given by determinants of
submatrices of the Kasteleyn
matrix~\cite{Kasteleyn,Percus,Kenyon:loc_stat}:

\begin{thm}
\label{thm:kenyon}
Let $(G,\omega)$ be a finite weighted planar bipartite graph. The probability that the dimers $e_1=(w_1,b_1),\dots,e_s=(w_s,b_s)$ are present in a random dimer configuration sampled with a probability proportional to its weight is:
\begin{equation}
  P(e_1,\dots,e_s)=\left(\prod\limits_{i=1}^s{\mathcal{K}(w_i,b_i)}\right)\det\left(\mathcal{K}^{-1}(b_i,w_j)\right)_{1\leq i,j\leq s}.
\end{equation}
\end{thm}

Let us go back to the rail yard graphs. In this case, the black (resp. white) vertices are the even (resp. odd) vertices.
Recall that in pure dimer coverings of a rail yard graph $G$ , all the odd
vertices with negative (resp.~positive) ordinate on the left (resp.~right)
boundary are unmatched. These pure dimer coverings 
on $G$ correspond to classical dimer coverings with finitely many diagonal edges
on the graph $\tilde{G}$,
where those unmatched vertices and the edges attached
to them are removed.
Compare Figure~\ref{fig:truedimers} and Figure~\ref{fig:covering}.

\begin{figure}
  \centering
  \includegraphics[width=0.75\linewidth]{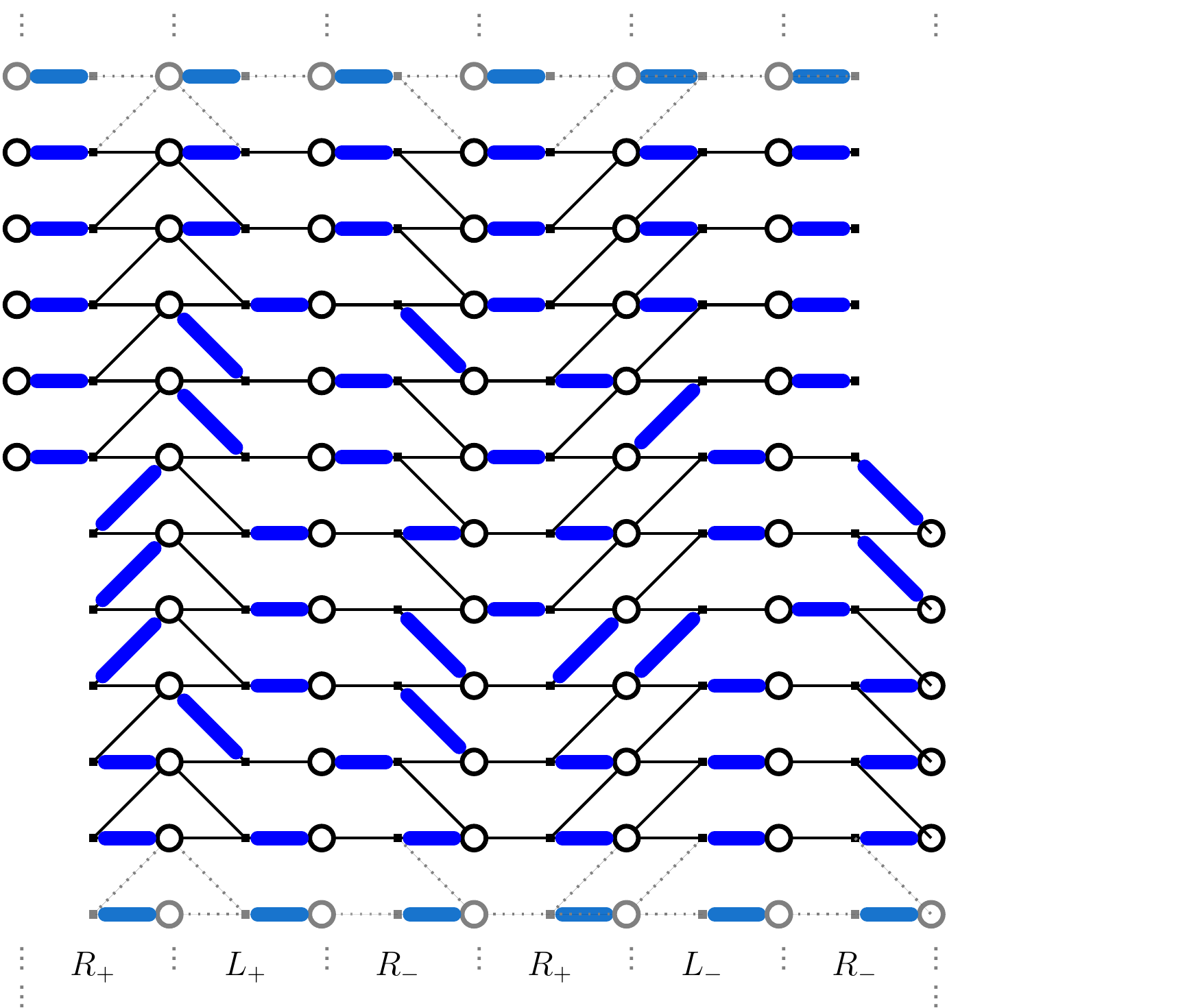}
  \caption{The perfect matching of the modified rail graph corresponding to the
    pure dimer covering of Figure~\ref{fig:covering}.}
  \label{fig:truedimers}
\end{figure}

Let us denote by $\mathfrak{M}$ the set of matched odd vertices of $G$.
Define $\eta:E\to\{-1,1\}$ by
\begin{equation}
\eta(e)=
\begin{cases}
-1 &\mbox{if $e$ is 
a horizontal edge whose right end is an even vertex,} \\
1 &\mbox{otherwise.}
\end{cases}
\end{equation}
It is easy to check that $\eta$ is a Kasteleyn orientation on $\tilde{G}$
(here the faces have degree $4$, $6$ or $8$). Construct the infinite
Kasteleyn matrix $\mathcal{K}$ on that graph as in \eqref{eq:def_kasteleyn}.
The restriction $\tilde{\mathcal{K}}$ of $\mathcal{K}$ to rows indexed by $\mathfrak{M}$
is a Kasteleyn matrix for $\tilde{G}$.
We now relate our correlation kernel $\mathcal{C}$ to the
infinite Kasteleyn matrix $\mathcal{K}$:

\begin{thm}
\label{thm:inverse_kast}
Let $G=\RYG(\ell,r,\underline{a}, \underline{b})$ be a rail yard graph and $\mathcal{K}$ be its Kasteleyn matrix for the previously defined orientation. Then if $\mathcal{C}$ is the matrix defined in Theorem~\ref{thm:det_corr}, we have:
\begin{enumerate}
  \item \text{for any $\alpha,\alpha'$ even vertices},
    $(\mathcal{C}\mathcal{K})_{\alpha,\alpha'}=\delta_{\alpha,\alpha'}$;
  \item \text{for any $\beta,\beta'$ odd vertices in $\mathfrak{M}$},
    $(\mathcal{K}\mathcal{C})_{\beta,\beta'}=\delta_{\beta,\beta'}$.
\end{enumerate}
\label{thminv}
\end{thm}

We state a lemma that will be useful to prove the theorem.
Recall that we have a bosonic operator $\Gamma_{R\pm}(x_m)$ (resp. $\Gamma_{L\pm}(x_m)$) at position $(2m+1/2,0)$ (resp. $(2m-1/2,0)$) when $a_m=R$ (resp. $a_m=L$). 
To simplify notations, we will also place a bosonic operator $Id$ at every position $(2m-1/2,0)$ (resp. $(2m+1/2,0)$) when $a_m=R$ (resp. $a_m=L$). 
This does not change the naturally ordered product of operators, and now we have one bosonic operator at each half-integer abscissa in $[2\ell-1,\ldots ,2r+1]$.
Let $B^-$ (resp. $B^+$) be the bosonic operator at abscissa $i-\frac{1}{2}$ (resp. $i+\frac{1}{2}$).
Denote by $x\sim y$ the fact that two vertices $x$ and $y$ are adjacent in $G$. 

\begin{lem}
\label{lem:Kast_com}
We have the following properties:
\begin{enumerate}
 \item Let $\alpha$ be an even vertex at position $(i,k)$.  Then
\begin{align}
\sum\limits_{\substack{\beta\sim \alpha \\ \beta^{\mathrm{x}}<i}}{\mathcal{K}(\beta, \alpha)\psi_{\beta^{\mathrm{y}}}^*B^-}&=-B^-\psi_k^* ; \label{eq:Kast_com_1} \\
\sum\limits_{\substack{\beta\sim \alpha \\ \beta^{\mathrm{x}}>i}}{\mathcal{K}(\beta, \alpha)B^+\psi_{\beta^{\mathrm{y}}}^*}&=\psi_k^*B^+.\label{eq:Kast_com_2}  
\end{align}
 \item Let $\beta$ be an odd vertex at position $(i,k)$.  Then
\begin{align}
\sum\limits_{\substack{\alpha\sim \beta \\ \alpha^{\mathrm{x}}<i}}{\mathcal{K}(\beta,\alpha)\psi_{\alpha^{\mathrm{y}}}B^-}&=B^-\psi_k ;\\
\sum\limits_{\substack{\alpha\sim \beta \\ \alpha^{\mathrm{x}}>i}}{\mathcal{K}(\beta,\alpha)B^+\psi_{\alpha^{\mathrm{y}}}}&=-\psi_kB^+. 
\end{align}
\end{enumerate}
\end{lem}

\begin{proof}
We will prove only identity \eqref{eq:Kast_com_1}. The other identities can be proved in a similar way.
We distinguish three cases : \\

\emph{Case 1: $B^-=Id$.}
To its left, $\alpha$ has only one neighbour $\beta$, which is at height $k$. We
conclude from the fact that $\mathcal{K}(\beta,\alpha)=-1$. \\

\emph{Case 2 : $B^-=\Gamma_{L-}(x)$.}
To its left, $\alpha$ has one neighbour $\beta_1$ at height $k$ and one neighbour $\beta_2$ at height $k-1$. We compute
\begin{align*}
\sum\limits_{\substack{\beta\sim \alpha \\ \beta^{\mathrm{x}}<i}}{\mathcal{K}(\beta, \alpha)\psi_{\beta^{\mathrm{y}}}^*\Gamma_{L-}(x)}
&=\mathcal{K}(\beta_1,\alpha)\psi_k^*\Gamma_{L-}(x) + \mathcal{K}(\beta_2, \alpha)\psi_{k-1}^*\Gamma_{L-}(x) \\
&=-\psi_k^*\Gamma_{L-}(x) + x\psi_{k-1}^*\Gamma_{L-}(x) \\
&=-\left[z^{-k}\right]\psi^*(z)\Gamma_{L-}(x)+\left[z^{-k+1}\right]x\psi^*(z)\Gamma_{L-}(x) \\
&=-\left[z^{-k}\right]\psi^*(z)\Gamma_{L-}(x)+\left[z^{-k}\right]\frac{x}{z}\psi^*(z)\Gamma_{L-}(x) \\
&=-\left[z^{-k}\right]\left(1-\frac{x}{z}\right)\psi^*(z)\Gamma_{L-}(x) \\
&=-\left[z^{-k}\right]\Gamma_{L-}(x)\psi^*(z) \\
&=-\Gamma_{L-}(x)\psi^*_k,
\end{align*}
where we used proposition \ref{prop:com_gamma_psi} to switch $\psi^*(z)$ and $\Gamma_{L-}(x)$. \\

\emph{Case 3 : $B^-=\Gamma_{L+}(x)$.}
In the proof of case 2, replace $\Gamma_{L-}(x)$ by $\Gamma_{L+}(x)$, $k-1$ by $k+1$ and $\frac{x}{z}$ by $xz$.
\end{proof}

We can now prove Theorem \ref{thminv}.

\begin{proof}[Proof of Theorem \ref{thminv}.]
In this proof, to make notations lighter, $Z(G; \underline{x})$ will be abbreviated as $Z$.
Let us prove the first part of the theorem. Fix two even vertices $\alpha$ at position $(i,k)$ and $\alpha'$ at position $(i',k')$. Then :

\begin{equation}
\label{eq:weighted_sum}
(\mathcal{C}\mathcal{K})(\alpha',\alpha)=\sum\limits_{\beta\sim \alpha}\mathcal{C}(\alpha',\beta)\mathcal{K}(\beta,\alpha).
\end{equation}

This sum has at most three nonzero terms, corresponding to the three odd neighbours of $\alpha$, at the abscissas $i-1$ and $i+1$. 
We now distinguish according to the position of $i'$ relatively to $i$.\\

\emph{Case 1 : $i'<i$.}

If $\beta\sim \alpha$ is at abscissa $i-1$, $\mathcal{C}(\alpha',\beta)$ can be written in the following form :

\begin{equation}
\mathcal{C}(\alpha',\beta)=\frac{1}{Z}\langle\emptyset|\Gamma^{(1)}\psi_{k'}\Gamma^{(2)}\psi^*_{\beta^{\mathrm{y}}}B^-B^+\Gamma^{(3)}|\emptyset\rangle,
\end{equation}
where $\Gamma^{(1)}, \Gamma^{(2)}$ and $\Gamma^{(3)}$ are the products of the bosonic operators located respectively before the abscissa $i'$, 
between the abscissas $i'$ and $i-1$ and after the abscissa $i+1$.

Similarly, if $\beta\sim \alpha$ is at abscissa $i+1$, $\mathcal{C}(\alpha',\beta)$ can be written in the following form:
\begin{equation}
\mathcal{C}(\alpha',\beta)=\frac{1}{Z}\langle\emptyset|\Gamma^{(1)}\psi_{k'}\Gamma^{(2)}B^-B^+\psi^*_{\beta^{\mathrm{y}}}\Gamma^{(3)}|\emptyset\rangle.
\end{equation}

For each $\beta\sim \alpha$ appearing in the sum of equation~\eqref{eq:weighted_sum}, we are going to move the $\psi^*_{\beta^{\mathrm{y}}}$ between $B^-$ and $B^+$. We separate the cases $\beta^{\mathrm{x}}<i$ and $\beta^{\mathrm{x}}>i$ to apply part~1 of the lemma:
\begin{align*}
(\mathcal{C}\mathcal{K})(\alpha',\alpha)&=\sum\limits_{\beta\sim \alpha}{\mathcal{K}(\beta,\alpha)\mathcal{C}(\alpha',\beta)} \\
&=\sum\limits_{\substack{\beta\sim \alpha \\ \beta^{\mathrm{x}}<i}}{\mathcal{K}(\beta,\alpha)\mathcal{C}(\alpha',\beta)}+
\sum\limits_{\substack{\beta\sim \alpha \\ \beta^{\mathrm{x}}>i}}{\mathcal{K}(\beta,\alpha)\mathcal{C}(\alpha',\beta)} \\
&=-\frac{1}{Z}\langle\emptyset|\Gamma^{(1)}\psi_{k'}\Gamma^{(2)}B^-\psi_k^*B^+\Gamma^{(3)}|\emptyset\rangle
+\frac{1}{Z}\langle\emptyset|\Gamma^{(1)}\psi_{k'}\Gamma^{(2)}B^-\psi_k^*B^+\Gamma^{(3)}|\emptyset\rangle \\
&=0.
\end{align*}

\emph{Case 2 : $i'>i$.}

Here $\psi_{k'}$ is to the left of any $\psi_{\beta^{\mathrm{y}}}^*$, so each $\mathcal{C}(\alpha',\beta)$ comes with a minus sign. Using here again part~1 of the lemma to move $\psi_{\beta^{\mathrm{y}}}^*$ between $B^-$ and $B^+$ for all three terms, we get similarly that $(\mathcal{C}\mathcal{K})(\alpha',\alpha)=0$ in this case. \\

\emph{Case 3 : $i'=i$.}

If $\beta$ is at abscissa $i-1$, then $\mathcal{C}(\alpha',\beta)$ can be written in the following form:
\begin{equation}
\mathcal{C}(\alpha',\beta)=-\frac{1}{Z}\langle\emptyset|\Gamma^{(1)}\psi^*_{\beta^{\mathrm{y}}}B^-\psi_{k'}B^+\Gamma^{(2)}|\emptyset\rangle,
\end{equation}
where the minus sign comes from the fact that $\psi^*_{\beta^{\mathrm{y}}}$ appears before $\psi_{k'}$.

If $\beta\sim \alpha$ is at abscissa $i+1$, we have:
\begin{equation}
\mathcal{C}(\alpha',\beta)=\frac{1}{Z}\langle\emptyset|\Gamma^{(1)}B^-\psi_{k'}B^+\psi^*_{\beta^{\mathrm{y}}}\Gamma^{(2)}|\emptyset\rangle.
\end{equation}

Applying again part~1 of the lemma, we get:
\begin{align*}
(\mathcal{C}\mathcal{K})(\alpha',\alpha)&=\sum\limits_{\beta\sim \alpha}{\mathcal{K}(\beta,\alpha)\mathcal{C}(\alpha',\beta)} \\
&=\sum\limits_{\substack{\beta\sim \alpha \\ \beta^{\mathrm{x}}<i}}{\mathcal{K}(\beta,\alpha)\mathcal{C}(\alpha',\beta)}+
\sum\limits_{\substack{\beta\sim \alpha \\ \beta^{\mathrm{x}}>i}}{\mathcal{K}(\beta,\alpha)\mathcal{C}(\alpha',\beta)} \\
&=\frac{1}{Z}\langle\emptyset|\Gamma^{(1)}\Gamma^{(2)}B^-\psi_k^*\psi_{k'}B^+\Gamma^{(3)}|\emptyset\rangle
+\frac{1}{Z}\langle\emptyset|\Gamma^{(1)}\Gamma^{(2)}B^-\psi_{k'}\psi_k^*B^+\Gamma^{(3)}|\emptyset\rangle.
\end{align*}

Using that $\psi_k^*\psi_{k'}+\psi_{k'}\psi_k^*=\delta_{k,k'}$ (Proposition~\ref{prop:anticom}) and that
\begin{equation}
Z=\langle\emptyset|\Gamma^{(1)}\Gamma^{(2)}B^-B^+\Gamma^{(3)}|\emptyset\rangle,
\end{equation}

We conclude that in this third case, $(\mathcal{C}\mathcal{K})(\alpha',\alpha)=\delta_{k,k'}$.
This concludes the proof of part 1 of the theorem : we have shown that for any even vertices 
$\alpha$ and $\alpha'$, $(\mathcal{C}\mathcal{K})(\alpha',\alpha)=\delta_{\alpha',\alpha}$.

To prove part 2 of the theorem, note that if $\beta$ and $\beta'$ are two odd
vertices in $\mathfrak{M}$, then
\begin{equation}
(\mathcal{K}\mathcal{C})(\beta,\beta')=\sum\limits_{\alpha\sim \beta}{\mathcal{K}(\beta,\alpha)\mathcal{C}(\alpha,\beta')}.
\end{equation}

We will only treat the case when $\beta$ is on the left boundary. The case when $\beta$ is on the right boundary is similar to this one. 
The case when $\beta$ is in the bulk is similar to the proof of part 1 of the theorem, by making use this time of part 2 of the lemma.
Let $(2\ell-1,k)$ be the coordinates of $\beta$ and $(i',k')$ be the coordinates
of $\beta'$. Since $\beta\in\mathfrak{M}$, we have $k>0$. 
Note that $\beta$ has neighbours only to its right. Again we distinguish according to the position of $\beta'$. \\

\emph{Case 1 : $i'>2\ell-1$.}

If $\alpha\sim \beta$, $\mathcal{C}(\alpha,\beta')$ can be written in the following form :
\begin{equation}
\mathcal{C}(\alpha,\beta')=\frac{1}{Z}\langle\emptyset|B^+\psi_{\alpha^{\mathrm{y}}}\Gamma^{(1)}
\psi^*_{k'}\Gamma^{(2)}|\emptyset\rangle.
\end{equation}
Using the second equation of part~2 of the lemma, we obtain:
\begin{align*}
(\mathcal{K}\mathcal{C})(\beta,\beta')&=\sum\limits_{\alpha\sim \beta}{\mathcal{K}(\beta,\alpha)\mathcal{C}(\alpha,\beta')} \\
&=\sum\limits_{\substack{\alpha\sim \beta \\ \alpha^{\mathrm{x}}>2\ell-1}}{\mathcal{K}(\beta,\alpha)\mathcal{C}(\alpha,\beta')} \\
&=-\frac{1}{Z}\langle\emptyset|\psi_kB^+\Gamma^{(1)}
\psi^*_{k'}\Gamma^{(2)}|\emptyset\rangle.
\end{align*}

Since $k>0$, we have $\langle\emptyset|\psi_k=0$. So in this case, $(\mathcal{K}\mathcal{C})(\beta,\beta')=0$. \\

\emph{Case 2 : $i'=2\ell-1$.}

If $\alpha\sim \beta$, $\mathcal{C}(\alpha,\beta')$ can be written in the following form :
\begin{equation}
\mathcal{C}(\alpha,\beta')=-\frac{1}{Z}\langle\emptyset|\psi^*_{k'}B^+\psi_{\alpha^{\mathrm{y}}}\Gamma^{(1)}
|\emptyset\rangle.
\end{equation}

Using the second equation of part~2 of the lemma, we obtain :
\begin{equation}
(\mathcal{K}\mathcal{C})(\beta,\beta')=+\frac{1}{Z}\langle\emptyset|\psi^*_{k'}\psi_kB^+\Gamma^{(1)}
|\emptyset\rangle.
\end{equation}

Using Proposition~\ref{prop:anticom}, we get :
\begin{align*}
(\mathcal{K}\mathcal{C})(\beta,\beta')&=-\frac{1}{Z}\langle\emptyset|\psi_k\psi^*_{k'}B^+\Gamma^{(1)}
|\emptyset\rangle + \delta_{k,k'}\frac{1}{Z}\langle\emptyset|B^+\Gamma^{(1)}
|\emptyset\rangle \\
&=0+\delta_{k,k'}\frac{1}{Z}Z=\delta_{k,k'}.
\end{align*}

To sum up, if $\beta$ is an odd vertex in $\mathfrak{M}$ on the left boundary
and $\beta'$ is any other odd vertex in $\mathfrak{M}$,
$(\mathcal{K}\mathcal{C})_{(\beta,\beta')}=\delta_{\beta,\beta'}$.
\end{proof}

\begin{rem}
  If $\tilde{\mathcal{C}}$ (resp.~$\tilde{\mathcal{K}}$) is the restriction of
  $\mathcal{C}$ (resp.~$\mathcal{K}$) to
  columns (resp.~rows) indexed by the set $\mathfrak{M}$ of matched odd
  vertices, then $\tilde{\mathcal{C}}$ is a left and right inverse of
  $\tilde{\mathcal{K}}$. Indeed, the extra terms coming from vertices not in
  $\mathfrak{M}$ vanish when evaluated against left and right vacuums.
\end{rem}

\section{Lozenge and domino tilings}
\label{sec:particular}

In this section we discuss how the RYG dimer model gives rise, upon
taking a constant or an alternating $LR$-sequence and performing a
simple change of coordinates, to a class of lozenge or domino tilings
of the plane, respectively. The class of lozenge tilings, discussed in
Section~\ref{sec:lozenge}, contains in particular plane partitions and
was already studied in~\cite{OR1,OR2}. The class of domino tilings,
discussed in Section~\ref{sec:steep}, is that of \emph{steep
  tilings}~\cite{BCC}, and contains tilings of the Aztec diamond and
pyramid partitions as special cases.

\subsection{Lozenge tilings}\label{sec:lozenge}
In this section we consider the case where the $LR$-sequence $\underline{a}$ is constant. We assume that $\underline{a}=L^{r-\ell+1}$ (it is easy to check that taking $\underline{a}=R^{r-\ell+1}$ gives rise to the same model up to a vertical reflection and an inversion of the sign sequence).
In this case the only two elementary RYG involved are of type $L+$ and $L-$. According to the discussion of Section~\ref{sec:bos}, the transfer matrix operators corresponding to these two graphs can be interpreted as the operators $\Gamma_{L+}$ and $\Gamma_{L-}$, whose action on the Bosonic Fock space  interlaces a partition ``upwards'' or ``downwards'' respectively, see~\eqref{eq:gdef}. It follows that, for each sign sequence $\underline{b}\in\{+,-\}^{r-\ell+1}$, admissible dimer coverings of $\RYG(\ell,r, \underline{a}, \underline{b})$ are in bijection with sequences $(\lambda^{(i)})_{\ell\leq i\leq r+1}$ of integer  partitions such that $\lambda^{(i)} \substack{\prec \\ \succ} \lambda^{(i+1)}$ where the interlacing relation is $\prec$ or $\succ$  if $b_i=+$ or $b_i=-$, respectively. The correspondence goes via Maya diagrams and is the one described in Section~\ref{sec:bos}. It is well-known that such sequences are in bijection with certain lozenge tilings of the plane, see e.g.~\cite{OR1}, so the reader may be in familiar ground. In the rest of Section~\ref{sec:lozenge}, we will just sketch how to recover those lozenge tilings from the rail yard graphs.

We first apply the following coordinate transformation to each vertex of the rail yard graph $\RYG(\ell, r, \underline{a}, \underline{b})$:
\begin{align}\label{eq:lozengeTransform}
(2\ell-1 + x, y) \longmapsto \left( 2\ell-1  + \frac{1}{2} \lceil x/2 \rceil + \frac{\sqrt5}{4}\lfloor x/2 \rfloor, y -\frac{1}{2}\sum_{i=1}^{\left\lceil x/2 \right\rceil} (-1)^{b_i} \right).
\end{align}
Figure~\ref{fig:lozengeScalingRules} displays the effect of this transformation on the elementary graphs of type $L+$ and $L-$. The transformation is designed in such a way that after the transformation all angles between incident edges are equal to $\frac{2\pi}{3}$, and that all edges have equal length (equal to $\frac{\sqrt{5}}{4}$). Therefore the concatenation of these graphs generate a portion of the regular hexagonal lattice, see Figure~\ref{fig:lozengeScaled}--Left. Since the planar dual of the hexagonal lattice is a triangular lattice,  any dimer covering of the hexagonal lattice induces a covering of this triangular lattice by lozenges (each dimer connects two vertices in the primal, that correspond to two triangles in the dual, and the union of these two triangles forms a lozenge).  See Figure~\ref{fig:lozengeScaled}--Right. We thus recover the promised class of lozenge tilings of the plane.

Note that the fundamental covering of $\RYG(\ell, r,
\underline{a},\underline{b})$ projects (via the coordinate
transformation~\eqref{eq:lozengeTransform} and the dualization to lozenges) to a
lozenge tiling in which all lozenges under a certain separating path are of
``horizontal'' type, while all lozenges above this path are of one of the two
``vertical'' types, see Figure~\ref{fig:lozengeFundamental}. This separating
path is the image under the transformation~\eqref{eq:lozengeTransform} of the
horizontal axis in the original embedding of $\RYG(\ell, r, \underline{a},\underline{b})$. Equivalently, this path is a lattice path taking up (+) or down (-) steps,  encoded by the sequence $\underline{b}$. This path and its image are represented by red dotted paths on the figures of this section.

Finally, note that Figure~\ref{fig:lozengeFundamental} can naturally be seen as
a 3-dimensional picture, namely as a portion of the boundary of the region
$(\mathbb{R}_+^2\setminus \lambda)\times \mathbb{R}_+$, where $\lambda$ is the
shape of an integer partition. Here the contour of the shape $\lambda$ coincides
with the separating path just mentioned, \text{i.e.} the partition $\lambda$ is
encoded in Russian notation by the sequence $\underline{b}$, see
Figure~\ref{fig:lozengeFundamental} again. Note also that all the lozenge
tilings corresponding to pure dimer coverings of $\RYG(\ell,r,\underline{a}, \underline{b})$ are obtained by ``adding cubes'' to this 3-dimensional diagram in such a way that the heights of cubes stay nonincreasing on horizontal coordinates, going away from the axes. In particular the case $\underline{b}=+^n -^n$ corresponds to plane partitions of half-width at most $n$. This 3-dimensional interpretation is well known and is the one already considered in~\cite{OR1}.

\begin{figure}
\includegraphics[width=0.75\linewidth]{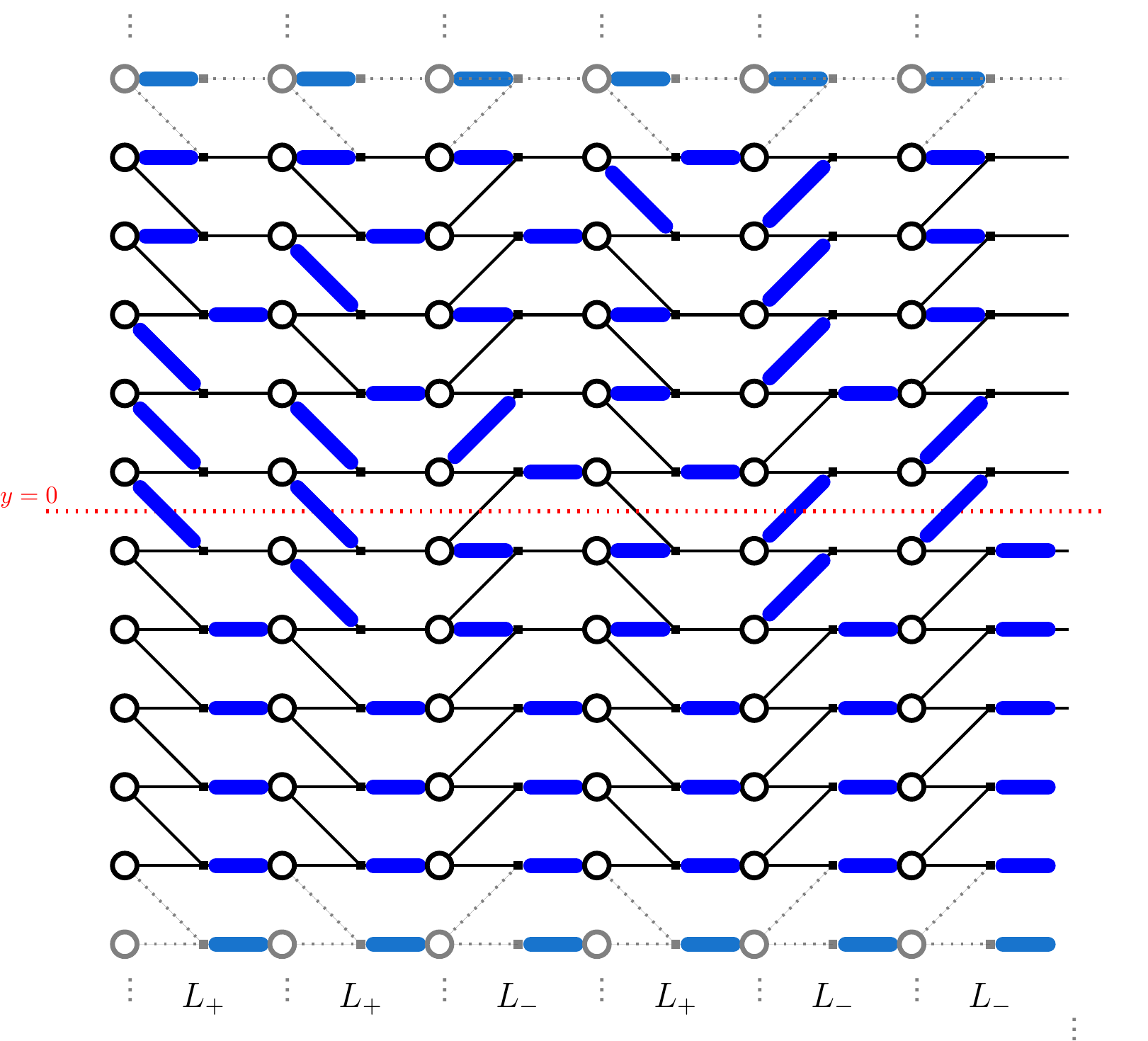}
\caption{A rail yard graph with a constant LR sequence $\underline{a}=L^6$ and sign sequence $\underline{b}=++-+--$, equipped with a pure covering.}\label{fig:lozenge}
\end{figure}
\begin{figure}
\includegraphics[width=0.6\linewidth]{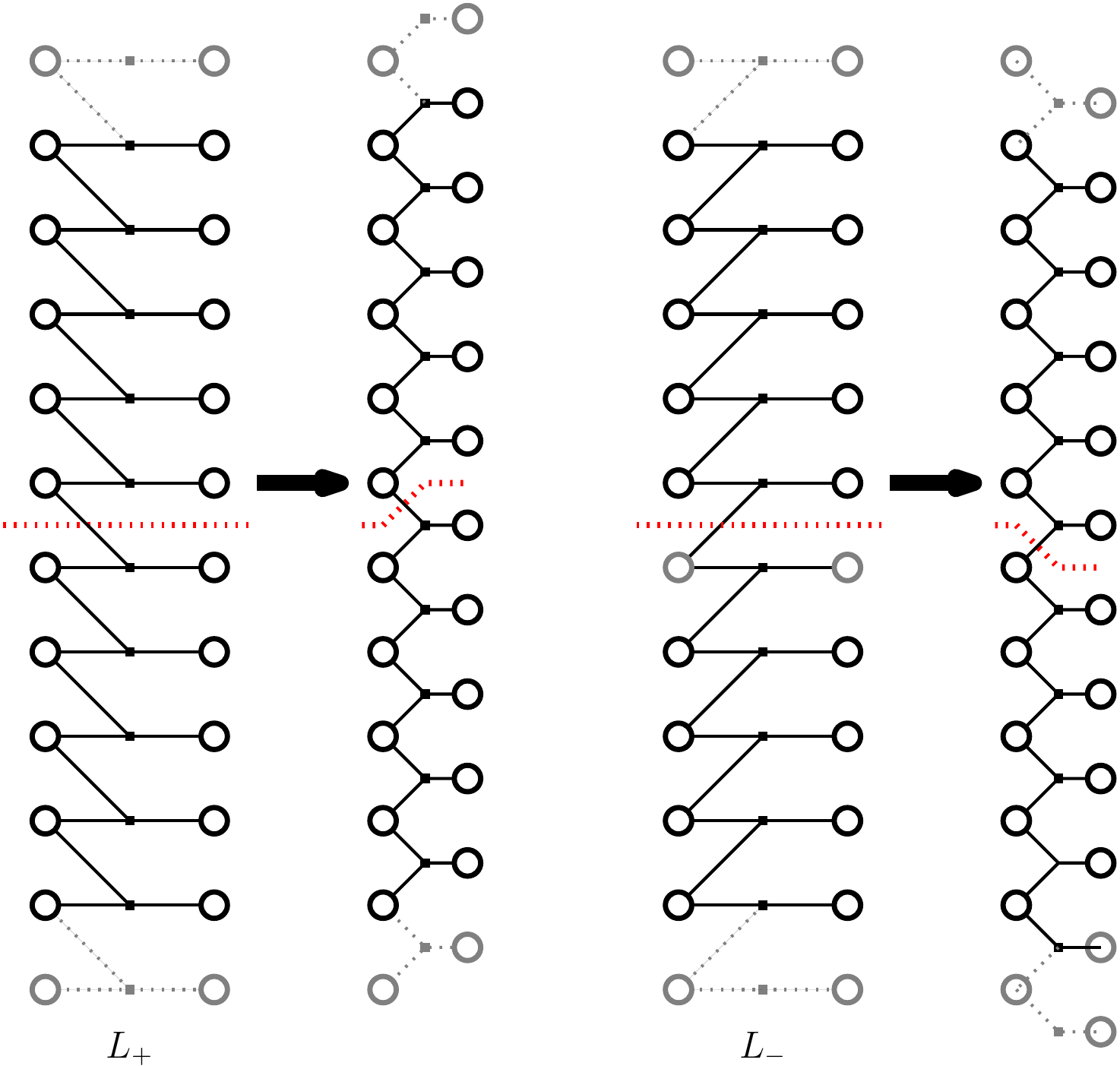}
\caption{Deformations of the elementary RYG of type $L+$ and $L-$ that generate a ``honeycomb'' lattice by concatenation.}\label{fig:lozengeScalingRules}
\end{figure}

\begin{figure}
\includegraphics[width=\linewidth]{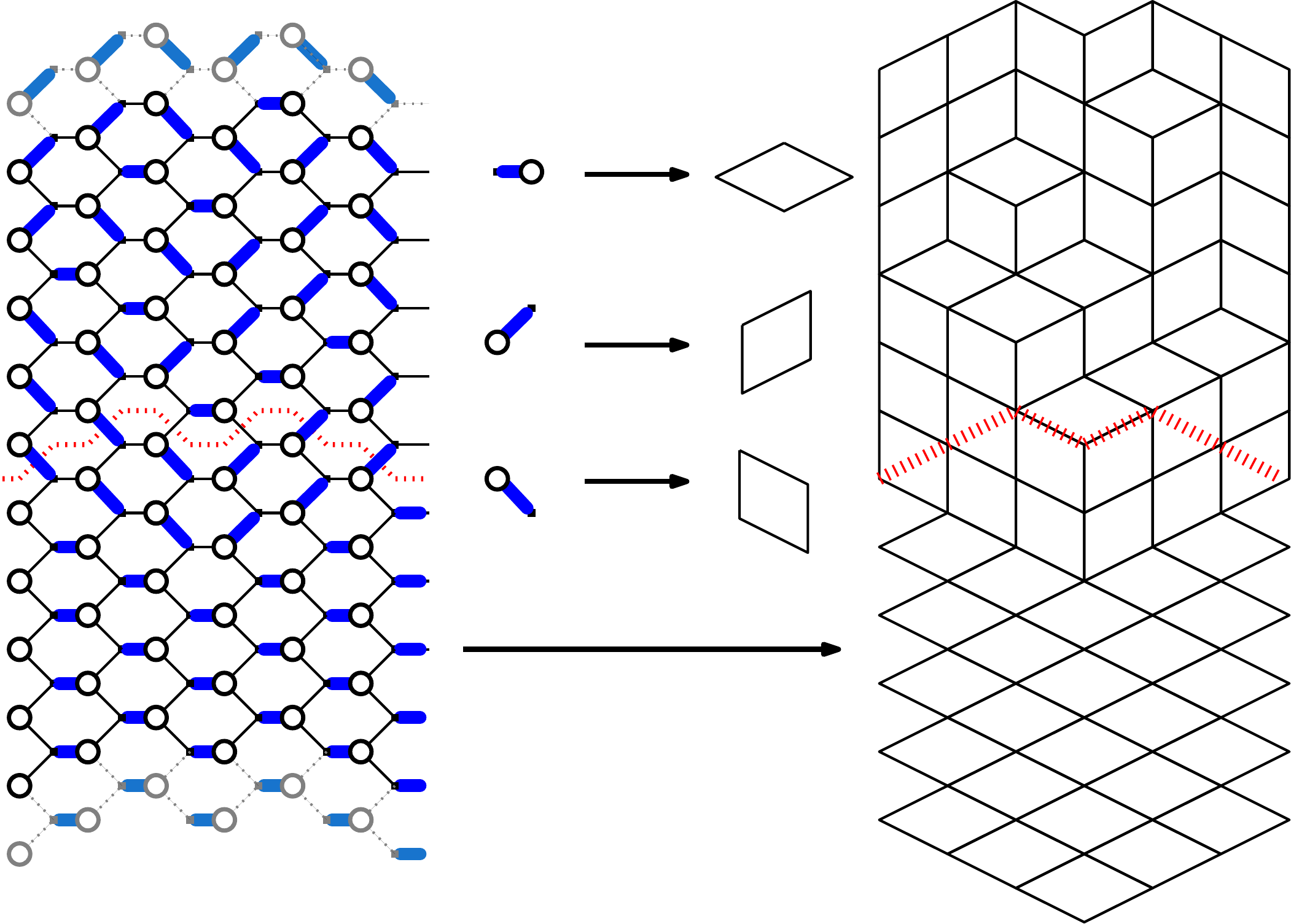}
\caption{Left: The dimer covering of Figure~\ref{fig:lozenge} as a covering of the honeycomb lattice, via the transformations of Figure~\ref{fig:lozengeScalingRules}. Right: The same objects displayed as a ``lozenge tiling'' of the plane.}\label{fig:lozengeScaled}
\end{figure}
\begin{figure}
\includegraphics[width=0.3\linewidth]{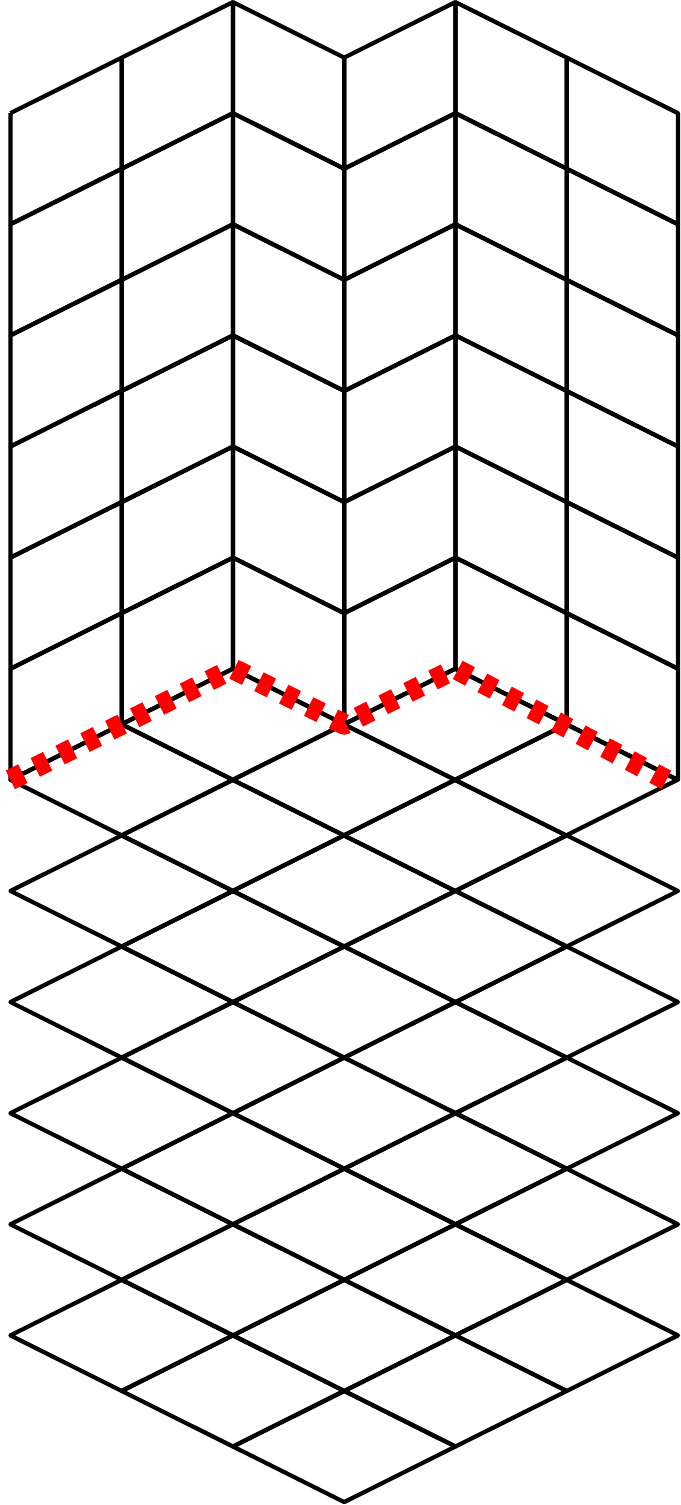}
\caption{The fundamental lozenge tiling corresponding to the case $\underline{a}=L^6$ and $\underline{b}=++-+--$.}\label{fig:lozengeFundamental}
\end{figure}

\subsection{Domino tilings} \label{sec:steep}
We now consider the case where the sequence $\underline{a}$ has even length and is an alternation of $L$ and $R$. 
 Up to an elementary symmetry we can assume that  $\underline{a}=(LR)^k$, where
 $r-\ell+1=2k$, and we fix an arbitrary sign sequence
 $\underline{b}\in\{+,-\}^{2k}$. It follows from a simple inspection of the face
 types of elementary rail yard graphs (Figure~\ref{fig:flipWeights}(a)) that all
 the inner faces of the graph $\RYG(\ell, r, \underline{a}, \underline{b})$ have degree $4$ or $8$. Moreover, for $0\leq i <k$, odd vertices located at the abscissa $x=-2\ell-1 + 4i+2$ all have degree $2$ (since they lie at the interface between an $L$-type and $R$-type elementary graphs, from left to right). These degree $2$ vertices are bounded by faces of degree $8$ on their two sides. See Figure~\ref{fig:aztecRYG} for an example.

We now let $\tilde R (\ell, r, \underline{a}, \underline{b})$ be the graph obtained by contracting all the edges incident to these inner vertices of degree $2$: in this graph all inner vertices of degree $2$ have disappeared, and all the non-boundary faces have degree $4$, see Figure~\ref{fig:aztecRYGscaled}--Left for an example. Moreover, it is easy to see that each dimer covering of $\RYG (\ell, r, \underline{a}, \underline{b})$ induces a dimer covering of $\tilde R (\ell, r, \underline{a}, \underline{b})$, with the same boundary conditions (just forget dimers on contracted edges and leave the other dimers as they were). The coordinate transformation on non deleted vertices that goes from $\RYG (\ell, r, \underline{a}, \underline{b})$ to $\tilde R (\ell, r, \underline{a}, \underline{b})$ is given by:
\begin{align}\label{eq:coordinateContractionSteep}
(-2\ell-1 + x,y) \longmapsto  \left(-2\ell-1 + x - 2 \left\lfloor \frac{x+2}{4}\right\rfloor,y\right).
\end{align}
Since all the inner vertices and faces in this new graph have degree $4$, this new graph is isomorphic to a portion of the square lattice. We let the reader check that this isomorphism can be made explicit by composing \eqref{eq:coordinateContractionSteep} with the transformation: 
\begin{align}\label{eq:coordinateTiltingSteep}
(-2\ell-1 + x,y) \longmapsto  \left(-2\ell-1 + x + y - K_x ,y - K_x\right),
\end{align}
where $K_x=\sum_{i=1}^{x} (-1)^{i+b_i}$. See Figure~\ref{fig:aztecRYGscaled} for an explicit example. The image of $\tilde R (\ell, r, \underline{a}, \underline{b})$ via these transformations is a portion of the square lattice lying in the \emph{oblique strip} $\{(X,Y): |Y-X+2\ell+1| \leq 2k\}$, see again Figure~\ref{fig:aztecRYGscaled}.

Similarly as in Section~\ref{sec:lozenge}, we can dualize this picture
to switch between a description in terms of dimers to one in terms of
tilings. The dual of the square lattice is again the square lattice,
and any dimer in the primal induces a \emph{domino} (the union of two
adjacent squares) in the dual. We thus recover a model of tilings of
the oblique strip by dominos, which are exactly the \emph{steep
  tilings} introduced in~\cite{BCC}. We invite the reader to consult
this reference for a thorough discussion on steep tilings, their link
with height functions, their encoding in terms of partitions,
etc. Here we just mention again that, not only do we recover here the
enumerative results already proved in~\cite{BCC}, but we obtain the inverse
Kasteleyn matrix and dimer correlations for these models, up to the
changes of coordinates described above.

Two particular subclasses of steep tilings had been considered
previously.  The first one is the class of domino tilings of the Aztec
diamond, which corresponds to the case where the sequence
$\underline{b}$ is also alternating. As far as we know, this is the
only case for which the inverse Kasteleyn matrix has been computed
before~\cite{Helfgott,MR3197656}. The other one is given by pyramid
partitions~\cite{Young:pyramid}, that correspond to the case where
$\underline{b}=+^\infty-^\infty$. See~\cite[Section~4.2]{BCC} for
their connection with steep tilings. We leave as an exercise the task
of making fully explicit the changes of coordinates and the
calculations for pyramid partitions, similarly as we will do for the
Aztec diamond in the next section.

\begin{rem}
  Going from rail yard graphs to steep tilings induces, in a sense, no
  loss of generality. Indeed, by forbidding diagonal dimers in the
  column $i$ of a rail yard graph (by taking their weight $x_i$ to be
  zero), this column becomes trivial in the sense that its two
  boundary states are necessarily equal, and hence it can be
  contracted in the graph (i.e.\ we can drop the corresponding
  elements from the LR and sign sequences). Starting from an infinite
  alternating LR sequence, it is possible to produce any LR sequence
  by performing such contractions, and hence to obtain any RYG. The
  coordinates of rail yard graphs however allow to express the dimer
  correlations in a more compact form, which is one of our motivations
  for introducing them.
  \label{rem:rygsteep}
\end{rem}

\section{The Aztec diamond revisited}
\label{sec:aztecCorr}

\begin{figure}
\includegraphics[width=0.75\linewidth]{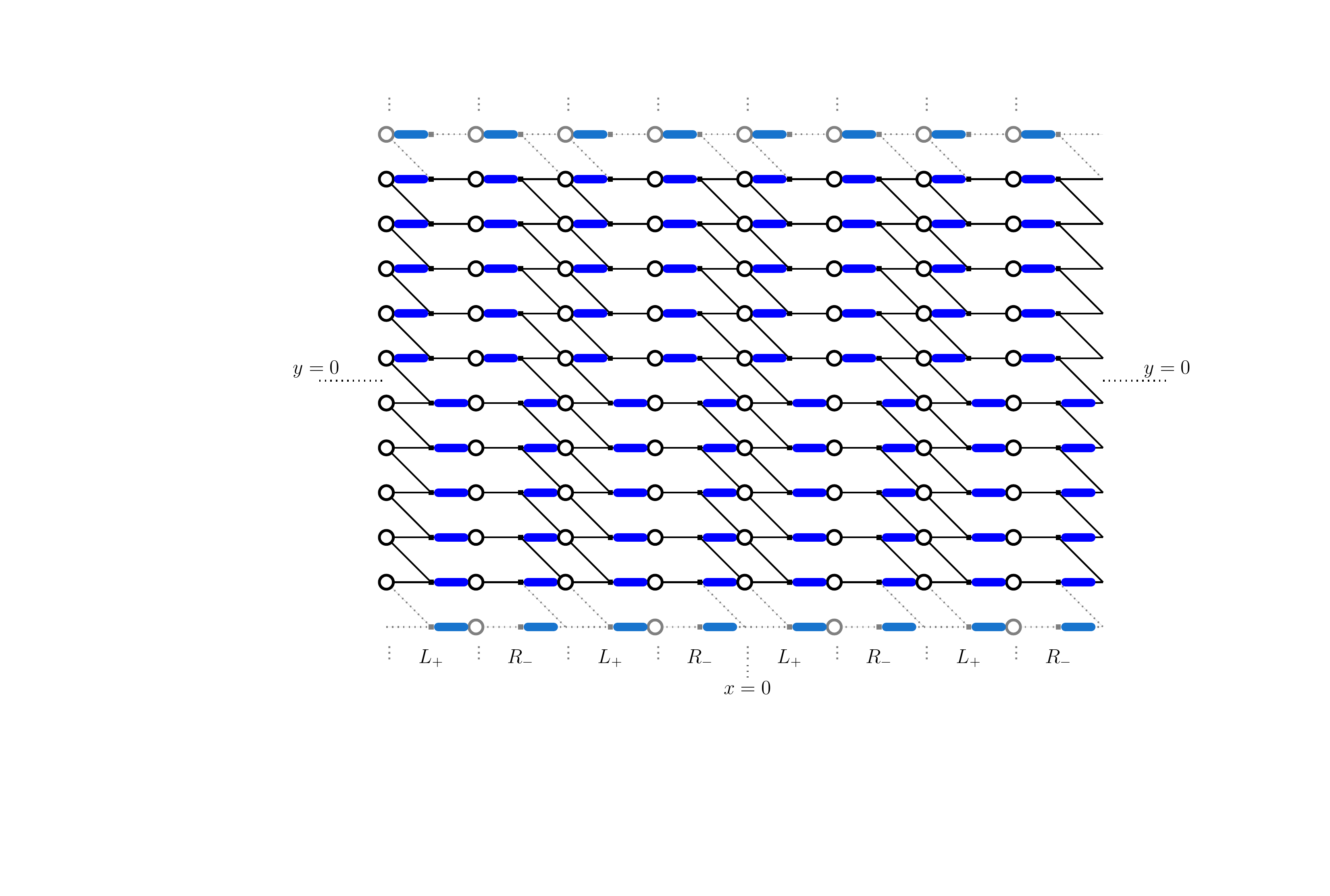}
\caption{A rail yard graph with LR sequence $\underline{a}=(LR)^n$ and sign sequence $\underline{b}=(+-)^n$ for $n=4$, equipped with its fundamental dimer covering.}\label{fig:aztecRYG}
\end{figure}

\begin{figure}
\includegraphics[width=\linewidth]{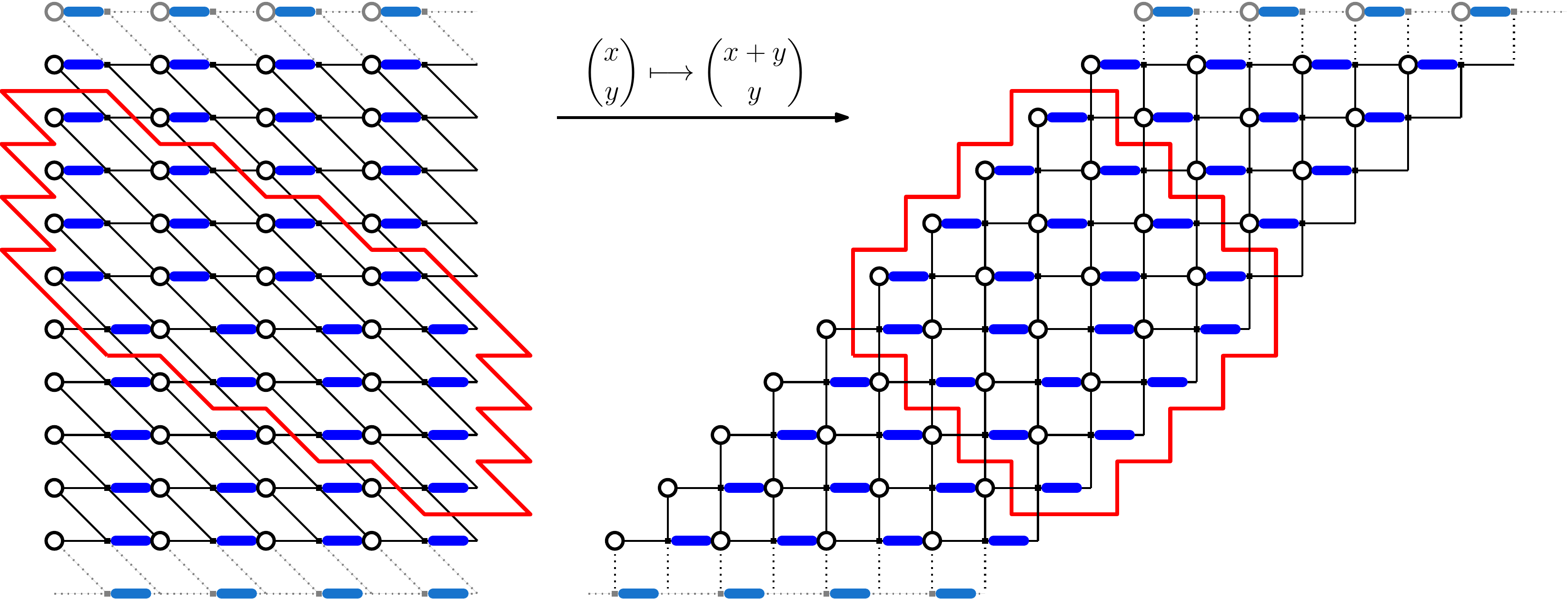}
\caption{Left: The graph obtained by contracting all the vertices of degree $2$ in the RYG of Figure~\ref{fig:aztecRYG}, that is, all the vertices of abscissa congruent to $2$ mod $4$. Right: the image of the graph on the left by the linear transformation $\binom{x}{y} \mapsto \binom{x+y}{y}$. The graph is a portion of a square lattice of mesh $1$. On both sides, the region of the graph where a covering may differ from the fundamental one is represented in red. One recognizes the shape of the Aztec diamond of size $n=4$. Note that the coordinate transformation on uncontracted vertices to go directly from the graph of Figure~\ref{fig:aztecRYG} to the graph on the right is given by $\binom{x}{y} \mapsto \binom{\phi(x)+y}{y}$ where $\phi(x) = x-2\lfloor \frac{x+2}{4}\rfloor$ if $x\geq 0$ and $\phi(-x)=-\phi(x)$.
}\label{fig:aztecRYGscaled} 
\end{figure}

In order to illustrate Theorem~\ref{thm:det_corr} on a concrete
example, we concentrate in this section on the particular case of
domino tilings of the Aztec diamond of size
$n$~\cite{EKLP1992,EKLP1992b}, which can be obtained as a rail yard
graph associated with the sequences $\underline{a}=(LR)^n$ and
$\underline{b}=(+-)^n$, see Figures~\ref{fig:aztecRYG}
and~\ref{fig:aztecRYGscaled}. We will suppose that $\ell=0$ and
$r=2n-1$, and denote by $G_n$ the corresponding rail yard graph. Note
that, though $G_n$ is an infinite graph, its admits only a finite
number ($2^{n(n+1)/2}$) of pure dimer coverings. Those pure dimer
coverings coincide with the fundamental covering outside a finite
region which is isomorphic to the Aztec diamond graph of size $n$,
see~\cite[Section~4.1]{BCC} for a discussion of this phenomenon in the
language of steep tilings. The Aztec diamonds are essentially the only
RYGs having a finite number of pure dimer configurations (this can be
seen from the enumerative results of Section~\ref{sec:enum} and
Remark~\ref{rem:rygsteep}).

Let us now discuss the probability distributions over domino tilings
of the Aztec diamond that we are considering. Recall the
definition~\eqref{eq:xweights} of the weight of a configuration in the
multivariate RYG dimer model.  For a generic sequence
$x_0,\ldots,x_{2n-1}$, we obtain the so-called Stanley weighting
scheme~\cite{Propp:talk,BYYangPhD}, see also \cite[Remark~2]{BCC}. The
partition function reads
\begin{equation}
  \label{eq:Zstan}
  Z(G_n,\underline{x}) = \prod_{\substack{0 \leq i < j \leq 2n-1 \\ \text{$i$ even, $j$ odd}}}
    (1 + x_i x_j)
\end{equation}
which is a polynomial in the $x$'s. By specialization we obtain the
following distributions considered originally in~\cite{EKLP1992}:
\begin{itemize}
\item the uniform distribution, obtained by taking $x_i=1$ for all
  $i$,
\item the \emph{biased} distribution, obtained by taking $x_i=1$ for
  $i$ even and $x_i=\lambda>0$ for $i$ odd (or equivalently
  $x_i=\sqrt{\lambda}$ for all $i$): this corresponds to attaching a
  weight $\lambda$ to each pair of diagonal dimers (which become
  vertical dominos in the Aztec diamond picture),
\item \label{page:qvol} the so-called $q^{\text{vol}}$ distribution, obtained by taking
  $x_i=q^i$ for $i$ odd and $x_i=q^{-i}$ for $i$ even (which is the
  $q$-RYG specialization), and more generally the biased
  $q^{\text{vol}}$ distribution, obtained by taking $x_i=\lambda q^i$
  for $i$ odd and $x_i=q^{-i}$ for $i$ even.
\end{itemize}
Let us mention that our present approach does not seem to apply to the
two-periodic weighting considered in \cite[Section~6]{MR3197656}, nor
to the weightings considered in \cite{MR3228361}.

The study of correlations in domino tilings of the Aztec diamond has
been a popular topic (especially among the members of the ``domino
forum'') and there are many previously known results, published or
unpublished. In the rest of this section, we rederive several of these
results as consequences of our general Theorem~\ref{thm:det_corr}.

\subsection{The biased creation rate and edge-probability generating
  function}
\label{subsec:edgep}

Let us consider the Aztec diamond of size $n$ in the natural
coordinates with the origin at the center. We are interested in the
probability to find a domino of a given type at a given position,
under the biased distribution. Recall that we may distinguish four
types of dominos: north-, south-, west- and
east-going~\cite{CEP1996}. By symmetry it is enough to consider only
one type of domino, and we denote by $P_\lambda(x,y,n)$ the
probability that $(x-1/2,y)$ is the center of a west-going domino in a
biased random tiling of the Aztec diamond of size $n$. Due to parity
constraints, this probability vanishes unless $x$ and $y$ are integers
such that $x+y+n$ is odd. As apparent from Figures~\ref{fig:aztecRYG}
and~\ref{fig:aztecRYGscaled}, west-going dominos correspond in the RYG
setting to diagonal dimers in columns of type $L+$ (which makes them
easier than north-going dominos to deal with, since there can be no
spurrious diagonal dimers outside the ``interesting'' region).

Getting an expression for $P_\lambda(x,y,n)$ (or its analogues for
other domino orientations) ameneable to asymptotic analysis has been
of interest to several people, for instance it is used
in~\cite{CEP1996} as a way to proving the arctic circle theorem and
its generalization to arbitrary $\lambda$. However, this paper states
the required expressions without proofs, and refers instead to a
preprint by Gessel, Ionescu and Propp that has not appeared so
far. The proof for $\lambda=1$ (uniform distribution) can be found in
Helfgott's senior thesis~\cite{Helfgott}. More recently, an expression
for the related generating function (still for $\lambda=1$ only) was
proved by two methods in Du's master thesis~\cite{Du:thesis}, and was
used in~\cite{BP11} as yet another route to the arctic circle theorem
(we note that those two references mention another lost ``DGIP''
preprint). At the suggestion of James Propp, which we thank for
pointing out this gap in the literature, we will explain how
expressions for $P_\lambda(x,y,n)$ (for general $\lambda$) can be
obtained as applications of our Theorem~\ref{thm:det_corr}.

The first step consists in going from dominos to dimers on RYGs. Using
the identification discussed above, a west-going domino centered on
$(x-1/2,y)$ corresponds to a dimer on the edge $(\alpha,\beta)$ with
\begin{equation}
  \label{eq:coordchange}
  \alpha=(2m,y-1/2), \qquad \beta=(2m-1,y+1/2), \qquad m = \frac{n+x-y+1}{2} \in [1..n].
\end{equation}
By Theorem~\ref{thm:det_corr} we immediately deduce the expression
\begin{multline}
  \label{eq:Plamb}
  P_{\lambda}(x,y,n) = \\ \frac{1}{(2i\pi)^2}
  \oint_{C_z}\oint_{C_w}
  \left(\frac{1-w}{1-z}\right)^m\left(\frac{1+\lambda/w}{1+\lambda/z}\right)^{n+1-m}
  \frac{(w/z)^y}{(z-w)(1-w)} dz dw
\end{multline}
where $C_z$ is a positively oriented contour containing $0$ and
$-\lambda$ in its interior, but not~$1$, and $C_w$ is a positively
oriented contour containing $C_z$ in its interior. Simpler expressions
can be obtained for two related quantities: the so-called \emph{biased
  creation rate}
\begin{equation} \label{eq:creationRateL}
  \mathrm{Cr}_\lambda(x,y,n) = \frac{\lambda+1}{\lambda}\left(P_\lambda(x,y,n)
  - P_\lambda(x+1,y,n-1)\right)
\end{equation}
and the \emph{edge-probability generating function}
\begin{equation}
  \label{eq:epgf}
  \Pi_\lambda(u,v,t) = \sum_{x,y,n} P_{\lambda}(x,y,n) u^x v^y t^n.
\end{equation}
Note that, in \cite{CEP1996}, the bias $p$ is related to our $\lambda$
by $p=1/(1+\lambda)$, and the biased creation rate is expressed in
terms of north-going dominos hence the present definition is adapted
to the case of west-going dominos.

Let us first consider the biased creation rate. Taking the difference
$P_\lambda(x,y,n) - P_\lambda(x+1,y,n-1)$ in the double contour
integral~\eqref{eq:Plamb}, the integrand is multiplied by a factor
$1-\frac{1+\lambda/z}{1+\lambda/w}=\frac{\lambda(z-w)}{z(w+\lambda)}$,
leading to a cancellation of the denominator $(z-w)$, hence to the factorization
\begin{equation}\label{eq:twofactorsL}
  \begin{split}
    \mathrm{Cr}_\lambda(x,y,n)=& (1+\lambda) \left(
      \oint_{C_z}\left(1-
        z\right)^{-m}\left(z+\lambda\right)^{-(n+1-m)} z^{n+1-m-y}
      \frac{\mathrm{d}z}{2i\pi z}
    \right) \\
    &\times \left( \oint_{C_w}\left(1-
        w\right)^{m-1}\left(w+\lambda\right)^{n-m} w^{y+m-n}
      \frac{\mathrm{d}w}{2i\pi w} \right).
  \end{split}
\end{equation}
Note that, by Cauchy's residue formula, the second integral is equal
to the coefficient of $w^{n-m-y}$ in $(1-w)^{m-1}(w+\lambda)^{n-m}$,
which is a Krawtchouk polynomial.  The first integral is of a similar
nature, except that the role of the zeros $\{1,-\lambda\}$ and of the
poles $\{0, \infty\}$ in the integrand are exchanged.
This suggests to perform the change of variables $z=\frac{1-u}{1+\lambda^{-1}u}$
 which transforms the first integral in~\eqref{eq:twofactorsL} into:
 \begin{equation}
   (1+\lambda)^{-n} \lambda^{m-1} 
   \oint_{C_u} u^{1-m} (1+\lambda^{-1}u)^{m+y-1} (1-u)^{n-m-y} \frac{-du}{2i\pi u},
\end{equation}
where $C_u$ is a small negatively oriented contour encircling $u=0$
(note that the image of the negatively oriented contour $C_u$ under
the change of variable $z=\frac{1-u}{1+\lambda^{-1}u}$ is a negatively
oriented contour encircling $z=1$, which is homotopic in
$\mathbb{C}\cup\{\infty\}\setminus\{1,-\lambda\}$ to the positively
oriented contour $C_z$). Using Cauchy's residue formula, the second
integral is equal to the coefficient of $u^{m-1}$ in
$(1+\lambda^{-1}u)^{m+y-1} (1-u)^{n-m-y}$.
Redistributing powers of $\lambda$ we finally obtain:
\begin{prop} The biased creation rate reads
  \begin{equation}
    \mathrm{Cr}_\lambda (x, y, n) = \left(\frac{\lambda}{1+\lambda}\right)^{n-1} c_\lambda(A,B,n-1) c_\lambda(B,A,n-1)
    \label{eq:bcr}
  \end{equation}
  where $A = n-m-y = (n-1-x-y)/2$, $B = m-1= (n-1+x-y)/2$ and $c_\lambda(A,B,n)$
  is the coefficient of $z^A$ in $(1-z)^{B}(1+\lambda^{-1}z)^{n-B}$.
\end{prop}
Note that $\frac{\lambda}{1+\lambda}=(p-1)$ and $\lambda^{-1}=\frac{p}{1-p}$, so we recover~\cite[Proposition~23]{CEP1996}, up to the exchange $p\mapsto 1-p$ and the antidiagonal reflection $(x+y,x-y)\mapsto (-x-y,x-y)$ that correspond to the fact that we consider west-going rather than north-going dominos.
\cite[Proposition~2]{CEP1996} also follows, by
taking $\lambda=1$, i.e.\ $p=1/2$.

We now turn to the edge-probability generating function. We need to
multiply \eqref{eq:Plamb} by
$u^x v^y t^n = (tu)^m (t/u)^{n+1-m} (uv)^y/t$ and sum over all
$n \geq 0$, $m\in [1..n]$ and $y \in \mathbb{Z}$. Assuming that
$\lambda<1$, $C_z$ and $C_w$ can be taken as circles with center $0$
and radiuses between $\lambda$ and $1$. For $t$ small enough and $u,v$
of modulus close to $1$, it is possible to interchange the double sum
over $n,m$ and the double contour integral to yield
\begin{equation}
  \label{eq:Pidoub}
  \Pi_\lambda(u,v,t) = \sum_{y \in \mathbb{Z}} \frac{1}{(2i\pi)^2}
  \oint_{C_z}\oint_{C_w} \phi(z,w) (uvw/z)^y \frac{dz}{z} \frac{dw}{w}
\end{equation}
where
\begin{equation}
  \label{eq:phiexp}
  \phi(z,w) =
    \frac{%
      \frac{tu(1-w)}{1-z}
    }{1-\frac{tu(1-w)}{1-z}} \cdot
    \frac{%
      \frac{t(1+\lambda/w)}{u(1+\lambda/z)}
    }{1-\frac{t(1+\lambda/w)}{u(1+\lambda/z)}} \cdot
    \frac{z w}{t(z-w)(1-w)}.
\end{equation}
The sum over all $y$ allows to get rid of one integral. More
precisely, assuming $|uv|<1$ and $|w|$ fixed, we take two different
contours $C_z$ depending on the sign of $y$, namely a circle $C_z^+$
(resp.\ $C_z^-$) of radius slightly larger (resp.\ smaller) than
$|uvw|$ if $y \geq 0$ (resp.\ $y<0$). Splitting the sum over $y$ in
two accordingly, we may interchange each sum with the integral,
resulting in
\begin{equation}
  \label{eq:Pisimp}
  \begin{split}
    \Pi_\lambda(u,v,t) &= \frac{1}{(2i\pi)^2} \oint_{C_w} \left( \left(
        \oint_{C_z^+} - \oint_{C_z^-} \right) \frac{\phi(z,w)}{z-uvw}
      dz \right) \frac{dw}{w} \\
    & = \frac{1}{2i\pi} \oint_{C_w} \phi(uvw,w) \frac{dw}{w}.
  \end{split}
\end{equation}
In the latter integral, the integrand has two poles, but only one of
them falls within $C_w$ for $t$ small and $|uv|$ close to $1$, and we
end up with:
\begin{prop}
  The biased edge-probability generating function reads
  \begin{equation}
    \label{eq:epgffin}
    \Pi_\lambda(u,v,t)=\frac{\lambda t}{(1-t/u)
      \left((1+\lambda)(1+t^2)-t(u+u^{-1})-\lambda t(v+v^{-1})\right)}.
  \end{equation}
\end{prop}
For $\lambda=1$, we recover the expression given in
\cite{Du:thesis,BP11} (again up to the change of variables needed to
pass from west-going to north-going dominos).

\begin{rem}
  The generating function associated with the biased creation rate
  takes a much more symmetric form, namely
  \begin{equation}
    \label{eq:bcrgf}
    \begin{split}
      \sum_{x,y,n} \mathrm{Cr}_{\lambda}(x,y,n) u^x v^y t^n &= \frac{\lambda+1}{\lambda} (1-t/u) \Pi_\lambda(u,v,t) \\
      &= \frac{(\lambda+1)t}{
        \left((1+\lambda)(1+t^2)-t(u+u^{-1})-\lambda
          t(v+v^{-1})\right)},
    \end{split}
  \end{equation}
  and it remains the same for other types of dominos. The
  combinatorial explanation of this phenomenon (and of the meaning of
  the term ``creation rate'') comes from the domino shuffling
  algorithm \cite{EKLP1992b}, which implies that
  $\mathrm{Cr}_{\lambda}(x,y,n)$ is the probability that, in a biased
  random tiling of the Aztec diamond of size $n$, the $2 \times 2$
  square centered at $(x,y)$ is covered by exactly two dominos
  (regardless of their orientation).
\end{rem}

\subsection{The inverse Kasteleyn matrix}
\label{subsec:aztecKasteleyn}

Chhita and Young gave in \cite[Section~5]{MR3197656} a formula for the
inverse Kasteleyn matrix of the Aztec diamond of size $n$, for
the biased $q^{\text{vol}}$ distribution (whose definition is recalled on page~\pageref{page:qvol}). Under this distribution the probability for a tiling is proportional
to $q$ to the number of flips from the fundamental
configuration times $\sqrt{\lambda}$ to the number of vertical dominos.
We now explain how
to relate their formula with the entries $\mathcal{C}_{\alpha,\beta}$.

Let $\alpha=(\alpha^{\mathrm{x}},\alpha^{\mathrm{y}})$ be an even vertex and $\beta=(\beta^{\mathrm{x}},\beta^{\mathrm{y}})$
an odd vertex, which have survived the contraction of edges. The coordinates we
will use are those of the contracted graph, so $\alpha^{\mathrm{x}}\in[0..2n-2]$ and
$\beta^{\mathrm{x}}\in[-1..2n-1]$.

In this particular case, for $k\in[-1..2n-1]$, the function $F_k(z)$ from
Equation~\eqref{eq:def_Fk} becomes:
\begin{equation}
  F_k(z) =
  \frac{1}{%
    \prod_{j=0}^{\lfloor k/2 \rfloor}(1-\sqrt{\lambda}q^{-2j}z)
    \prod_{j=\lfloor (k+1)/2\rfloor}^{n-1}(1+\sqrt{\lambda}q^{2j+1}/z)
  }.
\end{equation}

Chhita and Young use coordinates $(x_1,x_2)$ (resp.~$(y_1,y_2)$) to localize
odd (resp.~even) vertices. In their terminology, they are white and black
respectively. These coordinates correspond to axes that are along the diagonals
of Figure~\ref{fig:aztecRYGscaled} on the right.

The correspondence between the two systems of coordinates is
\begin{equation}
  \begin{cases}
    x_1 &= 1+\beta^{\mathrm{x}}+2\beta^{\mathrm{y}} \\
    x_2 &= 1+\beta^{\mathrm{x}}
  \end{cases},
  \quad
  \begin{cases}
    y_1 &= 1+\alpha^{\mathrm{x}}+2\alpha^{\mathrm{y}} \\
    y_2 &= 1+\alpha^{\mathrm{x}}
  \end{cases}.
\end{equation}

Performing the change of variable
\begin{equation}
  \zeta=-\frac{q^{y_2-2}}{w},\quad \theta = -\frac{q^{y_2-2}}{z},
\end{equation}
in the integral defining
$K^{-1}_{\text{col}}$ in~\cite[Theorem~5.1]{MR3197656} 
in the case when $x_2 \leq y_2$,
one recovers the same factors for the rational fraction in $\zeta$ and $\theta$ as for
$G_{\alpha,\beta}(\zeta,\theta)$, up to possibly numerical multiplicative constants. One
has just to check that the contours enclose the same sets of poles. Under the
change of variables, the contour $\Gamma_{1,q}$ becomes a contour for $\zeta$
enclosing
$-\sqrt{\lambda}q^{y_2},-\sqrt{\lambda}q^{y_2+2},\dots,-\sqrt{\lambda}q^{2n-1}$ (and which may or may not enclose zero, since the original
rational fraction in $w$ is regular at infinity), and the contour $\Gamma_0$ is
mapped to a large contour for $\theta$ containing separating infinity from a domain
containing all the poles and the contour for $\theta$. This contour can be deformed
freely as long as it does not cross the one for $\zeta$ or infinity, because
$1/F_k(\theta)$ is a polynomial, and thus has no poles.

In the case when $x_2 > y_2$, the extra term in~\cite{MR3197656} comes from the
residue at $w=z$, which can be integrated to the double integral, at the cost of
interchanging the nesting of the contours.

\subsection{The arctic circle theorem}
\label{subsec:aztecArctic}

Under the uniform measure, domino tilings of a large Aztec diamond exhibit a spatial phase transition, known
as the arctic circle phenomenon \cite{JPS98,J2005}. Outside the inscribed circle, with probability exponentially close to 1, all
dominos are arranged in a brickwall fashion. This is called the frozen region. Inside the circle, the probability
of each orientation is non degenerate and does not go to 0 or 1. We now explain how to recover
this phenomenon from our formalism. The parameters $q$ and $\lambda$ are now set to 1.

Let $\alpha=(\alpha^{\mathrm{x}},\alpha^{\mathrm{y}})=(2m,y-\frac{1}{2})$ be an
even vertex.
The probability $\rho_\alpha$ that
a dimer connects $\alpha$ with
$\beta=(\beta^{\mathrm{x}},\beta^{\mathrm{y}})=(2m-1,y+\frac{1}{2})$ is
given by~\eqref{eq:Plamb} with $\lambda=1$.
In the scaling limit $n\rightarrow \infty$, $m/n\rightarrow \tau$, $y/n\rightarrow \chi$, this probability becomes
\begin{equation}
 \rho_\alpha=\oint_{C_z}\oint_{C_w} \exp(n(S(z,\tau,\chi)-S(w,\tau,\chi)+o(1)))\frac{{\rm d}z{\rm d}w}{(z-w)(2i\pi)^2}
 \end{equation}
where
\begin{equation}
S(z,\tau,\chi)=-\tau\log(1-z)-(1-z)\log(1+1/z)-\chi\log(z).
\end{equation}

We now proceed as in \cite{OR1} to obtain the asymptotics of this probability. For fixed $(\tau,\chi)$, the function
$z\mapsto S(z,\tau,\chi)$ has two critical points.
\begin{itemize}
 \item 
 If the two critical points are real, the integral goes to 0 or 1 exponentially fast with $n$ by the saddle point method.
 The point $(\tau,\chi)$ is in the frozen region.
 \item If the two critical points are complex conjugate, one can move the contours so that they cross transversally at the complex
 critical points to apply again the saddle point method. By doing so, we pick the contribution of the residue at $z=w$, which gives
 the main contribution of the integral, giving a result strictly between 0 and 1.
\end{itemize}
The transition between those two regimes correspond to the value of $(\tau,\chi)$ for which the two critical points merge. This happens
when the discriminant of the numerator of $\partial S(z,\tau,\chi)/\partial z$ is equal to zero. This gives
\begin{equation}
 (2\tau-1)^2-4(\tau+\chi)(1-\chi-\tau)=0,
\end{equation}
which under the change of variables
\begin{equation}
  \left\{\begin{array}{l} u=2\tau+\chi\\
           v=\chi
         \end{array}\right .
\end{equation}
corresponds to the circle $2(u-1)^2+2v^2=1$ inscribed in the limiting
square of the Aztec diamond, given by $|u-1|+|v|\leq 1$.

\section{Conclusion}
\label{sec:conc}

We have introduced the rail yard graph dimer model, and computed its
partition function and correlation functions. We point out that it
corresponds essentially to the most general Schur process with
nonnegative transition probabilities, see the discussion in
\cite[Sections~1 and 2]{Borodin:dynamics}: we handle an arbitrary
finite number of ``$\alpha$'' and ``$\beta$'' specializations, and any
other specialization can be obtained by taking suitable limits (in
particular, to get the Poissonized Plancherel measure, one shall
consider the ``dilute'' limit of RYGs, namely take a sign sequence of
the form $+^n-^n$, an arbitrary LR sequence, and a constant weight
sequence $z/n$, then let $n \to \infty$).

Many directions can be explored from here. By applying the random
generation algorithms of \cite{BBBCCV14}, we may generate large RYG
dimer configurations, which allows to observe limit shape phenomena as
in the cases of (skew) plane partitions \cite{OR1,OR2} and of the
Aztec diamond, discussed above. RYG seem to allow for an even larger
variety of singular points on limit shapes, and of corresponding
limiting processes, which are currently under investigation.

The appearance of the rational edge-probability generating function
\eqref{eq:epgffin} in the context of the Aztec diamond (when summing over
diamonds of all sizes) raises the question whether such rationality
phenomenon may subsist for other types of RYGs. A natural idea is to
consider RYGs with periodic LR, sign and weight sequences. Preliminary
research indicates the rationality phenomenon occurs only in another
case, namely skew plane partitions of ``staircase'' shape. In other
cases, we obtain an algebraic, but not rational, generating function
(algebraicity being expected from the very nature of our
computations).

In this paper we have obtained the correlations for pure RYG dimer
configurations (by computing vacuum-to-vacuum expectation values of
fermionic operators). Other types of boundary conditions can be
considered, as in~\cite{BCC} where the corresponding partition
functions were computed (and the extension to RYGs is
straightforward). However, adapting the computation of correlation
functions done in the present paper is not so easy, since it requires
an adaptation of Wick's formula. For arbitrary but fixed boundary
conditions, we know from general facts, namely the generalized Wick
theorem \cite{AlZa13} or the Eynard-Mehta theorem \cite{MR2185331},
that correlations will still be of determinantal form, however it is
not clear how to compute explicitly the propagator/determinantal
kernel. Such computation could be done by Petrov \cite{MR3278913} for
some lozenge tilings, and we are looking for other tractable
cases. Also of interest is the case of free boundary conditions (that
is, we sum over all possible boundary states). When only one of the
boundaries of the RYG is free (corresponding to symmetric RYG dimer
configurations), the correlations are known to be Pfaffian
\cite{MR2185331}, and we have found an adaptation of Wick's formula
which would allow for a computation similar to that done in the
present paper, bypassing the use of the (Pfaffian analogue of)
Eynard-Mehta's theorem. Details should appear in a subsequent
publication, see also \cite{MR2864788,MR3091062,Panova14} for related
results. When the two boundaries of the RYG are free, the nature of
the correlations is unknown, even though the partition function can be
computed following the lines of~\cite{BCC}. We believe they should be
the Pfaffian analogues of correlations for RYGs with periodic boundary
conditions, related to the periodic Schur process of~\cite{B2007}.

Finally, a tantalizing question is whether it is possible to consider
``interacting'' deformations of our dimer models. Besides the
directions mentioned in the conclusion of \cite{BCC}, let us mention
that fermionic techniques have been recently used, together with
methods from constructive field theory, to prove rigorous results
about interacting dimer models, see e.g.\ \cite{GMT14} and references
therein. Another intriguing fact is that identities arising from
$z$-measures (which are instances of Schur measures) have found
applications in the context of quantum integrable systems
\cite{KKMST11}.

\section*{Acknowledgments}

We would like to thank Dan Betea, Alexei Borodin, Sunil Chhita,
Philippe Di Francesco, Patrik Ferrari, Jean-Michel Maillet, Richard
Kenyon, Leo Petrov, Senya Shlosman and Mirjana Vuleti\'c, for their
constructive comments and helpful discussions.

\appendix

\section{Commutation of bosonic and fermionic operators}
\label{sec:fermcom}

In this section, we give a self-contained combinatorial proof of the commutation relations between the bosonic and the fermionic operators stated in Proposition~\ref{prop:com_gamma_psi}.

\begin{proof}
Let us first prove \eqref{eq:cgp2}. It is equivalent to prove that for any $k$ and for any $\lambda$,
\begin{equation}
\label{eq:cgp2_bis}
\Gamma_{R+}(x)\psi_k|\lambda\rangle=
\left(\psi_k+x\psi_{k-1}\right)\Gamma_{R+}(x)|\lambda\rangle.
\end{equation}
Define
\begin{equation}
n_k=\# \{j>k, \lambda_j=\bullet\}.
\end{equation}

\emph{Case 1 : $\lambda$ has a white marble in position $k$.}

$\Gamma_{R+}(x)\psi_k|\lambda\rangle$ enumerates the admissible dimer covers of
an elementary rail yard graph of type $(R,+)$ and with right boundary equal to
$\psi_k|\lambda\rangle$. Each dimer cover is specified by the value $\mu$ of the
left boundary.
Here,
\begin{equation}
\psi_k|\lambda\rangle=(-1)^{n_k}
      |\lambda^{(k)}\rangle,
\end{equation}
so the vertex in position $k$ on the right boundary has to be incident to a
certain edge $e_k$ in the dimer cover. Thus
$\Gamma_{R+}(x)\psi_k|\lambda\rangle$ is a sum of two types of terms, the first
(resp. second) type corresponds to $e_k$ horizontal (resp. diagonal).

By the same argument as for the localization of horizontal dimers on a
``double'' column, dimer covers of the first type are in bijection with (and
have the same weight as) dimer covers with right boundary $\lambda$ and with
left boundary $\mu^{(k)}$. Such covers are enumerated by
$\Gamma_{R+}(x)|\lambda\rangle$. To obtain the original left boundary, $\mu$,
from this new cover, we need to apply $\psi_k$ to
$\Gamma_{R+}(x)|\lambda\rangle$. The sign appearing is again $(-1)^{n_k}$,
because the number of black marbles above position $k$ is the same in
$\mu^{(k)}$ as in $\lambda$. So the first term in the sum is equal to
\begin{equation}
\psi_k\Gamma_{R+}(x)|\lambda\rangle.
\end{equation}

By using the bijection used to localize diagonal dimers, and observing that the
weights differ by a factor $x$, we obtain that the second term of the sum is
equal to
\begin{equation}
x\psi_{k-1}\Gamma_{R+}(x)|\lambda\rangle.
\end{equation}

Note that the signs cancel out correctly because the number of black marbles
above position $k-1$ on the left boundary is equal to the number of black
marbles above position $k$ on the right boundary.

Thus we conclude in case 1.

\emph{Case 2 : $\lambda$ has a black marble in position $k$.}

Here, $\psi_k|\lambda\rangle=0$, so the left-hand side of \eqref{eq:cgp2_bis}
vanishes.

$\Gamma_{R+}(x)|\lambda\rangle$  enumerates the admissible dimer covers of an
elementary rail yard graph of type $(R,+)$ and with right boundary equal to
$\lambda$. Each dimer cover is specified by the value $\mu$ of the left
boundary. The vertex in position $k$ on the right boundary has to be incident to
a certain edge $e_k$ in the dimer cover. Thus
$\Gamma_{R+}(x)\psi_k|\lambda\rangle$ is a sum of two types of terms, the first
(resp. second) type corresponds to $e_k$ horizontal (resp. diagonal).

If $\mu$ is a term of the first type, $\mu$ has a black marble in position $k$,
thus $\psi_k+x\psi_{k-1}|\mu\rangle=x\psi_{k-1}|\mu\rangle$ and
\begin{equation}
x\psi_{k-1}|\mu\rangle=
\begin{cases}
x(-1)^{n_k+1}\mu^{(k-1)} & \text{if $\mu$ has a white marble in position $k-1$} \\
0 & \text{otherwise}.
\end{cases}
\end{equation}

The sign is due to the fact that the number of black marbles in $\mu$ above the
position $k-1$ is equal to the number of black marbles in $\lambda$ above the
position $k$, to which we must add the black marble in $\lambda$ in position
$k$.

If $\mu$ is a term of the second type, $\mu$ has a black marble in position
$k-1$, thus $\psi_k+x\psi_{k-1}|\mu\rangle=\psi_k|\mu\rangle$ and

\begin{equation}
\psi_k|\mu\rangle=
\begin{cases}
(-1)^{n_k}\mu^{(k)} & \text{if $\mu$ has a white marble in position $k$} \\
0 & \text{otherwise}.
\end{cases}
\end{equation}

So all the nonzero terms of
$\left(\psi_k+x\psi_{k-1}\right)\Gamma_{R+}(x)|\lambda\rangle$ have black
marbles in positions $k-1$ and $k$, and each term appears twice, with the same
weight (because if $\mu$ is a term of the second type in
$\Gamma_{R+}(x)|\lambda\rangle$, it already carries a factor $x$ coming from the
diagonal dimer $e_k$) and with opposite sign.

Thus the right-hand side of \eqref{eq:cgp2_bis} also vanishes.

This concludes the proof of \eqref{eq:cgp2}. Formula \eqref{eq:cgp4} is proved
similarly, replacing $\psi_{k-1}$ by $\psi_{k+1}$ and $z$ by $\frac{1}{z}$.

We noted that $\Gamma_{R+}(x)$ was conjugated to $\Gamma_{L+}(x)$ \textit{via}
$\omega$. Observe now that $\psi_{l}^*$ is conjugated to $\psi_l$ \textit{via}
$\omega$, up to a sign $(-1)^{s_l}$ verifying $(-1)^{s_{l+1}}=-(-1)^{s_l}$
($s_l$ is defined by an equation analogous to \eqref{eq:def_sl}). This enables
us to deduce formulas \eqref{eq:cgp5} and \eqref{eq:cgp7}.

Using that $\Gamma_{R+}(x)$, $\Gamma_{L+}(x)$ and $\psi(z)$ are respectively
adjoint of  $\Gamma_{R-}(x)$, $\Gamma_{L-}(x)$ and $\psi(1/z)$, we deduce the
last four formulas.

\end{proof}

\section{Wick's formula}
\label{sec:wick}

In this section, we provide a proof of the identity
\begin{multline}
  \langle \emptyset | \mathcal{T}\left(\Psi(\alpha_1),\Psi^*(\beta_1),\ldots,
    \Psi(\alpha_s),\Psi^*(\beta_s)\right) | \emptyset \rangle = \\
  \det_{1 \leq i,j \leq s}  \langle \emptyset |
  \mathcal{T}\left(\Psi(\alpha_i),\Psi^*(\beta_j)\right) | \emptyset \rangle
  \label{eq:wickth}  
\end{multline}
used in \eqref{eq:wickapply}.

Let $\Pi$ be the set of partitions of $\{1,\ldots,2s\}$ into unordered
pairs.  An element $\tau\in\Pi$ can be canonically written
$\tau=\lbrace\lbrace i_1,j_1\rbrace,\ldots,\lbrace
i_s,j_s\rbrace\rbrace$
with $i_1<i_2<\cdots<i_s$ and $i_t<j_t$ for all $t$, and viewed as a
permutation 
\begin{equation}
 \tau=\begin{pmatrix} 1 & 2 & \cdots & 2s-1 & 2s \\ 
 i_1 & j_1 & \cdots & i_s & j_s \end{pmatrix},
\end{equation}
which allows to define its sign $\epsilon(\tau)$. Recall that the
\emph{Pfaffian} of a $2s\times 2s$ antisymmetric matrix
$\mathcal{A}=(\mathcal{A}_{ij})_{1\leq i,j\leq 2s}$ is defined as
\begin{equation}
  \label{eq:pfaf}
  \Pf(\mathcal{A})=\sum_{\tau\in\Pi} \epsilon(\tau)
  \prod_{t=1}^s \mathcal{A}_{\tau(2t-1),\tau(2t)}.
\end{equation}
Let $F$ be the space of (countably infinite) linear combinations of
$\psi_k$'s and $\psi^*_k$'s.

\begin{prop}[Wick's formula]
  \label{prop:wick}
  For $X_1,\ldots,X_{2s}$ elements of $F$, we have
  \begin{equation}
    \label{eq:wickfor}
    \langle\emptyset| X_1 \cdots X_{2s} |\emptyset \rangle = \Pf(\mathcal{A})
  \end{equation}
  where $\mathcal{A}$ is the $2s\times2s$ antisymmetric matrix such
  that $\mathcal{A}_{ij}=\langle\emptyset|X_i X_j|\emptyset\rangle$
  for $i<j$.
\end{prop}

\begin{proof}
  Let $F^+$ (resp.\ $F^-$) be the vector space spanned by the $\psi_k$
  with $k>0$ and the $\psi^*_k$ with $k<0$ (resp.\ the $\psi_k$ with
  $k<0$ and the $\psi^*_k$ with $k>0$), so that $F = F^+ \oplus F^-$.
  For $X \in F$, we denote by $X^+$ and $X^-$ its projections on these
  two subspaces: observe that $\langle\emptyset|X^+=0$ and
  $X^-|\emptyset\rangle=0$.

  Let $X_1,\ldots,X_{2s}$ be elements of $F$. Observe that, by
  proposition \ref{prop:anticom}, $\{X_i^-,X_j^+\}$ is a scalar for
  any $i<j$, and
  \begin{equation}
    \label{eq:crochet_accolade}
    \begin{split}
      \langle\emptyset|X_iX_j|\emptyset\rangle
      &=\langle\emptyset|X_i^+X_j^++X_i^+X_j^-+X_i^-X_j^-+X_i^-X_j^+|\emptyset\rangle \\
      &=\langle\emptyset|X_i^+X_j^++X_i^+X_j^-+X_i^-X_j^-+\{X_i^-,X_j^+\}-X_j^+X_i^-|\emptyset\rangle \\
      &=\{X_i^-,X_j^+\}.
    \end{split}
 \end{equation}
 We proceed similarly to compute the left-hand side of
 \eqref{eq:wickfor}: we write each of the $X_i$ as $X_i=X_i^++X_i^-$,
 expand the product and get a sum of $2^{2s}$ terms.  In each of these
 terms, we move all the $X_i^-$ to the right. To get a nonzero
 contribution, we must pair up each $X_i^-$ with one $X_j^+$ such that
 $i<j$, and multiply all the anticommutators obtained in this
 fashion. So each nonzero contribution is equal, up to a sign, to
 \begin{equation}
   \prod_{t=1}^s{\{X^-_{\tau(2t-1)},X^+_{\tau(2t)}\}}
 \end{equation}
 where $\tau \in \Pi$ is some partition into pairs of
 $\{1,\ldots,2s\}$. Furthermore, the sign of this contribution
 corresponds to the number of swaps needed to bring each
 $X^-_{\tau(2t-1)}$ exactly to the left of $X^+_{\tau(2t)}$, i.e.\ it
 is the signature of $\tau$. This yields exactly the right-hand side
 of \eqref{eq:pfaf} with
 $\mathcal{A}_{ij}=\{X_i^-,X_j^+\}=\langle\emptyset|X_i X_j|\emptyset\rangle$.
\end{proof}

To prove the wanted identity \eqref{eq:wickth}, we set
$Y_{2t-1}=\Psi(\alpha_t)$ and $Y_{2t}=\Psi^*(\beta_t)$,
$t=1,\ldots,s$. By the definition of natural ordering, there is a
permutation $\sigma$ such that
\begin{equation}
  \mathcal{T}\left(Y_1,\ldots,Y_{2s}\right) = \epsilon(\sigma)
  Y_{\sigma(1)} \cdots Y_{\sigma(2s)}.
\end{equation}
By Proposition~\ref{prop:wick} and since
$Y_{\sigma(i)} Y_{\sigma(j)} =
\mathcal{T}(Y_{\sigma(i)},Y_{\sigma(j)})$ for any $i<j$, we have
\begin{equation}
  \begin{split}
    \langle \emptyset | \mathcal{T}\left(Y_1,\ldots,Y_{2s}\right) |
    \emptyset \rangle &= \epsilon(\sigma) \Pf_{1 \leq i,j \leq 2s}
    \langle \emptyset | \mathcal{T}(Y_{\sigma(i)}, Y_{\sigma(j)}) |
    \emptyset \rangle \\ &= \Pf_{1 \leq i,j \leq 2s} \langle \emptyset
    | \mathcal{T}(Y_{i}, Y_{j}) | \emptyset \rangle
  \end{split}
\end{equation}
(when we simultaneously permute the rows and columns of an
antisymmetric matrix, the Pfaffian is multiplied by the sign of the
permutation). Finally, we observe that
$\langle \emptyset | \mathcal{T}(Y_{i}, Y_{j}) | \emptyset \rangle$ is
nonzero if and only if $i,j$ have different parities, which allows to
rewrite the Pfaffian as the wanted determinant without sign.

\bibliographystyle{halpha}
\bibliography{dimerstat}

\end{document}